\documentclass[sigconf]{acmart}

\usepackage{booktabs} 
\usepackage{amsfonts, amssymb}
\usepackage{mathtools}
\usepackage{amsmath}
\usepackage{subcaption}
\usepackage[show]{chato-notes}

\usepackage{tikz}
\usetikzlibrary{arrows.meta}

\usepackage[latin1]{inputenc}
\usepackage{amssymb, verbatim}
\usepackage{amsmath}
\usepackage{mathrsfs}
\usepackage{graphicx}
\usepackage{balance}  
\usepackage[show]{chato-notes}
\usepackage{adjustbox}

\usepackage{siunitx}

\usepackage{xspace}
\usepackage{float}

\usepackage{algorithm,algpseudocode}
\usepackage{algorithmicx}
\algdef{SE}[DOWHILE]{Do}{doWhile}{\algorithmicdo}[1]{\algorithmicwhile\ #1}%

\usepackage{tikz}
\usepackage{amsmath}
\usetikzlibrary{decorations.pathmorphing, positioning}

\newcommand{\yslant}{0.20}
\newcommand{\xslant}{-0.6}

\newcommand     {\Cset}    {{\Bbb C}}
\newcommand     {\Rset}    {{\Bbb R}}

\newcommand     {\Id}   {\mathcal{I}}

\newtheorem{thm}{Theorem}[section]

\def\real{\mathbb{R}}

\newcommand	{\MRF}{TaCMM\xspace}
\newcommand{\MULTIJ}{MultiJaccard\xspace}

\newcommand{\NORM}{\propto}

\newcommand     {\rr}    {{\bf r}}
\renewcommand     {\AA}    {{\mathcal A}}
\newcommand     {\LL}    {{k}}
\newcommand     {\Zset}    {{\Bbb Z}}

\setcopyright{none}
\renewcommand\footnotetextcopyrightpermission[1]{} 
\begin{document}
\title{A New Framework for Centrality Measures in Multiplex Networks}


\author{Carlo Spatocco}
\affiliation{%
  \institution{Dip. Matematica - La Sapienza}
  \city{Rome} 
  \state{Italy} 
}
\email{spatocco@mat.uniroma1.it}

\author{Giovanni Stilo}
\affiliation{%
  \institution{Dip. Informatica - La Sapienza}
  \city{Rome} 
  \state{Italy} 
}
\email{stilo@di.uniroma1.it}

\author{Carlotta Domeniconi}
\affiliation{%
  \institution{George Mason University}
  \city{Fairfax} 
  \state{USA} 
}
\email{cdomenic@gmu.edu}

\author{Alessandro D'Andrea}
\affiliation{%
  \institution{Dip. Matematica - La Sapienza}
  \city{Rome} 
  \state{Italy} 
}
\email{dandrea@mat.uniroma1.it}

\renewcommand{\shortauthors}{Spatocco et al.}

\begin{abstract}

Any kind of transportation system, from trains, to buses and flights, can be modeled as networks. In biology, networks capture the complex interplay between phenotypes and genotypes. More recently, people and organizations heavily interact with one another using several media (e.g. social media platforms, e-Mail, instant text and voice messages), giving rise to  correlated communication networks.



The non-trivial  structure of  such complex systems makes the analysis of their collective behavior a challenge. The problem is even more difficult when the information is distributed across networks (e.g., communication networks in different media); in this case, it becomes impossible to have a complete, or even partial picture, if situations are analyzed separately within each network due to sparsity.


A multiplex network is well-suited to model the complexity of this kind of systems by preserving the semantics associated with each network. Centrality measures are fundamental for the identification of key players, but existing approaches are typically designed to capture a predefined aspect of the system, ignoring or merging the semantics of the individual layers.

To overcome the aforementioned limitations, we present a Framework for Tailoring Centrality Measures in Multiplex networks (\MRF), which offers a flexible methodology that encompasses and generalizes previous approaches. The strength of \MRF is to enable the encoding of specific dependencies between the subnets of multiplex networks to define semantic-aware centrality measures.


We develop a theoretically sound iterative method, based on Perron-Frobenius theory, designed to be effective also in high-sparsity conditions. We formally and experimentally prove its convergence for ranking computation. We provide a thorough investigation of our methodology against existing techniques using different types of subnets in multiplex networks. The results clearly show the power and flexibility of the proposed framework.

\end{abstract}

%
%



\maketitle


\section{Introduction}\label{sec:introduction}
Networks are present in all aspects of our world and constitute the backbone of many utilities, such as gas, electricity, and water. Networks permeate any kind of transportation systems (train, buses, air-flights, and naval). In biology, networks model the complex interplay between phenotypes and genotypes. 
People, states, and organizations heavily interact with one another on a daily basis using several types of media,  giving rise to many correlated communication networks. As a result, the participating entities become heavily interconnected through several social platforms (Facebook, Twitter, Instagram, etc.), telephone, short text and voice messages (e.g., e-Mail, Messenger, SMS, and WhatsApp).


The non-trivial  structure of  such complex systems makes the analysis of their collective behavior a challenge. The problem is even more difficult when the information is distributed across networks (e.g., communication networks in different media); in this case, it becomes impossible to have a complete, or even partial picture, if situations are analyzed separately within each network due to sparsity.

A multiplex network is well-suited to model the complexity of this kind of systems by preserving the semantics associated with each network. Centrality measures are fundamental for the identification of key players, but existing approaches are typically designed to capture a predefined aspect of the system, ignoring or merging the semantics of the individual layers.
 
As an example, consider a scenario with three companies; in reality, each one of them could be fragmented in sub-companies, and different kinds of media could be used to communicate.
The first one is the \textit{Wayne Enterprises, Inc.}, \textbf{\textit{WayneCorp}}, which owns mining companies, oil drilling and refineries, and also has business in technology, biotechnology, pharmaceuticals, and health-care. 
The second one is \textbf{\textit{LexCorp}}, an international conglomerate with interests in utilities, waste management, industrial manufacturing, chemicals, bio-engineering, weapons, pharmaceuticals, oil, and more. 
The last one is the consulting company \textbf{\textit{E-Corp.}}, which actively provides solutions in severals fields to the \textit{WayneCorp} and the \textit{LexCorp}.

Suppose there is a suspicion of leaking confidential information involving various sectors from the \textit{WayneCorp} to the \textit{LexCorp}. The main suspects are among the \textit{E-Corp} employees.
Our aim is to identify the employees responsible for the information leakage.
To this end,  we need to find the users that have the following characteristics: (1) The employee works for both the \textit{WayneCorp} and for the \textit{E-Corp.}; (2) The position within the \textit{WayneCorp} network allows the employee to collect  various information; this role must be enforced by the capacity of collecting information in high quantity and quality; (3) The employee must also have the capacity to spread the collected information  directly or indirectly to employees  in higher positions in the \textit{LexCorp} company.

A natural way of modeling the semantics of the described scenario is to use a multiplex network, where each layer collects the interactions that involve the employees of a given company. The first layer collects the interactions among the employees of the \textit{WayneCorp}, and the interactions among the employees of the \textit{WayneCorp} and of the \textit{E-Corp}. The second layer collects the interactions among the employees  of the \textit{LexCorp}, and the interactions of the employees of the \textit{LexCorp} and of the \textit{E-Corp}.
The last layer, contains all the interactions between the employees of the \textit{E-Corp}.

Even if our model of the real-world scenario is accurate and preserves the wanted 
semantics, no existing approach in the literature can assist us in achieving the aims stated above.  The problem is even more difficult because in our scenario the information is distributed across networks.
The sparsity of the networks is a big challenge. In our example, the three networks have a certain degree of local density, but many nodes are disconnected (e.g., in the first layer, the nodes corresponding to the \textit{LexCorp} employees). One possible solution is to collapse all the layers in one, but in this case we'd loose semantics, and retaining semantics is paramount for our goal. Another solution is to  apply a specific (centrality) measure to each layer of the multiplex network, and then combine the results. But this trivial solution  does not consider all the interactions simultaneously. Furthermore, the standard centrality measures are not well-suited for our scope, and no existing method enables a flexible environment to  define different centrality measures.

To overcome the aforementioned limitations, we present a Framework for Tailoring Centrality Measures in Multiplex networks (\MRF), which offers a flexible methodology that encompasses and generalizes previous approaches. The strength of \MRF is to enable the encoding of specific dependencies between the subnets of multiplex networks to define semantic-aware centrality measures.

We develop a theoretically sound iterative method, based on Perron-Frobenius theory, designed to be effective also in high-sparsity conditions. We formally and experimentally prove its convergence for ranking computation. We provide a thorough investigation of our methodology against existing techniques using different types of subnets in multiplex networks. The results clearly show the power and flexibility of the proposed framework.

The rest of the paper is organized as follows. Section \ref{related} discusses related work. In Section \ref{sec:TaCMM}, we formally define the problem, and discuss our proposed framework, the iterative solution, its convergence and implementation.  Section \ref{sec:exp} presents our extensive empirical evaluation and analysis. Section \ref{conclusion} concludes the paper.

\section{Related Work}
\label{related}
The concept of centrality in networks has always been fundamental for  understanding  the system being modeled. The first centrality measure introduced in the literature,  \textit{degree centrality}, simply assigns to each node the number of  incident edges. Bavelas  \cite{bavelas1948mathematical} introduced \textit{closeness centrality} for undirected graphs, defined as the inverse of the sum of the length of all the paths from a node to every other node in the network. Lin \cite{lin1976foundations} modified the concept of closeness centrality for directed networks, by taking into account the unreachable pairs of nodes. Based on the notion of shortest paths, Anthonisse \cite{anthonisse1971rush} first, and  Freeman  \cite{freeman1977set} later, introduced the concept of \textit{betweenness centrality} as the probability that a random shortest path goes through a node. 

Another approach to the problem of measuring centrality is focusing on the properties of the adjacency matrix, rather then studying the network combinatorial properties. Several studies are based on the principal eigenvector of the adjacency matrix \cite{seeley1949net}, \cite{wei1952algebraic}, and \cite{claude1966theorie}. These measures typically find the left principal eigenvector of a matrix, and this is possible, as in our case, thanks to the Perron-Frobenius theorem that ensures the existence and uniqueness of this vector under the hypotheses of irreducibility and aperiodicity.  
The main idea  is to model the network as a Markov chain and associate probabilities to nodes. Many well-known centrality measures built upon this idea, e.g. \textit{PageRank} \cite{page1999pagerank}, \textit{HITS} \cite{Kleinberg:1999}, and \textit{SALSA} \cite{lempel2001salsa}.
Both \textit{PageRank} and \textit{HITS} were conceived to rank web pages. \textit{SALSA} is similar to \textit{HITS}, but uses normalized matrices.  \textit{PageRank} assigns a single score value to each node, representing the  probability of finding a net surfer on a certain web page after an infinite number of clicks on links, starting from a random web page. \textit{HITS} and \textit{SALSA}, instead, give two score values. They use the adjacency matrix and its transpose, and  iteratively compute the score vectors. The returned scores capture the importance of a web page (its \textit{authority}), measured by considering the web pages that point to it, and the  \textit{hubness}, measured by considering the pages linked to by the page under consideration.

Multiplex PageRank \cite{halu2013multiplex}  extends PageRank  to a multiplex network. The main idea is to let  the centrality of a node in a layer be influenced by  the centrality of the same node in different layers. Starting from a fixed sequence of layers, the method computes the classic PageRank over the first layer, and uses it to calculate the  PageRank scores for the second one, and so on. The result is a score vector defined as the limit on the number of iterations. The formulation also considers some parameters to express the nature of the interactions between the layers. Varying those parameters, different kinds of Multiplex PageRank can be obtained: additive, multiplicative, combined, and neutral. 
This solution provides a single score vector and consequently causes loss of information. In contrast, our method produces several score vectors, relative to the layers, hence more information is captured by the rankings. It's up to the user whether and how the resulting rankings should be aggregated.

The exponential growth of data in the last decades has also increased the complexity of the resulting networked data. As a result, more sophisticated structures, e.g. multilayer or multiplex networks, capable of capturing more than one relation between nodes have been developed. This also raised the need for new centrality measures. Many approaches  extended the classic centrality measures to the case of multilayer or multiplex networks. In particular, an elegant tensor-based formulation was introduced to extend all the standard centrality measures for multilayer networks \cite{de2013mathematical,dedomenico2015ranking}.

A tensor $M=M_{j\beta}^{i\alpha}$ can be thought as a four-dimensional matrix with  positive real entries representing the weight of the edge between the node $i$ in  layer $\alpha$ and the node $j$ in the layer $\beta$. Using this new language, all the classic centrality measures are extended to multilayer graphs.
Score vectors (tensors) are computed considering the whole structure of the system. This is achieved by considering all the incident edges of a fixed node, including those across layers. The structure of the used tensor differentiates the  measures being calculated. Once the tensor structure is fixed, the resulting measures are also fixed. Our approach is fundamentally different.  We enable a framework that encompasses and generalizes all previous proposed measures, and allows the user to choose the most suitable setting based on the problem under investigation. 




\sloppy
\section{A Framework for Tailoring Centrality Measures}\label{sec:TaCMM}
In this section, we present \MRF, a Framework for Tailoring Centrality Measures in Multiplex networks. \MRF offers a flexible methodology that encompasses and generalizes previous approaches.
We start this section by formally describing the problem setting.
We then show how  \MRF can encode specific dependencies between  subnets of multiplex networks to define semantic-aware centrality measures.
We  present an iterative method to compute the rankings, based on Perron-Frobenius theory, and designed to be effective also in high-sparsity conditions. Furthermore, we present a proof of convergence, and finally we give  implementation details of the \MRF framework.

\subsection{Problem Definition}
\label{sec:problem_def}
Let $\mathcal{M}$ be a multiplex network (or multigraph) composed of $\mathcal{L}$ direct graphs $G_\ell=(V,E_\ell)$, with $0 \le \ell < \mathcal{L}$ (see Figure \ref{fig:multiplex}). Each graph $G_\ell$  contains the same set of vertices $V$, where  $|V|=n$. $E_\ell$ is the set of direct edges of graph $G_\ell$.
Let $A_\ell$ be the $n \times n$ out-link matrix of $G_\ell$, where $A_\ell (ij)=w(i,j) > 0$ if a direct edge from node $i$ to node $j$ exists, and $0$ otherwise.
 $A^T_\ell$ is the transpose of $A_\ell$ and represents the in-link matrix of $G_\ell$. 
 
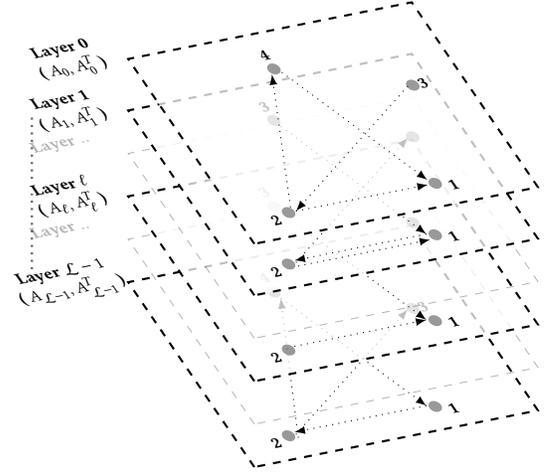
\begin{figure}[ht]
\centering
\begin{tikzpicture}[scale=0.65,on grid]
	\begin{scope}[
		yshift=-130,
		every node/.append style={yslant=\yslant,xslant=\xslant},
		yslant=\yslant,xslant=\xslant
	]
		\fill[white,fill opacity=.75] (1,1.5) rectangle (6.8,5.8); 
		\draw[black, dashed,thick] (1,1.5) rectangle (6.8,5.8); 
		
		\draw[fill=black,opacity=0.4, thick] 
			(5,2) node(111){} circle (.1) 
			(2,2) circle (.1)
			(5.8,4.1) circle (.1)
			(3.5,5) circle (.1);

		\draw[-latex, thin,dotted]
			(3.55,4.85) to (4.85,2.05);
			
		\draw[-latex, thin,dotted]
			(2.15,1.92) to (3.55,4.85); 			
			
		\draw[-latex, thin,dotted]
			(2.15,2) to (5.7,4.1);
				
		\draw[-latex, thin,dotted]
			(4.85,2.05) to (2.15,2.05);

		\fill[black]
			(-0.1,6.0) node[above, scale=.7] {\textbf{Layer $\mathcal{L}-1$}}
			(-0.1,5.5) node[above, scale=.7] {\textbf{(} $A_{\mathcal{L}-1}$, $A^T_{\mathcal{L}-1}$ \textbf{)}}
			(5.1,1.9) node[right,scale=.7]{\bf 1}
			(1.9,1.9) node[left,scale=.7]{\bf 2}
			(5.8,4.1)  node[right,scale=.7]{\bf 3}
			(3.5,5.1) node[above,scale=.7]{\bf 4}; 
			
	\end{scope} 
	
	\begin{scope}[
		yshift=-105,
		every node/.append style={yslant=\yslant,xslant=\xslant},
		yslant=\yslant,xslant=\xslant	]	
		\draw[black, dashed, very thin,opacity=0.3] (1,1.5) rectangle (6.8,5.8);  	
		\fill[black]
			(-0.1,6.0) node[above, scale=.7,opacity=0.3] {\textbf{Layer $..$}};
	\end{scope}

	\begin{scope}[
		yshift=-80,
		every node/.append style={yslant=\yslant,xslant=\xslant},
		yslant=\yslant,xslant=\xslant	]	
		\fill[white,fill opacity=.75] (1,1.5) rectangle (6.8,5.8); 
		\draw[black, dashed, thick] (1,1.5) rectangle (6.8,5.8);  
			
		\draw[fill=black,opacity=0.4, thick]  
			(5,2) node(111){} circle (.1) 
			(2,2) circle (.1)
			(5.8,4.1) circle (.1)
			(3.5,5) circle (.1);
			
			\draw[-latex, thin,dotted]
			(3.55,4.85) to (4.85,2.05);
			\draw[-latex, thin,dotted]
			(2.15,2.05) to (4.85,2.05);
		
		\fill[black]
			(-0.1,6.0) node[above, scale=.7] {\textbf{Layer $\ell$}}
			(-0.1,5.5) node[above, scale=.7] {\textbf{(}  $A_\ell$, $A_\ell^T$ \textbf{)} }
			(5.1,1.9) node[right,scale=.7]{\bf 1}
			(1.9,1.9) node[left,scale=.7]{\bf 2}
			(3.5,5.1) node[above,scale=.7]{\bf 3};
	\end{scope}
	
	\begin{scope}[
		yshift=-55,
		every node/.append style={yslant=\yslant,xslant=\xslant},
		yslant=\yslant,xslant=\xslant	]	
		\draw[black, dashed, very thin,opacity=0.3] (1,1.5) rectangle (6.8,5.8);  	
		
		\fill[black]
			(-0.1,6.0) node[above, scale=.7,opacity=0.3] {\textbf{Layer $..$}};
		
	\end{scope}

	\begin{scope}[
		yshift=-30,
		every node/.append style={yslant=\yslant,xslant=\xslant},
		yslant=\yslant,xslant=\xslant
	] 
		\fill[white,fill opacity=.75] (1,1.5) rectangle (6.8,5.8); 
		\draw[black, dashed, thick] (1,1.5) rectangle (6.8,5.8); 
		
		\draw[fill=black,opacity=0.4, thick]  
			(5,2) node(111){} circle (.1) 
			(2,2) circle (.1)
			(5.8,4.1) circle (.1)
			(3.5,5) circle (.1);
		 
		\draw[-latex, thin,dotted]
			(3.55,4.85) to (4.85,2.05);
		\draw[-latex, thin,dotted]
			(2.15,1.92) to (4.85,1.92);

		\draw[-latex, thin,dotted]
			(2.15,2) to (5.7,4.1);
				
		\draw[-latex, thin,dotted]
			(4.85,2.05) to (2.15,2.05); 
		\fill[black]
			(-0.1,6.0) node[above, scale=.7] {\textbf{Layer 1}}	
			(-0.1,5.5) node[above, scale=.7] {\textbf{(}  $A_1$, $A_1^T$ \textbf{)} }
			(5.1,1.9) node[right,scale=.7]{\bf 1}
			(1.9,1.9) node[left,scale=.7]{\bf 2}
			(3.5,5.1) node[above,scale=.7]{\bf 3};	
	\end{scope}
	
	\draw[black, dotted, thick, decoration={ segment length=1mm, amplitude=0.6mm},opacity=0.65] (-4.45,1) to (-4.45, 2.55);
	\draw[black, dotted, thick, decoration={ segment length=1mm, amplitude=0.6mm},opacity=0.65] (-4.45,2.85) to (-4.45, 4.2);

	\begin{scope}[
		yshift=0,
		every node/.append style={yslant=\yslant,xslant=\xslant},
		yslant=\yslant,xslant=\xslant
	]
		\fill[white,fill opacity=.75] (1,1.5) rectangle (6.8,5.8); 
		\draw[black, dashed, thick] (1,1.5) rectangle (6.8,5.8); 
		
		\draw [fill=black,opacity=0.4, thick]
			(5,2) node(111){} circle (.1)
			(2,2) circle (.1)
			(5.8,4.1) circle (.1)
			(3.5,5) circle (.1);

		\draw[-latex, thin,dotted]
			(3.6,4.9) to (4.9,2.1);
		\draw[-latex, thin,dotted]
			(2.15,2) to (4.85,2);
			
		\draw[-latex, thin,dotted]
			(5.7,4.1) to (2.15,2);
			
		\draw[-latex, thin,dotted]
			(2.1,2.1) to (3.4,4.9);

		\fill[black]
			(-0.1,6.0) node[above, scale=.7] {\textbf{Layer 0}}
			(-0.1,5.5) node[above, scale=.7] {\textbf{(} $A_0$, $A_0^T$ \textbf{)} }
			(5.1,1.9) node[right,scale=.7]{\bf 1}
			(1.9,1.9) node[left,scale=.7]{\bf 2}
			(5.8,4.1)  node[right,scale=.7]{\bf 3}
			(3.5,5.1) node[above,scale=.7]{\bf 4}; 
			
	\end{scope} 
\end{tikzpicture}
\caption{Multiplex Network with $\mathcal{L}$ layers
    \label{fig:multiplex}}
\end{figure}

We want to design a  formal method which allows to define flexible cyclic rankings over a multiplex network.
Then our goal is to iteratively compute those cyclic ranking of normalized\footnote{Here and below, ``$\propto \bf{r}$'' means ``$= {\bf{r}}/||{\bf{r}}||_1$''.}  score vectors ${\bf{r}}_s \in [0,1]^n$.

As an example, let us consider the case of two layers, that is $\mathcal{L} = 2$.  If we want to compute HITS-like rankings, we need to formalize 
the dependencies as follows:\\
$${\bf{r}}^{t}_0 \NORM A_0^T{\bf{r}}^{t}_1 $$
$${\bf{r}}^{t}_1 \NORM A_0{\bf{r}}^{t}_2$$
$${\bf{r}}^{t}_2 \NORM A_1^T{\bf{r}}^{t}_3$$
$${\bf{r}}^{t}_3 \NORM A_1{\bf{r}}^{t-1}_0$$
and each score vector ${\bf{r}}_s^t$ can be computed as follows:
$${\bf{r}}^{t}_0 \NORM A_0^TA_0A_1^TA_1 {\bf{r}}^{t-1}_0$$
$${\bf{r}}^{t}_1 \NORM A_0A_1^TA_1A_0^T {\bf{r}}^{t-1}_1$$
$${\bf{r}}^{t}_2 \NORM A_1^TA_1A_0^TA_0 {\bf{r}}^{t-1}_2$$
$${\bf{r}}^{t}_3 \NORM A_1A_0^TA_0A_1^T {\bf{r}}^{t-1}_3$$

\subsection{Semantic-aware centrality measures}\label{subsec:tailored}
Here we introduce a graphical representation (given in Figure \ref{fig:frmwGraph1} (c)) that captures the cyclic dependencies among the semantic-aware HITS-like rankings ${\bf{r}}_0$, ${\bf{r}}_1$, ${\bf{r}}_2$, and ${\bf{r}}_3$, as computed above. The resulting graph is a ring with rankings as nodes, and direct edges labeled with the matrices that specify the dependencies among pairs of rankings.

We call this ring a {\textit{configuration}}. A configuration $c$ induces an order among the matrices, as defined by the dependencies. The configuration in Figure \ref{fig:frmwGraph1} (c) induces the ordered sequence $A_0 A_1^T A_1 A_0^T$, where we adopt the convention of starting the sequence from the lowest indexed matrix. Since the dependencies are cyclic, 
Figure \ref{fig:frmwGraph1} (c) equally represents the additional three {\textit{equivalent}} sequences -- $A_1^T A_1 A_0^T A_0$, $A_1 A_0^T A_0 A_1^T$, and $A_0^T A_0 A_1^T A_1$ -- obtained by shifting the initial sequence by $h$ positions, where $0 \le h < |c|$.
Thus, a configuration identifies an {\textit{equivalence class}} of cyclic ordered sequences, and we use the lowest index order convention to select the representative one. For the configuration in Figure \ref{fig:frmwGraph1} (c), the corresponding equivalent class is $\{ A_0 A_1^T A_1 A_0^T,\quad A_1^T A_1 A_0^T A_0, \quad A_1 A_0^T A_0 A_1^T,\quad A_0^T A_0 A_1^T A_1  \}$. Using the \textit{shift} $h=3$ yields the HITS-like formula $A_0^T A_0 A_1^T A_1$. Henceforth, we refer to ${\mbox{shift}}_h$ as the set of members of any configuration that are shifted by $h$ position.

\begin{figure}[ht]
\centering
\begin{tikzpicture}
\begin{scope}[every node/.style={circle,thick,draw}]
    \node (A)[minimum size=0.6cm] at (4.5,2.5) {$r_0$};
    \node (B)[minimum size=0.6cm] at (7.5,2.5) {$r_1$};
    \node (C)[minimum size=0.6cm] at (7.5,0) {$r_2$};
    \node (D)[minimum size=0.6cm] at (4.5,0) {$r_3$}; 
    \node[draw=none] at (6,-1.2) {$(c)$};
 \end{scope}

\begin{scope}[>={Stealth[black]},
              every node/.style={fill=white,circle},
              every edge/.style={draw=black, thick}]
    \path [->] (B) edge[bend right=0] node {$A_0^T$} (A);
    \path [->] (C) edge[bend right=0] node {$A_0$} (B);
    \path [->] (D) edge[bend right=0] node {$A_1^T$} (C);
    \path [->] (A) edge[bend right=0] node {$A_1$} (D);
    \end{scope}

\begin{scope}[every node/.style={circle,thick,draw}]
    \node (E)[minimum size=0.6cm] at (0,0) {$a$};
    \node (F)[minimum size=0.6cm] at (3,0) {$h$};
     \node[draw=none] at (1.5,-1.2) {$(b)$};
    \end{scope}
\begin{scope}[>={Stealth[black]},
              every node/.style={fill=white,circle},
              every edge/.style={draw=black, thick}]
    \path [->] (F) edge[bend right=30] node {$A_0^T$} (E);
    \path [->] (E) edge[bend right=30] node {$A_0$} (F);    
\end{scope}

\begin{scope}[every node/.style={circle,thick,draw}]
    \node (G)[minimum size=0.6cm] at (1.5,2.5) {$p_r$};
     \node[draw=none] at (1.5,1.6) {$(a)$};
      \end{scope}
\begin{scope}[>={Stealth[black]},
              every node/.style={fill=white,circle},
              every edge/.style={draw=black, thick}]
    \path [->] (G) edge[out=30,in=100,looseness=6] node {$M$} (G);   
\end{scope}

\end{tikzpicture}
\caption{(a) PageRank; (b) HITS; and (c) Multiplex HITS-like rankings in \MRF Framework
    \label{fig:frmwGraph1}}
\end{figure}
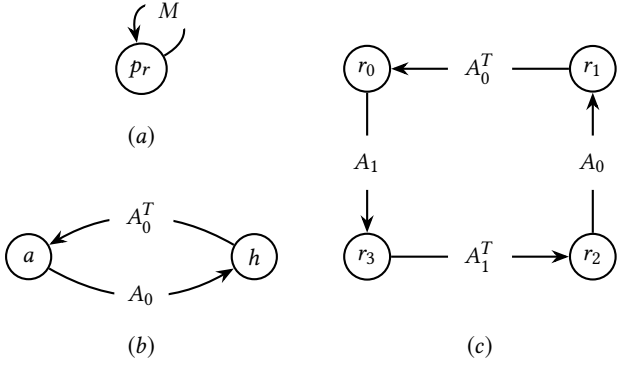

We want to generalize the above setting to enable the computation of {\textit{all}} configurations, and associated rankings, involving any subset of out-link and in-link matrices, and in any possible order. To this end, let us consider the set of all out-link and in-link matrices associated to a multiplex network with $\mathcal{L}$ layers:
$\mathcal{A} = \{ A_0, A_0^T,  A_1, A_1^T,  \dots, A_{\mathcal{L}-1}, A_{\mathcal{L}-1}^T \}$.
As discussed above, a {\textit{configuration}} is an equivalence class of non-empty ordered sequences of elements  of $\mathcal{A}$. Each matrix may occur more than once.
Given $|\mathcal{A}|=2\mathcal{L}$, the corresponding total number of  configurations (without repetitions) is:

\begin{equation}
 \sum_{k=1}^{2\mathcal{L}} \binom{2\mathcal{L}}{k} (k-1)! 
\end{equation}

Figure \ref{fig:all_conf}  lists all possible configurations of four matrices, where \textit{configuration} $(f)$ identifies the one depicted in Figure \ref{fig:frmwGraph1} (c).

\begin{center}
\begin{figure}[h!]
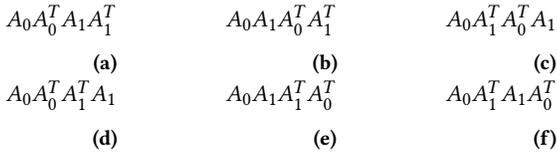

  \begin{subfigure}[b]{.31\linewidth}
    $A_0 A_0^T A_1 A_1^T$  
    \caption{}
    \label{fig:conf_a}
  \end{subfigure}
  \hspace*{\fill}
    \begin{subfigure}[b]{.31\linewidth}    
      $A_0  A_1 A_0^T A_1^T$   
    \caption{}
    \label{fig:conf_b}
  \end{subfigure}
  \hspace*{\fill}
    \begin{subfigure}[b]{.31\linewidth}    
      $A_0  A_1^T A_0^T A_1  $  
    \caption{}
    \label{fig:conf_c}
  \end{subfigure}
  
     \begin{subfigure}[b]{.31\linewidth}
    $A_0 A_0^T A_1^T A_1 $  
    \caption{}
    \label{fig:conf_d}
  \end{subfigure}
  \hspace*{\fill}
    \begin{subfigure}[b]{.31\linewidth}    
      $A_0  A_1 A_1^T A_0^T $   
    \caption{}
    \label{fig:conf_e}
  \end{subfigure}
  \hspace*{\fill}
    \begin{subfigure}[b]{.31\linewidth}    
      $A_0  A_1^T A_1 A_0^T   $  
    \caption{}
    \label{fig:conf_f}
  \end{subfigure}
  
    \caption{All possible configurations (without repetitions) for a multiplex network of two layers
    \label{fig:all_conf}}
\end{figure}
\end{center}

Our definition of configurations allows repetitions of matrices, and therefore an infinite number of total configurations.
We observe that this kind of formalization is very flexible and can be applied  to any kind of centrality measures, which are defined by a cyclic dependency.
For example, it's also possible to formalize and compute two classic centrality measures, such as PageRank and HITS, by using the \MRF Framework, as shown in Figures \ref{fig:frmwGraph1}(a) and \ref{fig:frmwGraph1}(b).


\subsection{Iterative Method with Perturbation}\label{subsec:iterative-method}
Our setup will be slightly more general than that described in Section \ref{sec:problem_def}.
Using a given configuration $c$ and the set of all the adjacency matrices $\AA$ (whose entries are nonnegative real numbers)  of a multiplex $\mathcal{M}$, it possible to select the ordered set of size $k$ composed by the $n\times n$ matrices $M_s,\, s \in \Zset_\LL$.
We may get back to the examples in Section \ref{sec:problem_def} where  $\LL = 2{\mathcal L}$ and $c$ is the configuration reported in Figure \ref{fig:all_conf}(f). 


We aim to find score vectors $\rr_s$, all with nonnegative entries, such that
\begin{equation}\label{scorevectors}
\rr_s = u_s M_s \rr_{s+1}, \quad \mbox{ for all } s\in \Zset_\LL.
\end{equation}
In general, each graph $G_\ell$ will be sparsely connected, but we require that their superposition, i.e., the graph
$$G = \left(V, \bigcup_{0 \le \ell < \mathcal{L}} E_\ell\right),$$
satisfies all irreducibility and aperiodicity assumptions that are required in Perron-Frobenius theory. Since the description is invariant under cyclic permutation of indices, we need to determine, without loss of generality, only the value of $\rr_0$, which satisfies
$$\rr_0 = u_0 u_1 \dots u_{\LL-1} M_0 M_1 \dots M_{\LL-1} \rr_0,$$
i.e., an eigenvector, relative to the eigenvalue $\lambda = u_0 u_1 \dots u_{\LL-1}$, of the composition matrix $M = M_0 M_1 \dots M_{\LL-1}$, whose entries are certainly nonnegative. Using Perron-Frobenius theory, this is usually achieved by choosing any given nonnegative entry vector $\rr_0^0 \in [0,1]^n$ and iterating
$$\rr_0^{t+1} = \frac{M \rr_0^t}{\parallel M \rr_0^t \parallel_1}, \quad t>0,$$
thus obtaining a sequence of vectors of norm $1$ that converges to the (normalized) principal eigenvector of $M$.
However, it is easy to find instances in which the product $M$ fails to satisfy irreducibility and aperiodicity. For instance, using the multiplex network corresponding to Figure \ref{fig:ringGraph}, one obtains $M = A_0^T A_0 A_1^T A_1 = 0$. This shows that the principal eigenvalue of $M$ may fail to be simple, and that $M \rr$ may happen to vanish even when $\rr$ is a vector with non-negative entries.

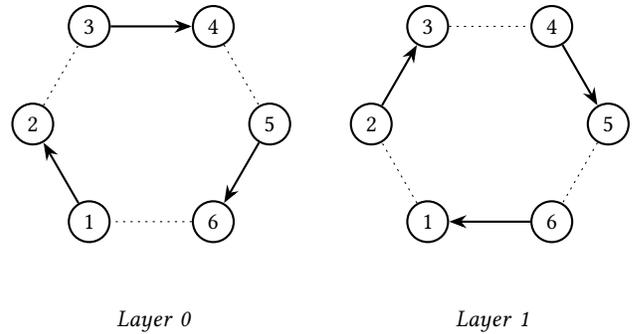
\begin{figure}[h]
\centering

\begin{tikzpicture}[scale=0.75,on grid]
\begin{scope}[every node/.style={circle,thick,draw}]				
    \node (A)[minimum size=0.2cm] at (2.5,1.767949192) 	{$1$};
    \node (B)[minimum size=0.2cm] at (1.5,3.5) 			{$2$};
    \node (C)[minimum size=0.2cm] at (2.5,5.232050808) 	{$3$};
    \node (D)[minimum size=0.2cm] at (4.7,5.232050808) 	{$4$}; 
    \node (E)[minimum size=0.2cm] at (5.7,3.5) 			{$5$}; 
    \node (F)[minimum size=0.2cm] at (4.7,1.767949192) 	{$6$}; 
    \node[draw=none] at (3.65,0) {$Layer$ \textit{0}};
    
    \node (A1)[minimum size=0.2cm] at (8.5,1.767949192) 	{$1$};
    \node (B1)[minimum size=0.2cm] at (7.5,3.5) 			{$2$};
    \node (C1)[minimum size=0.2cm] at (8.5,5.232050808) 	{$3$};
    \node (D1)[minimum size=0.2cm] at (10.7,5.232050808) 	{$4$}; 
    \node (E1)[minimum size=0.2cm] at (11.7,3.5) 			{$5$}; 
    \node (F1)[minimum size=0.2cm] at (10.7,1.767949192) 	{$6$}; 
    \node[draw=none] at (9.65,0) {$Layer$ \textit{1}};

 \end{scope}

\begin{scope}[>={Stealth[black]},
              every node/.style={fill=white,circle},
              every edge/.style={draw=black, thick}]
    \path [->] (A) edge[bend right=0]  			 (B);
    \path [-] (B) edge[bend right=0,dotted,thin] (C);
    \path [->] (C) edge[bend right=0]  			 (D);
    \path [-] (D) edge[bend right=0,dotted,thin] (E);
    \path [->] (E) edge[bend right=0]  			 (F);
    \path [-] (F) edge[bend right=0,dotted,thin] (A);
    
    \path [-] (A1) edge[bend right=0,dotted,thin]  			 (B1);
    \path [->] (B1) edge[bend right=0] (C1);
    \path [-] (C1) edge[bend right=0,dotted,thin]  			 (D1);
    \path [->] (D1) edge[bend right=0] (E1);
    \path [-] (E1) edge[bend right=0,dotted,thin]  			 (F1);
    \path [->] (F1) edge[bend right=0] (A1);
    
    \end{scope}

\end{tikzpicture}
\caption{Ring of six nodes split between a two-layer multiplex network.
    \label{fig:ringGraph}}
\end{figure}

Thus, in principle, there may exist several different nontrivial choices for score vectors; furthermore, it might be infeasible to find them by iterating (and normalizing) the action of $M$ on a given positive coefficient vector.

The approach we present employs the perturbed composition $M(\tau) = (M_0 + \tau \Id)(M_1 + \tau \Id) \dots (M_{\LL-1} + \tau \Id)$, where $\Id$ is the $n \times n$ identity matrix, which has nonnegative entries and satisfies irreducibility and aperiodicity for each strictly positive real choice of $\tau$. The matrix $M(\tau)$ has a unique positive real principal eigenvalue, which can be easily showed to depend continuously (and even analytically) on $\tau \in \Rset_+$. When $\tau$ approaches $0$, $M(\tau)$ tends to $M$ and the principal eigenvalues $\lambda(\tau) \in \Rset_+$ of $M(\tau)$ converge to some real nonnegative eigenvalue $\lambda(0) := \lim_{\tau \to 0^+} \lambda(\tau)$ of $M$. As Perron-Frobenius theory cannot be applied to $M$, the eigenvalue $\lambda(0)$ may fail to be simple. However, by continuity, its norm bounds from above the norm of all other eigenvalues.

When $\tau\in \Rset$ is strictly positive, we may find a unique normalized principal eigenvector $\rr_0(\tau)$ of $M$, with nonnegative coefficients. In next section, we prove that the limit $\rr_0(0) := \lim_{\tau \to 0^+} \rr_0(\tau)$ exists, so that $\rr_0(0)$ is an eigenvector of $M$ relative to the real dominant eigenvalue $\lambda(0)$.

One important fact to stress is that from
$$\rr_s(\tau) = u_s(\tau) M_s(\tau) \rr_{s+1}(\tau), \quad s\in \Zset_\LL,$$
which is the perturbed version of \eqref{scorevectors}, follows that the $\rr_s(\tau)$ may be recovered from knowledge of  $\rr_0(\tau)$ by inductively 	setting
$$\rr_{s-1}(\tau) = M_{s-1}(\tau) \rr_s(\tau)\big/ \parallel M_{s-1}(\tau) \rr_s(\tau)\parallel_1.$$
This fails to hold in general when we take $\lim_{\tau\to 0^+}$, as $u_s(\tau)$ may tend to $0$; nevertheless, $\rr_s(0)$ stays proportional to $M_s \rr_{s+1}(0)$, which may however vanish, for all $s$.

Let us discuss now practical implementations of the above strategy. For each positive value of $\tau$, and any nonzero choice of $\rr_0^0(\tau) \in [0,1]^n$, it is possible to iteratively run a sequence of 
\begin{equation}\label{iteration}
\rr_0^t(\tau) = M(\tau) \rr_0^{t-1}(\tau) \big/ \parallel M(\tau) \rr_0^{t-1}(\tau)\parallel_1, \quad t>0,
\end{equation}
which will approach $\rr_0(\tau)$ to the desired precision. We may then compute $\rr_0(\tau)$ for smaller and smaller positive real values of $\tau$ until the desired convergence to $\rr_0(0)$ is achieved. In practice, we will run a finite number $\delta$ of iterations of \eqref{iteration} for a given value of $\tau$, then halve the value of $\tau$ and run $\delta$ more iterations of \eqref{iteration}, until the desired precision is achieved. Here $\delta$ must be fine-tuned with the geometry of the problem, which depends on the distance of the principal eigenvalue of $M(\tau)$ from the other eigenvalues as a function of $\tau$.

\subsubsection{Proof}\label{subsubsec:proof}
We keep the same setting as from last section, so that 
$M_s$ are $n\times n$ matrices with non-negative entries, and we consider the product $M(\tau) = (M_0 +\tau\Id)\cdot(M_1 +\tau\Id)\cdot \ldots \cdot (M_{\LL-1} + \tau\Id)$ which satisfies irreducibility and aperiodicity for each positive real choice of $\tau$.

We denote by $v(\tau), \tau \in \Rset_+,$ the principal eigenvector of $M(\tau)$ normalized so that its (nonnegative) entries sum to $1$. Recall that the  spectral projector $P(\tau)$ associated to the principal eigenvalue is a matrix with positive coefficients so that $P(\tau)(1, 1, \dots, 1)$ is a positive multiple of $v(\tau)$.

\begin{thm}
The analytic function $\Rset_+ \ni \tau \mapsto v(\tau) \in [0,1]^n$ extends with continuity to $\tau = 0$.
\end{thm}
\begin{proof}
The matrix $M(\tau)$ depends polynomially on $\tau$, hence it makes sense for complex values of $\tau$, and results from \cite[Sect. 1]{kato1980} apply. We argue that up to replacing $M(\tau)$ with $M(\tau^p)$ for a suitable choice of $p$, its spectral projectors $P(\tau)$ are meromorphic functions in an opportunely small neighbourhood of $0 \in \Cset$, having $0$ as only possible singularity.

The eigenvector $V(\tau) = P(\tau)(1, 1, \dots, 1)$ then depends meromorphically on $\tau\in \Cset$, and its entries add to a meromorphic function of $\tau$ which is certainly non-zero, hence non-trivial, as it is strictly positive on the positive real half-line. After dividing $V(\tau)$ by this function, we obtain a meromorphic function $F(\tau)$ which restricts to $v(\tau^p)$ for positive real values of $\tau$. However, all entries of $v(\tau^p)$ are positive and bounded by $1$ on any positive real neighbourhood of $\tau = 0$, so that $F(\tau)$ cannot have a pole in $\tau=0$. We conclude that $F(\tau) = v(\tau^p)$ extends analytically, hence continuously to a complex neighbourhood of $0$. 
This proves the statement, as $\tau \mapsto v(\tau)$ is obtained by composing the analytic function $F$ with $\tau \mapsto \tau^{1/p}$.
\end{proof}

The case $p>1$ is exceptional, and when using actual data gathered from real networks one may assume it never occurs.

\subsection{Implementation details}\label{subsec:impl}
Here we present a detailed description of the implementation of the methodology.

Let  $\tau \in \real, 0 < \tau < 1$, be the perturbation factor; $\mathcal{A}$ the set of all matrices of multiplex $\mathcal{M}$, and  $c \in \mathcal{C}$ one of the possible configurations.

Given the matrix $M_s$, we generate a perturbed version $M_s(\tau) = M_s + \tau\Id$. We then compute the ordered product $M$ of the sequence of matrices $\{\pi_c(\mathcal{A})\}$ selected by the configuration $c$, using the projection $\pi$. We iteratively compute the first rank $r_0$ and decrease $\tau$ by half, each time  the $L^1$-\textit{Norm} of the difference between the last two computed ranks is equal to zero, i.e. $\parallel {\bf {r}}_0^{t} - {\bf {r}}_0^{t+1}\parallel_1 = 0$.
The perturbed matrices $M_s(\tau)$ are updated before proceeding with the  next iteration.

The elements of  ${\bf {r}}_0^{0}$ are all initialized to the same value $\frac{1}{|V|}$ as shown in Algorithm \ref{alg:main}.
The method stops when the stationary point is reached. Using the desired formulation is possible to propagate (propagateScores() in the pseudocode) the computed rank ${\bf {r}}_0$ to the other ranks ${\bf {r}}_s$.

The iterative equations of the HITS-like example with $\mathcal{L} = 2$ are as follows:
$${\bf {r}}^{t}_0 \NORM A_0(\tau)^T A_0(\tau) A_1(\tau)^T A_1(\tau)   {\bf {r}}^{t-1}_0$$
$${\bf {r}}^{t}_1 \NORM A_0(\tau) A_1(\tau)^T A_1(\tau)A_0(\tau)^T   {\bf {r}}^{t-1}_1$$
$${\bf {r}}^{t}_2 \NORM A_1(\tau)^T A_1(\tau)A_0(\tau)^T A_0(\tau)  {\bf {r}}^{t-1}_2$$
$${\bf {r}}^{t}_3 \NORM A_1(\tau)A_0(\tau)^T A_0(\tau) A_1(\tau)^T  {\bf {r}}^{t-1}_3$$

Algorithm \ref{alg:main} presents the implementation to compute  ${\bf {r}}_0$ score.

\begin{algorithm}
  \caption{Iterative Method with Perturbation}
  \label{alg:main}
  \begin{algorithmic}[1]
  \Require The set $\mathcal{A}$ of all matrices of $\mathcal{M}$; a configuration $c \in \mathcal{C}$ and $\tau^0$.
  \Ensure  Set ${\bf {r}}_0^0 := \frac{1}{|V|}$, $t := 0$, $lst := 0$
 \Do
 	\State $ lst := t$
    \State $M :=  \prod_{0\le s < |\pi_c(\mathcal{A})|} \quad( M_s +\tau \Id)$
    \Do
     \State $\rr_0^{t+1} :\NORM  M \cdot  {\bf {r}}_0^{t}$
     \State $t := t+1$
    \doWhile{$ ( \parallel {\bf {r}}_0^{t} - {\bf {r}}_0^{t+1}\parallel_1 \neq 0 )$}
    \State $\tau := \tau/2$     
  \doWhile{$ ( \parallel {\bf {r}}_0^{t} - {\bf {r}}_0^{lst}\parallel_1 \neq 0 )$}
 
  \State \Return $propagateScores({\bf {r}}_0^t,\mathcal{A},c)$;
  \end{algorithmic}
\end{algorithm}


\section{Experiments}\label{sec:exp}
In this section, we present several experiments, designed to analyze the consistency and the behavior of the \MRF framework under different inputs and using different \textit{configurations}.

We performed four sets of experiments to investigate the following issues:
\begin{enumerate}
\item Understanding the relationship between the rankings computed by the \MRF framework (run using proper \textit{configurations} ), and those produced by   methods known in the literature;
\item Investigating the impact of different \textit{configurations} on the computed rankings; 
\item Empirical analysis of convergence speed of the proposed method;
\item Comparing the theoretical computational complexity of  the \MRF framework to that of known methods  in the literature.
\end{enumerate}

\subsection{Experimental environment}
To perform consistent and effective tests, we must be able to work in a controlled environment.  To this end, we  generated synthetic multiplexes and chose graph types that are well suited to our purpose.  In the following, we also define the measures we used to compare the results obtained in different experiments and from various methods. 

\subsubsection{Creation of synthetic multiplexes}\label{exp:env:syn-mult}
We designed a multiplex generator to create synthetic and controlled multiplexes. The generator starts from a given graph $Gen$, and creates a multiplex version of it, $\mathcal{M}$, where each layer is a modified version of $Gen$, and the degree of changes can be modulated.
In particular, the generator assigns each edge of $Gen$ to a layer of the multiplex network $\mathcal{M}$, according to a probability $p$.
Let ${p_0,...,p_\ell}$ be the probabilities assigned to the layers $\ell \in L(M)$ (the probabilities can be independent, or they can sum to one). Then:\\
$$\forall_{e\in E_{Gen}},\quad\forall_{ \ell \in L(\mathcal{M})}: e\cup E_\ell \quad according \quad to \quad p_\ell$$

\subsubsection{Graph Generators}\label{exp:env:gen} 
To perform our experiments we also need graph generators (to initialize $Gen$) that expose certain characteristics. We experimented with two different types of graphs: the first has no particular structure, and the second  contains communities. In both cases, we  chose a network model that belongs to the family of \textit{random graphs}\cite{newman2010networks}.

The first chosen model is  \textit{Erdos-Renyi $G(N, p)$ }\cite{erdos1960evolution}. The constructed random graph  is  obtained by connecting $N$ nodes randomly. Each of the $\binom{N}{2}$ possible edges is included in the graph with probability $p$ (called \textit{edge probability}) independently of the other edges. The resulting graphs are  characterized by an absence of sub-structure, as shown  in Figure \ref{fig:graph_types}.

The second graph generator is the  \textit{Stochastic Block Model, (SBM)}. The natural number $n$  of  the $SBM(n,\lambda,P)$ function  corresponds to the number of nodes, $\lambda = (\lambda_1,\ldots,\lambda_r)$ is a partition of $n$, representing the \textit{communities}, and $P$ is a matrix $r\times r$ where $r = |\lambda|$. The graph is built by taking $n$ vertices and by partitioning them according to $\lambda$. An edge between a vertex $v \in \lambda_i$ and a vertex $u \in \lambda_j$  is added with probability $P_{ij}$. SBM has the capability of mimicking communities, as shown in Figure \ref{fig:graph_types}.

\begin{center}
\begin{figure}[h!]
\includegraphics[width=0.98\linewidth]{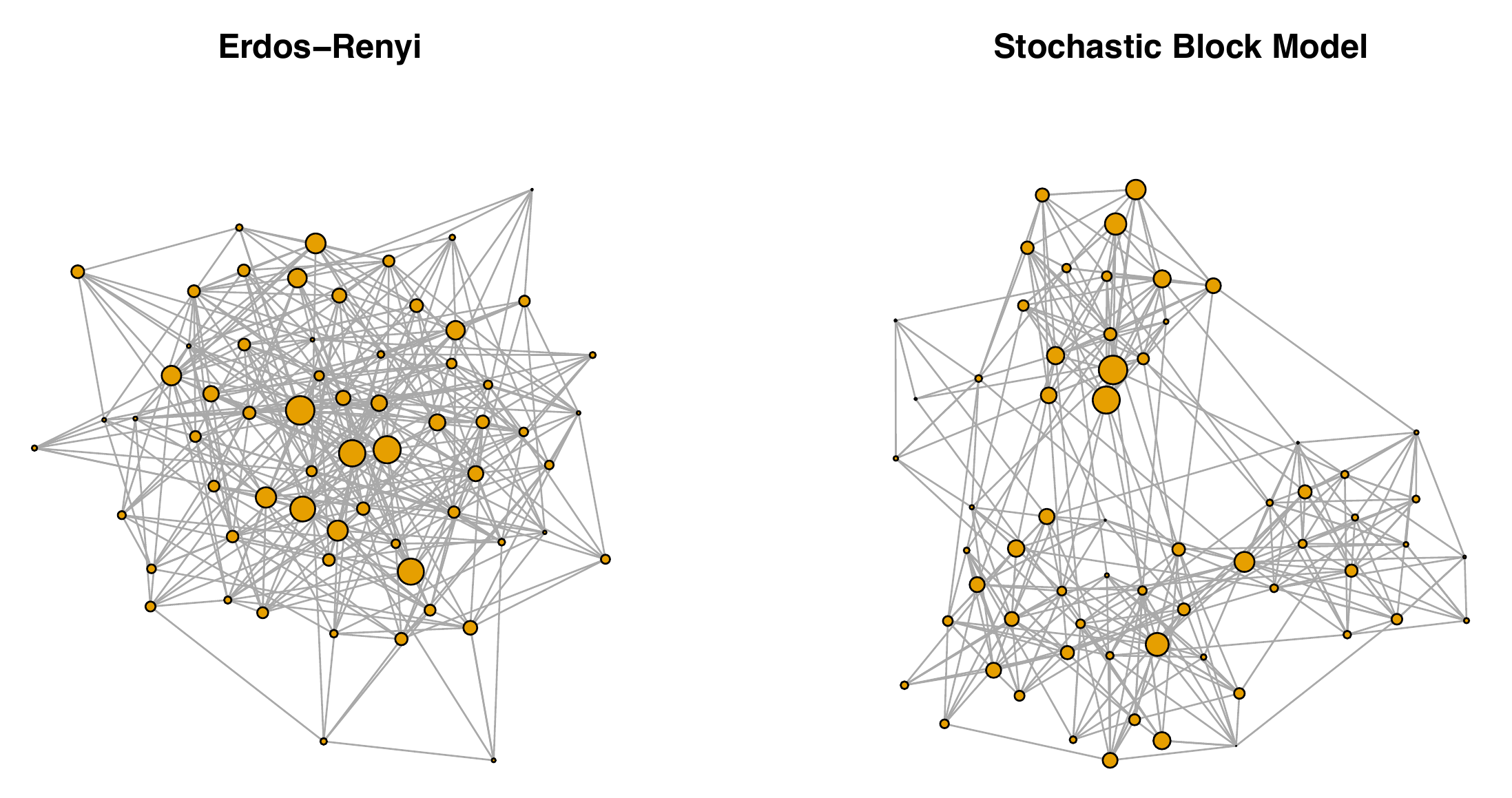}
 \caption{Examples of graphs generated using the Erdos-Renyi Model (left) and the Stochastic Block Model (right)
              \label{fig:graph_types}}
\end{figure}
\end{center}

\subsubsection{Adopted Measures}\label{exp:env:measures}
Here we discuss the measures we have used to compare results across different experiments and methods.

As proposed by Vigna in \cite{vigna2015weighted}, the weighted Kendal tau correlation measure $\tau_w$ is the best choice to compare two  rankings. This coefficient is a variation of the Kendall's tau measure, which was created to overcome problems caused by the presence of ties in rankings. In particular, we use the weighted Kendal tau to compare rankings computed by the \MRF framework, as well as to compare rankings produced by our framework and those produced by other methods. We observe that, even if several rankings are produced by a given method, it's possible to compute (without any loss of precision, or generality) the weighted tau coefficient of two rankings, each obtained by concatenating all the rankings produced by each method being compared.

To measure the amount of overlap between pairs of layers in a generated multiplex network, we define the $\MULTIJ$ coefficient, an extended version of the Jaccard coefficient:
\begin{equation}
 \MULTIJ(\mathcal{M}) = \frac{\sum_{\ell \in L(\mathcal{M})} \sum_{\ell' \in L(\mathcal{M}), \ell \ne \ell'} \frac{E_\ell \cap E_{\ell'}}{E_\ell \cup E_{\ell'}}}{|L(\mathcal{M})| \cdot (|L(\mathcal{M})| - 1)}
\end{equation}


Finally, for each experiment, we  provide the  \textit{confidence interval}. The interval is displayed in the plots as a grey area around the line that describes the average. It's computed according to the following formula:
\begin{equation}
I_c = (\bar{x}-t^*\frac{s}{\sqrt[]{N}},\bar{x}+t^*\frac{s}{\sqrt[]{N}})
\end{equation}
where $t^*$ is the value of the Student's t distribution related to the p-value 0.95, and $s$ is the unbiased estimation of the standard deviation obtained by: 
\begin{equation}
s= \sqrt[]{\frac{\sum_{k=1}^N(x_k-\bar{x})^2}{N-1}}
\end{equation}
$\bar{x}$ is the mean of the resulting values.

\subsection{Comparison between methods}\label{exp:batch1}
This set of experiments is designed to investigate the relationship  between the rankings produced by the \MRF framework and those computed by  \textit{PageRank} \cite{ilprints422}, \textit{HITS} \cite{Kleinberg:1999}, and \textit{Versatile} \cite{dedomenico2015ranking}.

To achieve a fair comparison, we setup the \MRF framework accordingly for each compared methods. 
To produce the relative rankings, we applied  \textit{PageRank} and \textit{HITS} to each layer. \textit{Versatile}  was directly executed on the whole multiplex network. For  \MRF,  \textit{PageRank-like} converted the adjacency matrices of the multiplex by applying the transformation presented in \cite{ilprints422}, and used the configuration $A_0A_1$. The \textit{HITS-like} method used the  adjacency matrices of the multiplex,  and the configuration $A_0^TA_0A_1^TA_1$. Finally, the \textit{Versatile-like} method used the  adjacency matrices of the multiplex, and the configuration $A_0A_1$.

The overall setting of this experiment is summarized in the following steps:
\begin{enumerate}

\item Graphs $Gen(V,E)$ are generated using  the graph generators presented in \ref{exp:env:gen}, and varying the number of nodes $|V| \in \{2^s, 6 \le s \le 10\}$. We set the \textit{edge probability} of the Erdos-Renyi model to 0.5, and  we use an intra-community probability of 0.5   and an inter-community probability of 0.2 for the SBM  graphs.

\item The synthetic multiplex network $\mathcal{M}$ with $\mathcal{L}=2$ layers is created according to the procedure described in  \ref{exp:env:syn-mult}. For our experiment, we set the same probability $p_\ell$ for all the layers, and we vary it from 0.0 to 1.0 with a step of 0.05.

\item We execute the original methods (\textit{PageRank}, \textit{HITS}, \textit{Versatile}), and the \MRF methods (\textit{PageRank-like}, \textit{HITS-like}, and \textit{Versatile-like}) to produce all the rankings.

\item The rankings are compared using $\tau_w$ and the \MULTIJ measures as discussed in \ref{exp:env:measures}.

\end{enumerate}

\begin{center}
\begin{figure}[ht]
  \begin{subfigure}[b]{0.48\linewidth}
  \includegraphics[width =\linewidth]{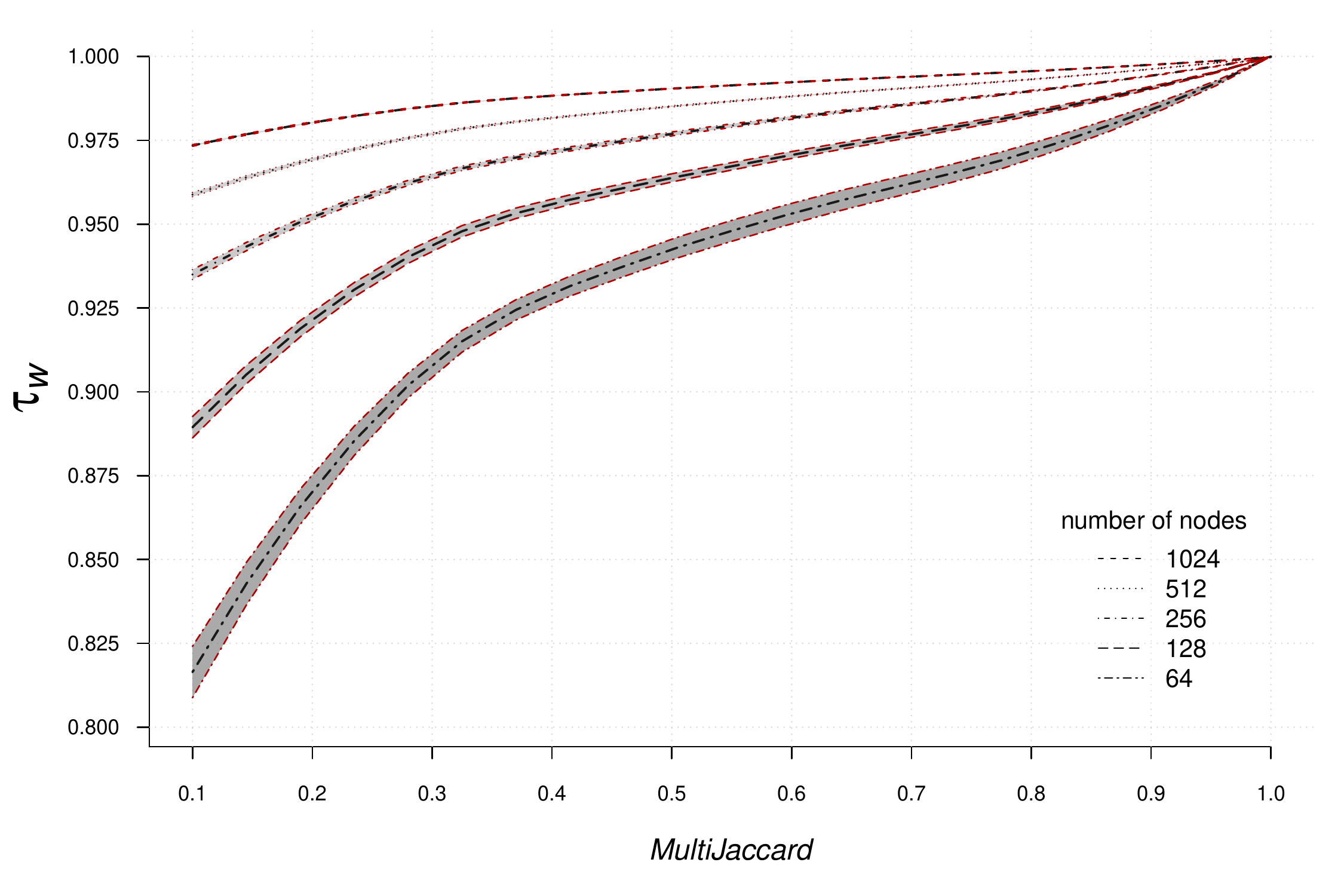}
  \caption{\label{fig:pr_exp1_erdos}}
  \end{subfigure}
  \begin{subfigure}[b]{0.48\linewidth}
  \includegraphics[width =\linewidth]{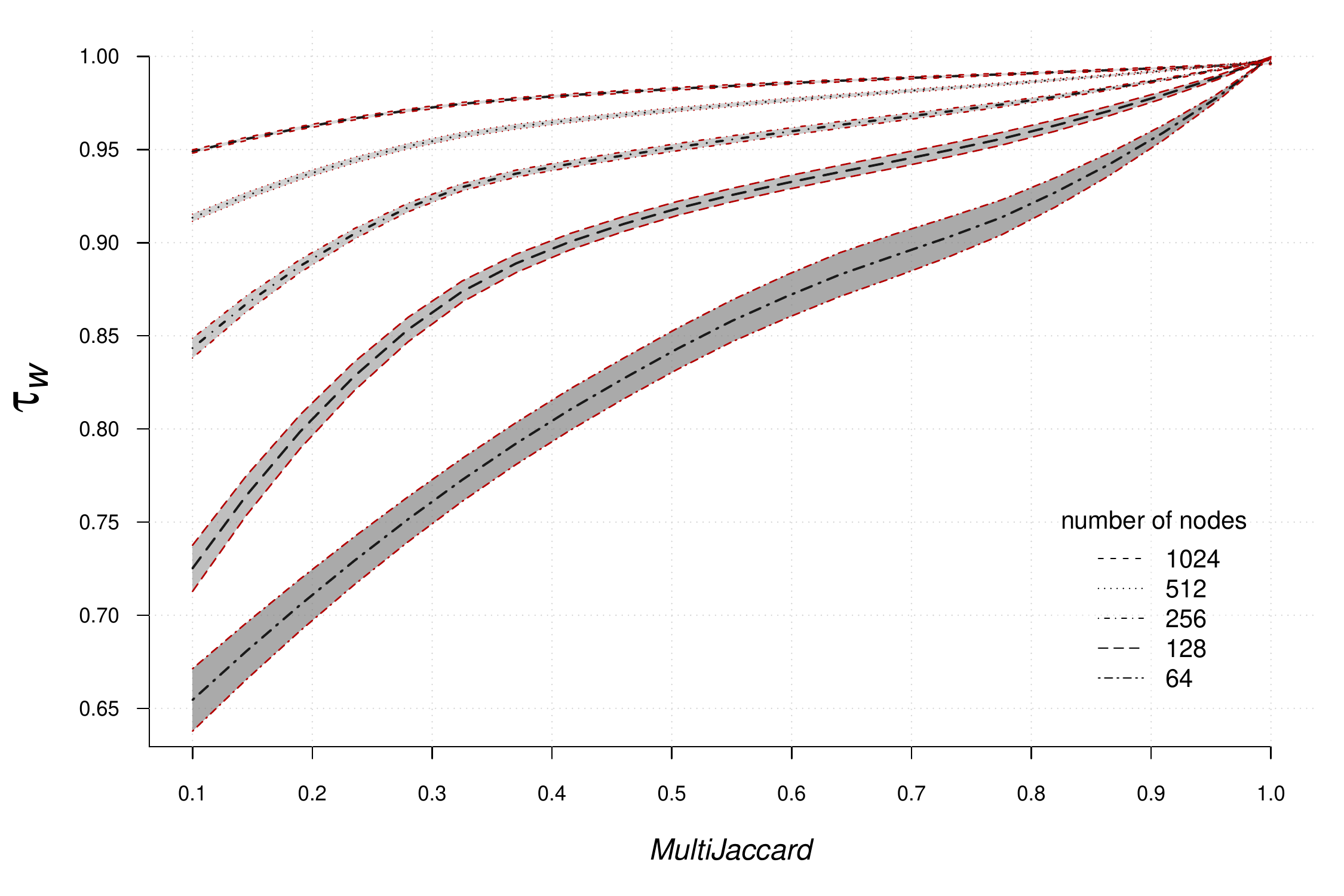}
  \caption{\label{fig:pr_exp1_sbm}}
  \end{subfigure}
  \\
  \begin{subfigure}[b]{0.48\linewidth}
  \includegraphics[width =\linewidth]{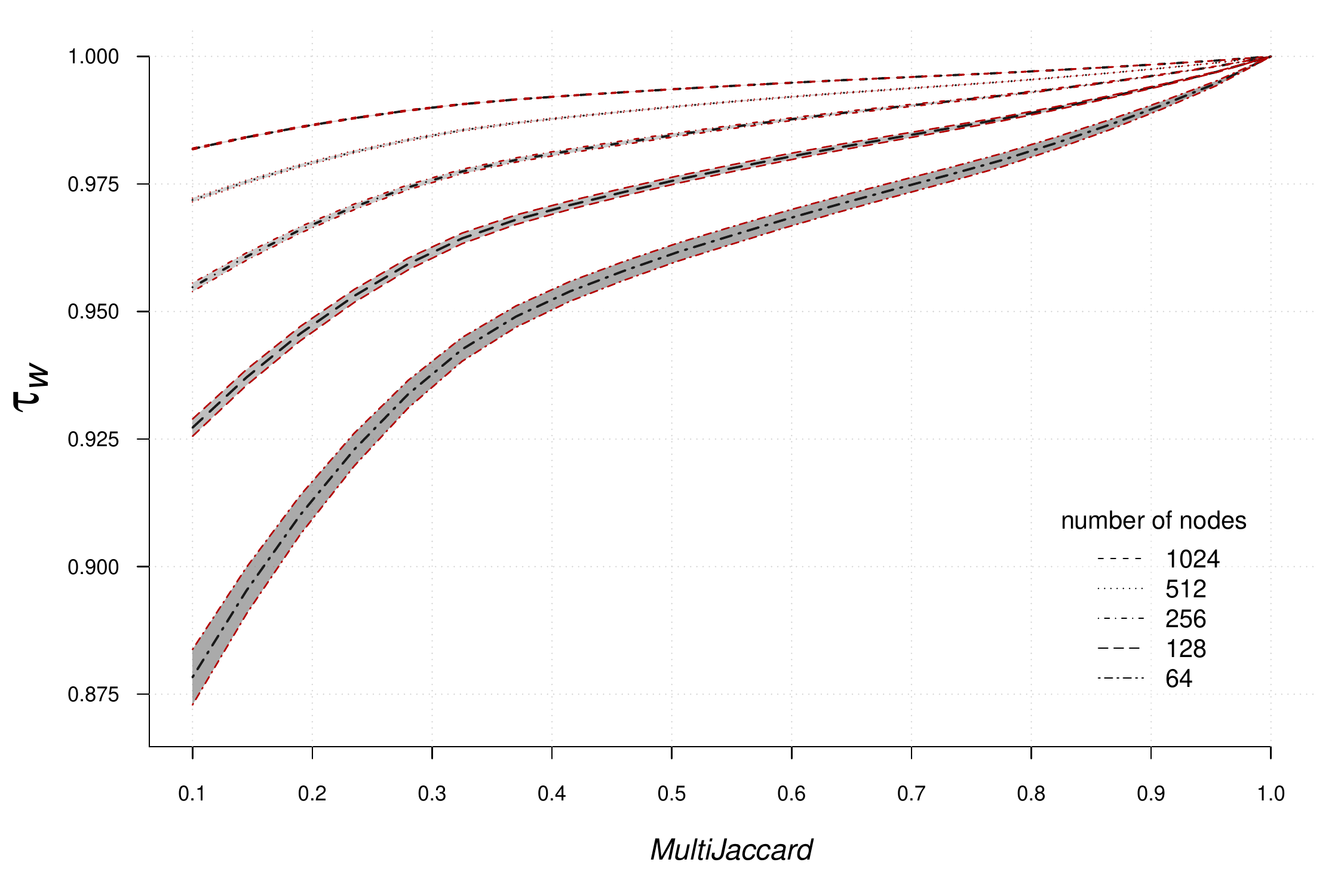}
  \caption{\label{fig:hits_exp1_erdos}}
  \end{subfigure}
  \begin{subfigure}[b]{0.48\linewidth}
  \includegraphics[width =\linewidth]{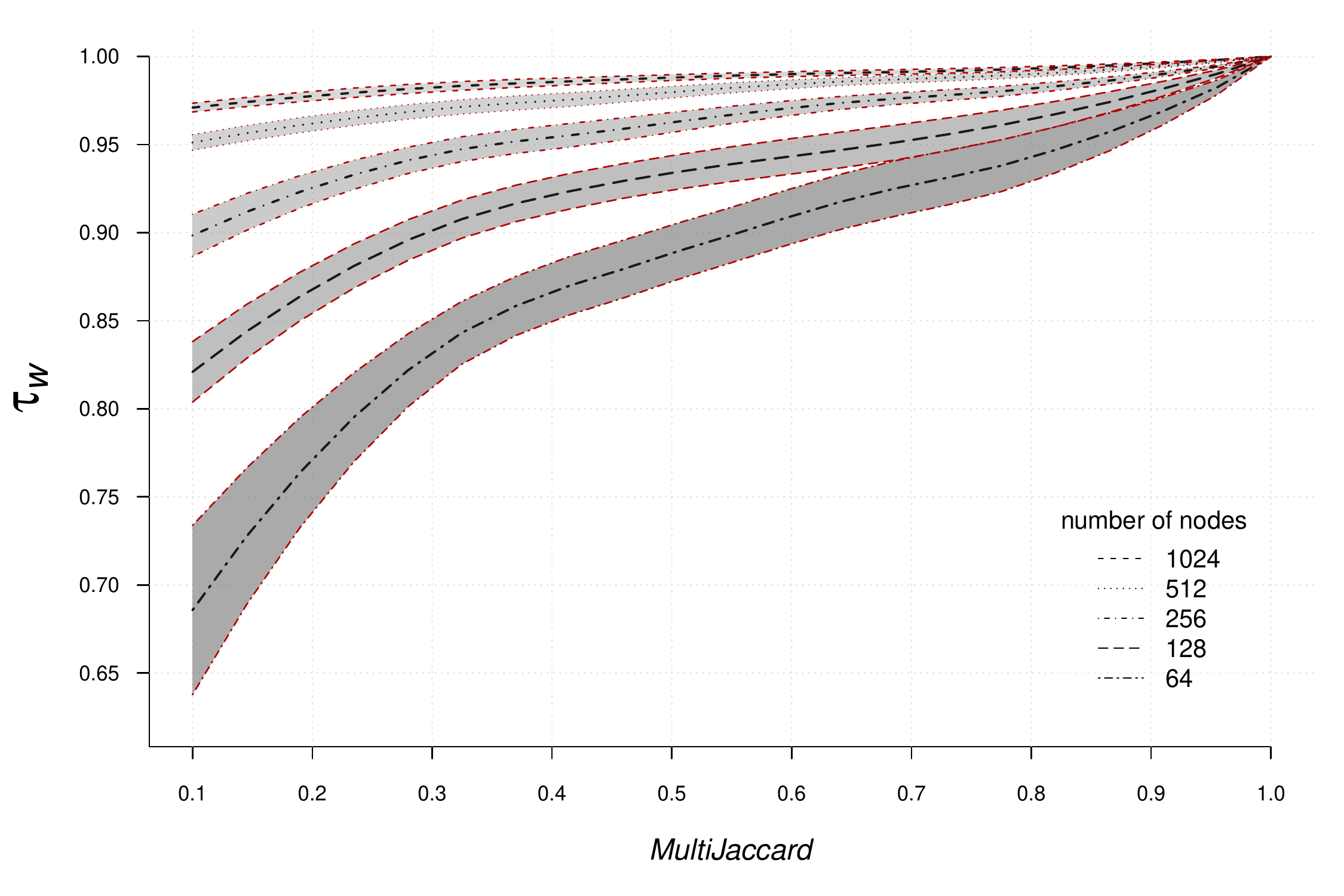}
  \caption{\label{fig:hits_exp1_sbm}}
  \end{subfigure}
  \\
  \begin{subfigure}[b]{0.48\linewidth}
  \includegraphics[width =\linewidth]{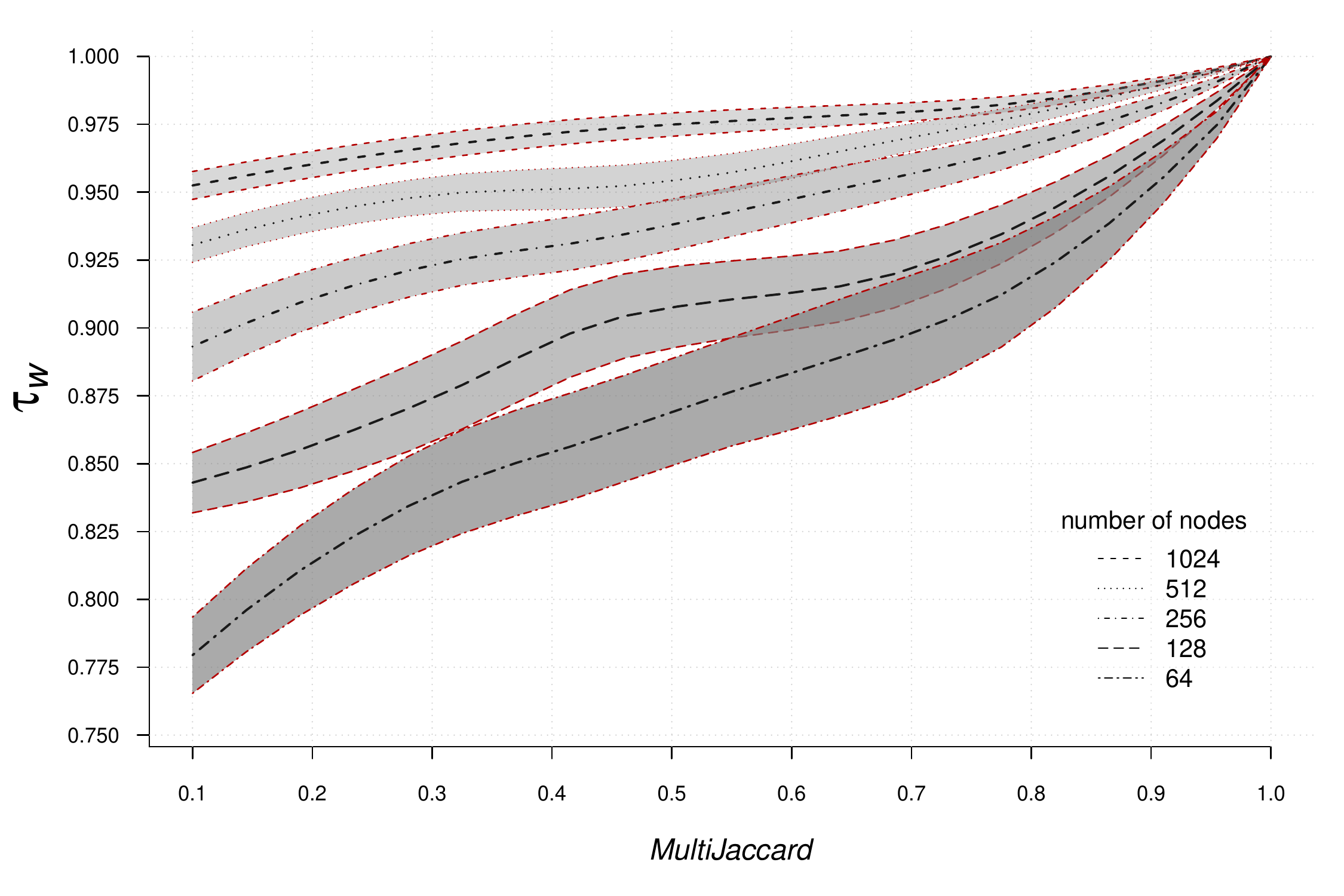}
  \caption{\label{fig:dd_exp1_erdos}}
  \end{subfigure}
  \begin{subfigure}[b]{0.48\linewidth}
  \includegraphics[width =\linewidth]{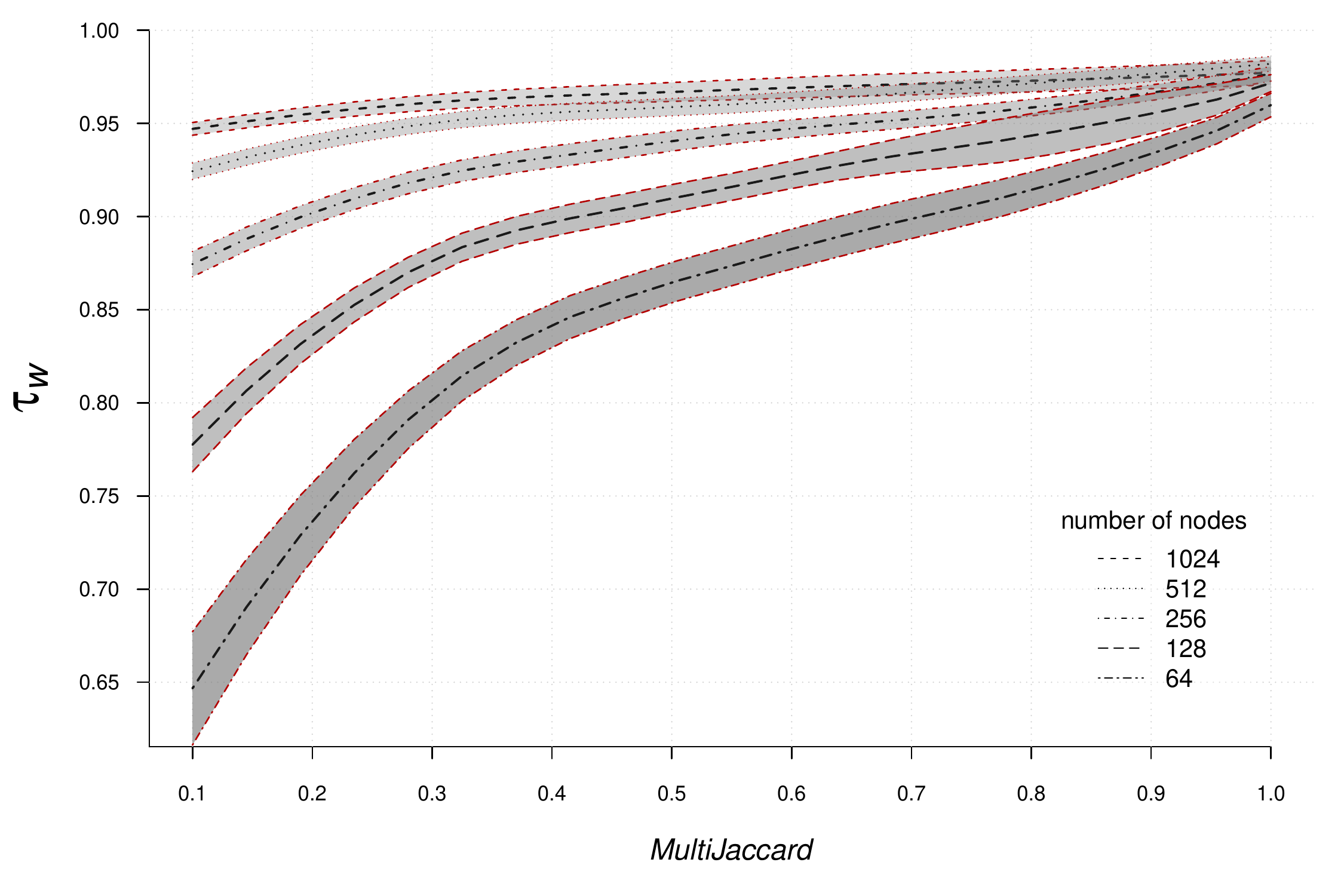}
  \caption{\label{fig:dd_exp1_sbm}}
  \end{subfigure}
  
  \caption{Similarity (Weighted tau) between rankings for varying MultiJaccard coefficients. (a) and (b): \textit{PageRank} vs \textit{PageRank-like}; (c) and (d): \textit{HITS} vs \textit{HITS-like}; (e) and (f): \textit{Versatile} vs \textit{Versatile-like}. (a), (c), and (e) correspond to an Erdos-Renyi graph generator; (b), (d), and (f) correspond to a Stochastic Block Model generator.
  \label{fig:e1}}
\end{figure}
\end{center}

Figure \ref{fig:e1} shows how ranking similarity changes for varying MultiJaccard coefficients. The averages (black dotted lines) are obtained executing all the experiments 32 times. The grey zone gives the confidence interval, as discussed in \ref{exp:env:measures}.
Figures \ref{fig:e1}(a) and \ref{fig:e1}(b)  compare \textit{PageRank} and \textit{PageRank-like};
Figure \ref{fig:e1}(c) and \ref{fig:e1}(d)  compare \textit{HITS} and \textit{HITS-like}; and
Figure \ref{fig:e1}(e) and \ref{fig:e1}(f)  compare \textit{Versatile} and \textit{Versatile-like}. Figures \ref{fig:e1}(a), \ref{fig:e1}(c), and \ref{fig:e1}(e) use  Erdos-Renyi as a generator, while Figures \ref{fig:e1}(b), \ref{fig:e1}(d), and \ref{fig:e1}(f) use SBM  as a generator.

Across all the experiments, the $\tau_w$ coefficient increases monotonically when the \MULTIJ coefficient increases. This behavior is expected: as  the \MULTIJ increases, so does the similarity among the layers, and when \MULTIJ reaches 1.0, all the layers contain the same graph.
Using Erdos-Renyi as generator produces a smaller range of  $\tau_w$ values ($\approx0.8-1.0$) compared to the range produced using an SBM generator ($\approx0.6-1.0$). A graph corresponding to the Erdos-Renyi model has an edge distribution closer to uniform (see Figure \ref{fig:graph_types}), and this is reflected in higher values of $w_\tau$, for a given \MULTIJ value (compared to SBM).

The confidence intervals present a similar trend when the \MRF framework is compared against \textit{PageRank} and \textit{HITS}. The plots related to the Erdos-Renyi generator (Figures \ref{fig:e1} (a) and (c)) have smaller confidence intervals  when compared to those related to the SBM generator (Figures \ref{fig:e1} (b) and (d)). The wider confidence intervals in Figures \ref{fig:e1} (b) and (d) are due to the more complex topology of the SBM  model. In contrast, the comparison between \textit{Versatile} and \textit{Versatile-like} manifests a different behavior. The Erdos-Renyi related  plot in Figure\ref{fig:e1} (c) has wider confidence intervals than the SBM one in Figure \ref{fig:e1} (d). Since the \textit{Versatile} method works directly on the whole multiplex, when the \MULTIJ is equal to 1.0 the  $w_\tau$ coefficient does not reach exactly the value of 1.0 in Figure \ref{fig:e1} (f), although it comes very close to it.

Another interesting aspect is the ``stratification'' of the data. The larger the number of nodes is, the smaller is the difference between the compared methods. This phenomenon might be due to the fact that the score of a node is less influential (when more nodes are present), and the two rankings become more similar even for lower values of the \MULTIJ coefficient.

\subsection{Configuration and Shift Impact}\label{exp:batch2}
This set of experiments is designed to investigate the impact that different \textit{configurations} and  \textit{shifts} in a multiplex of two layers have on the produced rankings.
To compare the results, we need to consider a reference ranking. To this end, we chose to use the rankings obtained by \textit{HITS} on every layer.
For this set of experiments we use the same type of generators, we fix the number of nodes of the generated multiplex to 256, and we measure the  results adopting the same measures (\MULTIJ and $w_\tau$) described in Section \ref{exp:batch1}.
All results are given in Figures \ref{fig:e2_conf_A}, \ref{fig:e2_conf_B}, and \ref{fig:e2_shift}. Each configuration and each member of the equivalence class produces different results. In general, we observe that increasing values of \MULTIJ  give larger $w_\tau$ values, although the latter does not always reach the value of one.  Henceforth, we discuss representative results.

\subsubsection{Impact of the configuration}\label{exp:batch2:conf}
Figure \ref{fig:e2_conf_short} shows all the members of the equivalence classes $A_0A_1^TA_1A_0^T$ ((a) and (b)) and $A_0A_0^TA_1A_1^T$ ((c) and (d)).  Figures (a) and (c) use an Erdos-Renyi generator, while Figures (b) and (d) use a stochastic block model as generator. 
If we consider the same member of the equivalence class and compare the results obtained with different generators, we can see that the trends are the same in both cases.
It's important to observe that when \MULTIJ is equal to one in Figures \ref{fig:e2_conf_short} (a) and (b), the members $A_0^TA_0A_1^TA_1$ and   $A_1^TA_1A_0^TA_0$  become equal to $A_0^TA_0A_0^TA_0$, and the produced ranking is exactly the same as the one given by \textit{HITS}. 
Similarly, in Figures \ref{fig:e2_conf_short}(c) and (d), we see that when the members $A_0^TA_1A_1^TA_0$ and $A_1^TA_1A_0^TA_0$  become equal to $A_0^TA_0A_0^TA_0$, the produced ranking is the same as the one given by \textit{HITS}.
For both \textit{configurations} $A_0A_1^TA_1A_0^T$ and $A_0A_0^TA_1A_1^T$, it's 
important to observe that, when the member of the class has the order between a matrix and a transposed matrix inverted ($A_0A_1^TA_1A_0^T$, $A_1A_0^TA_0A_1^T$ and $A_0A_0^TA_1A_1^T$, $A_1A_1^TA_0A_0^T$), then there is no correlation between the produced ranks and the one computed by \textit{HITS}. With an Erdos-Renyi generator, in particular, the average in the latter case is always close to zero.
For a complete overview of the experiments related to different \textit{configurations}, see Figures \ref{fig:e2_conf_A} and \ref{fig:e2_conf_B} of the Appendix. 

\subsubsection{Impact of the shift}\label{exp:batch2:shift}
As defined in Section \ref{sec:problem_def}, a \textit{configuration} identifies an equivalence class. It's then necessary to identify a member of the class by rotating the representative member by $h$ position. For a complete overview of the experiments that analyze changes in similarity due to each $\textit{shift}_h$, see Figure \ref{fig:e2_shift} of the Appendix.
In Figure  \ref{fig:e2_shift_short}, we present the analysis related to $\textit{shift}_1$ ((a) and (b)) and $\textit{shift}_3$ ((c) and (d)). Figures (a) and (c) use  an Erdos-Renyi generator, while Figures (b) and (d) use a stochastic block model generator. As expected, $\textit{shift}_1$ and $\textit{shift}_3$ are those shifts in which the configurations $A_0A_1^TA_1A_0^T$ and $A_0A_0^TA_1A_1^T$   converge to a $w_\tau$ equal to one when \MULTIJ reaches one. The plots clearly show that the configurations should be grouped by two, and the groups expose the same relative trends. 
In $\textit{shift}_1$, the identified groups are: 
$\{A_0^TA_1A_1^TA_0,A_1^TA_1A_0^TA_0\}$, $\{A_0^TA_1^TA_1A_0,A_1^TA_0^TA_1A_0\}$, $\{A_1A_0^TA_1^TA_0, A_1A_1^TA_0^TA_0\}$, as shown in Figures \ref{fig:e2_shift_short} (a) and (b).
In $shift_3$ the identified groups are: $\{A_0^TA_0A_1^TA_1,A_1^TA_0A_0^TA_1\}$, $\{A_0^TA_0A_1A_1^T,A_1^TA_0A_1A_0^T\}$, $\{A_1A_0A_0^TA_1^T, A_1A_0A_1^TA_0^T\}$, as shown in Figures \ref{fig:e2_shift_short} (c) and (d).
The first group of each shift has larger $w_\tau$ values. The remaining groups overlap with one other. Furthermore, for the set of experiments related to \textit{shift} analysis, we can see that increasing \MULTIJ  values produce larger $w_\tau$ values.

\begin{center}
\begin{figure}[ht]
  \begin{subfigure}[b]{0.48\linewidth}
  \includegraphics[width =\linewidth]{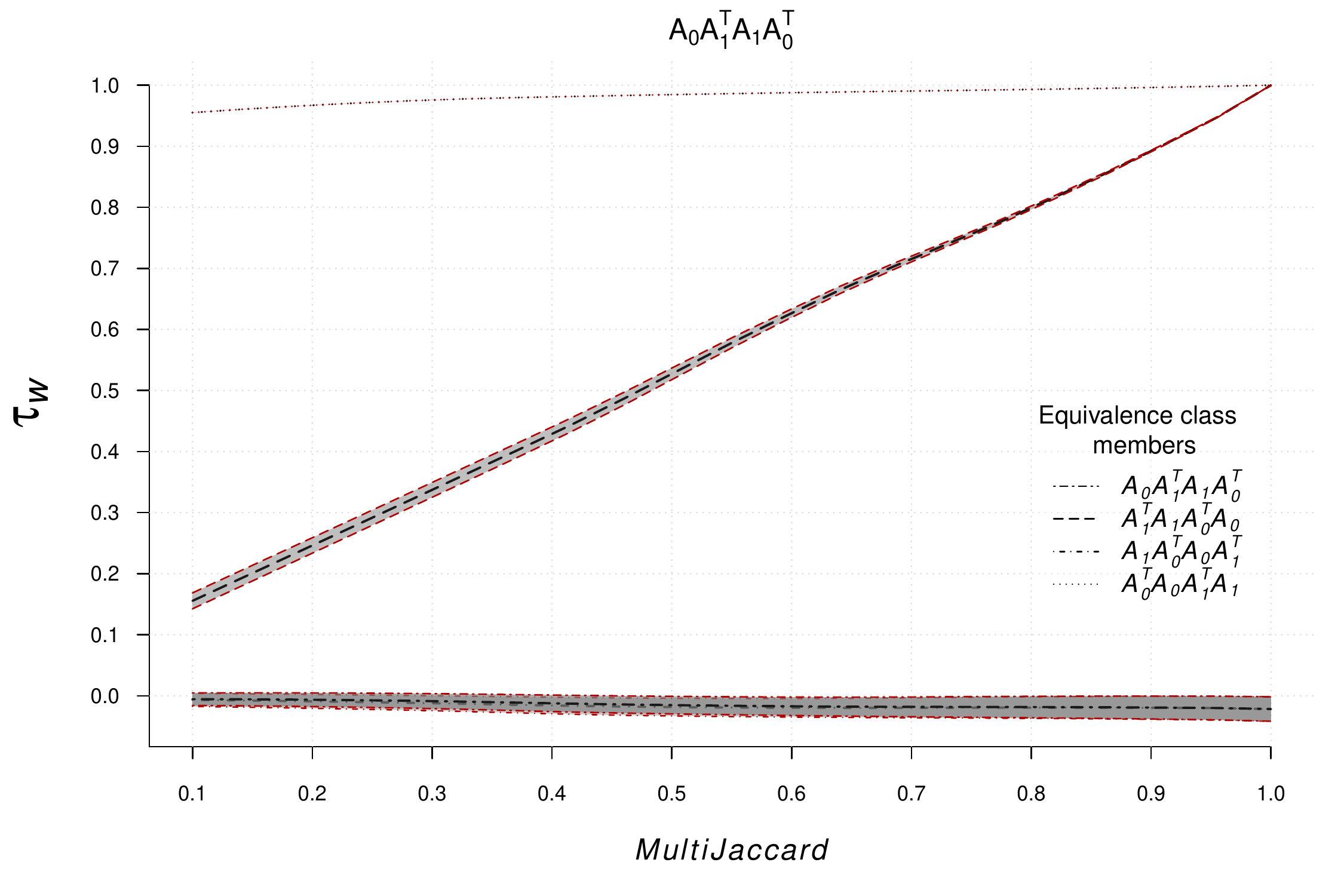}
  \caption{\label{fig:e2_s_erdos_c_A0A1TA1A0T}}
  \end{subfigure}
  \begin{subfigure}[b]{0.48\linewidth}
  \includegraphics[width =\linewidth]{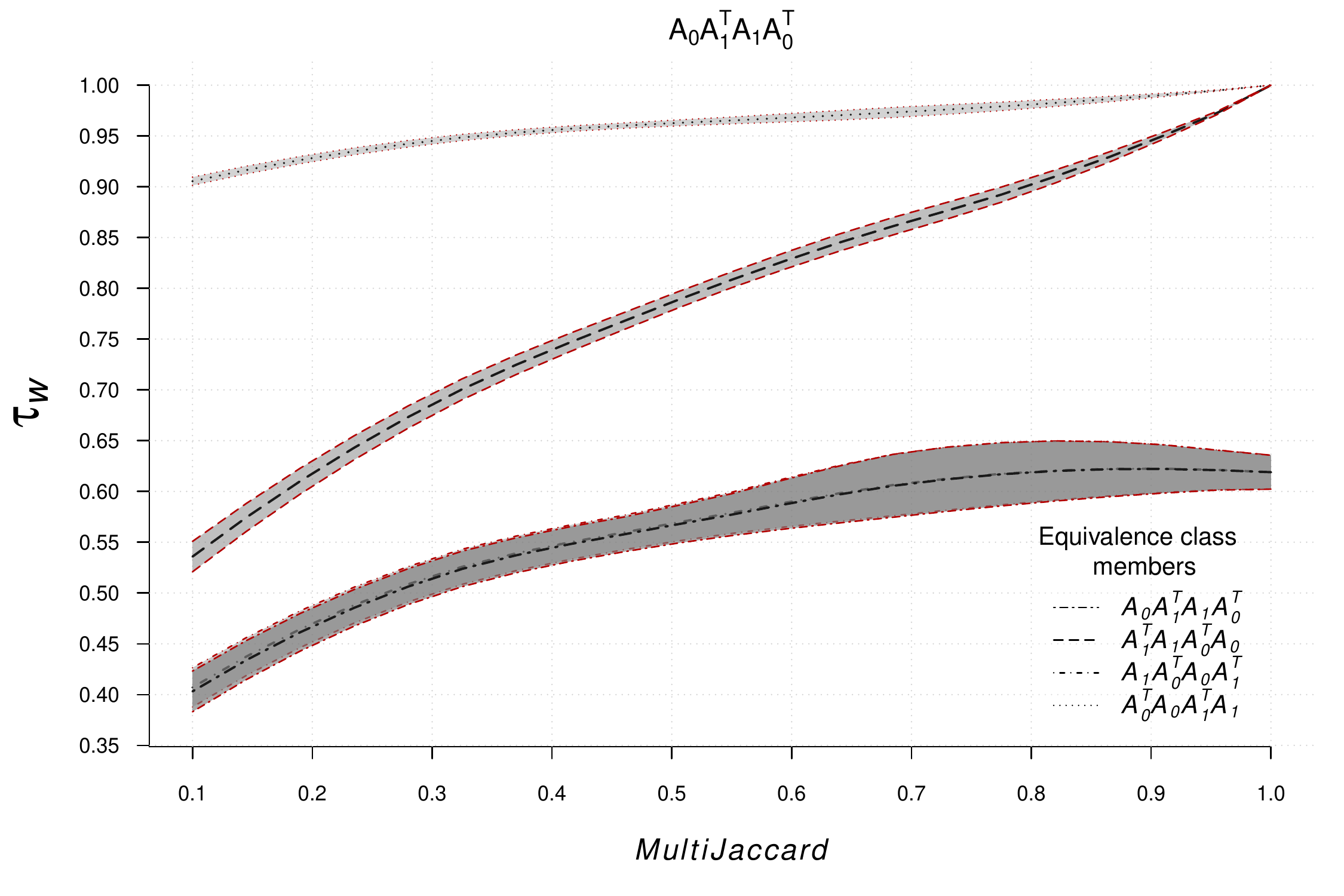}
  \caption{\label{fig:e2_s_sbm_c_A0A1TA1A0T}}
  \end{subfigure}
  \\
  \begin{subfigure}[b]{0.48\linewidth}
  \includegraphics[width =\linewidth]{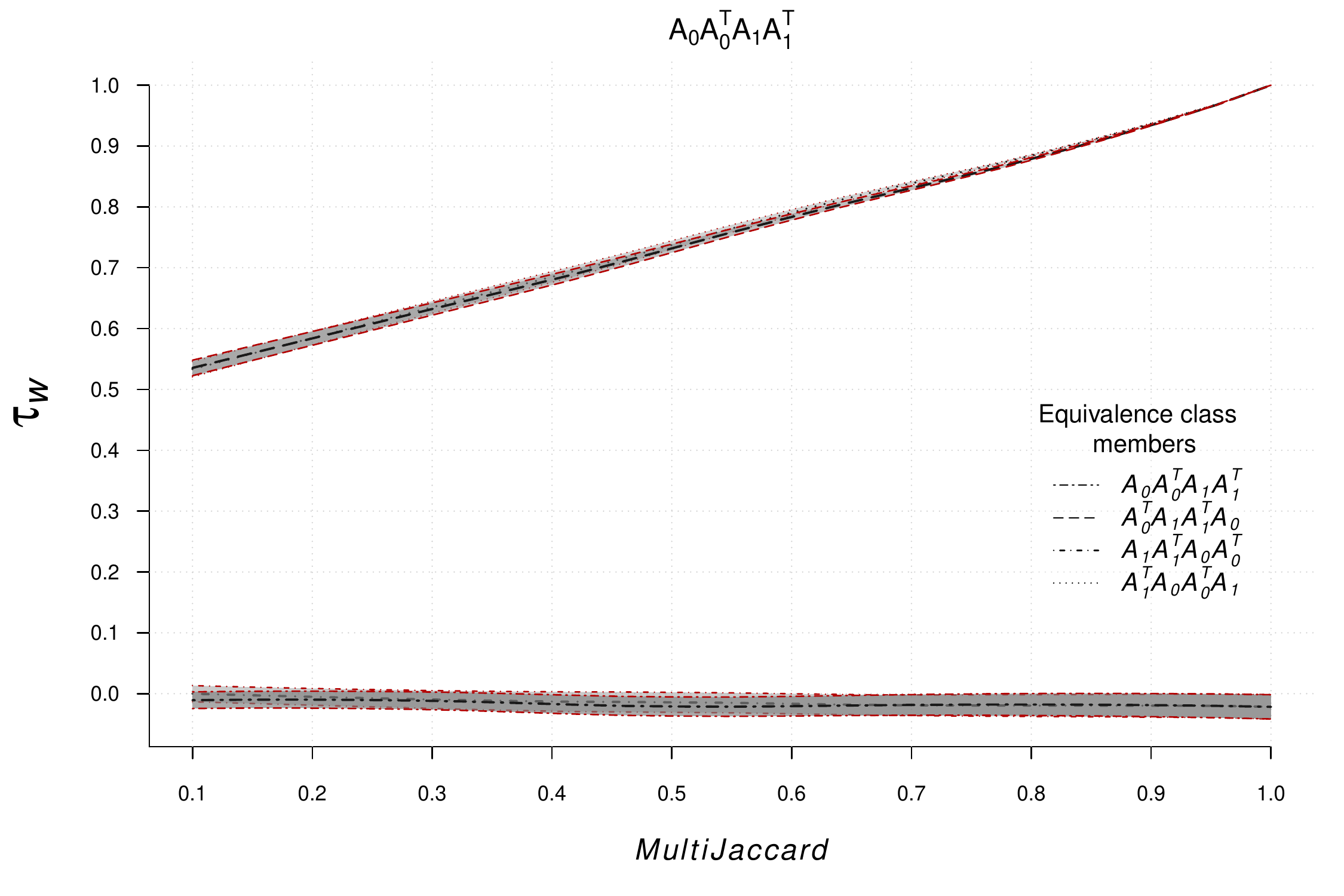}
  \caption{\label{fig:e2_s_erdos_c_A0A0TA1A1T}}
  \end{subfigure}
  \begin{subfigure}[b]{0.48\linewidth}
  \includegraphics[width =\linewidth]{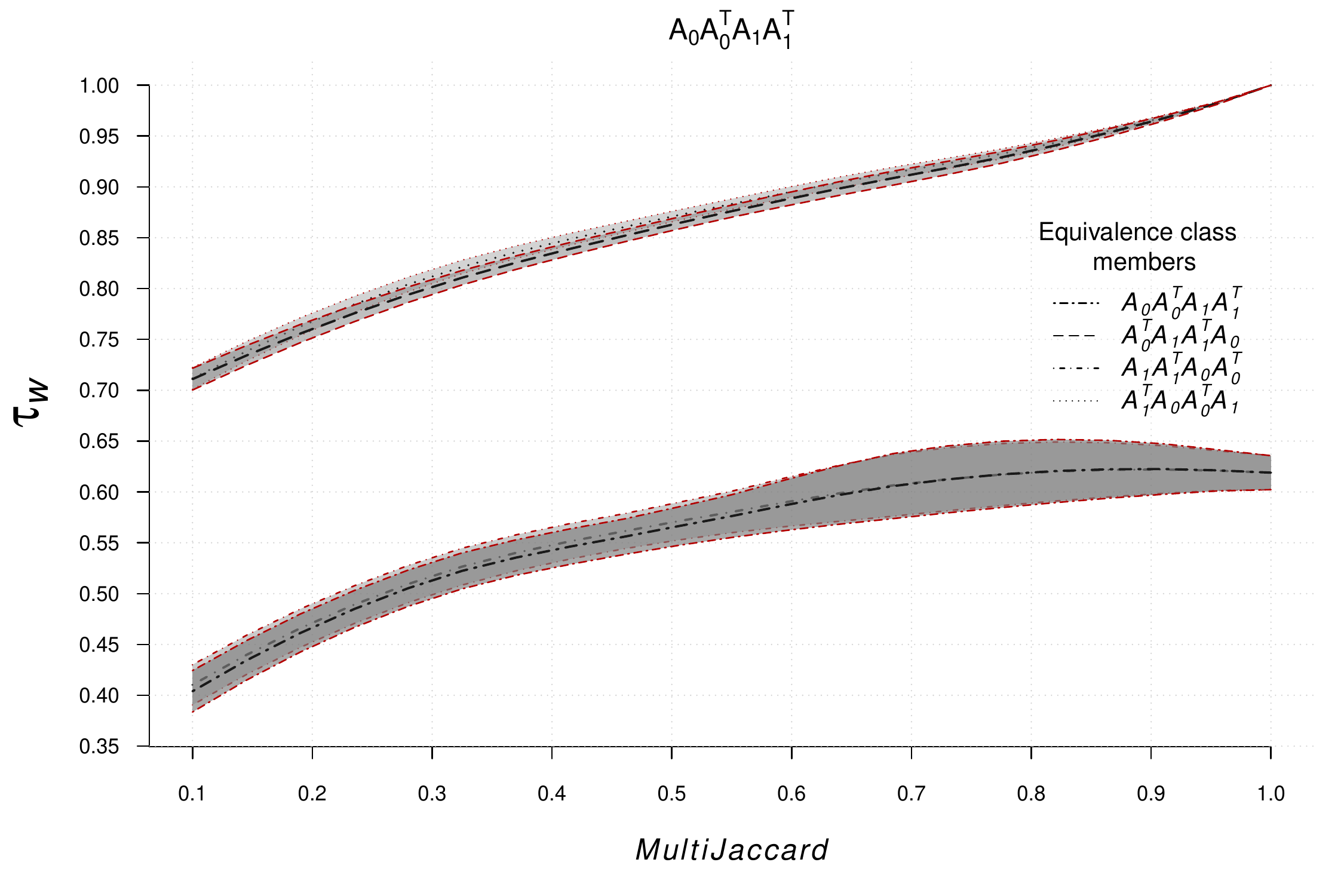}
  \caption{\label{fig:e2_s_sbm_c_A0A0TA1A1T}}
  \end{subfigure}
 
  \caption{Similarity between rankings produced by \textit{HITS} and \MRF for a multiplex network of two layers, considering all the members of the equivalence classes. (a) and (b):  $A_0A_1^TA_1A_0^T$; (c) and (d):  $A_0A_0^TA_1A_1^T$. (a) and (c) use an Erdos-Renyi graph as generator; (b) and (d) use stochastic block model graph as generator.
  \label{fig:e2_conf_short}}
\end{figure}
\end{center}

\begin{center}
\begin{figure}[ht]
    \begin{subfigure}[b]{0.48\linewidth}
  \includegraphics[width =\linewidth]{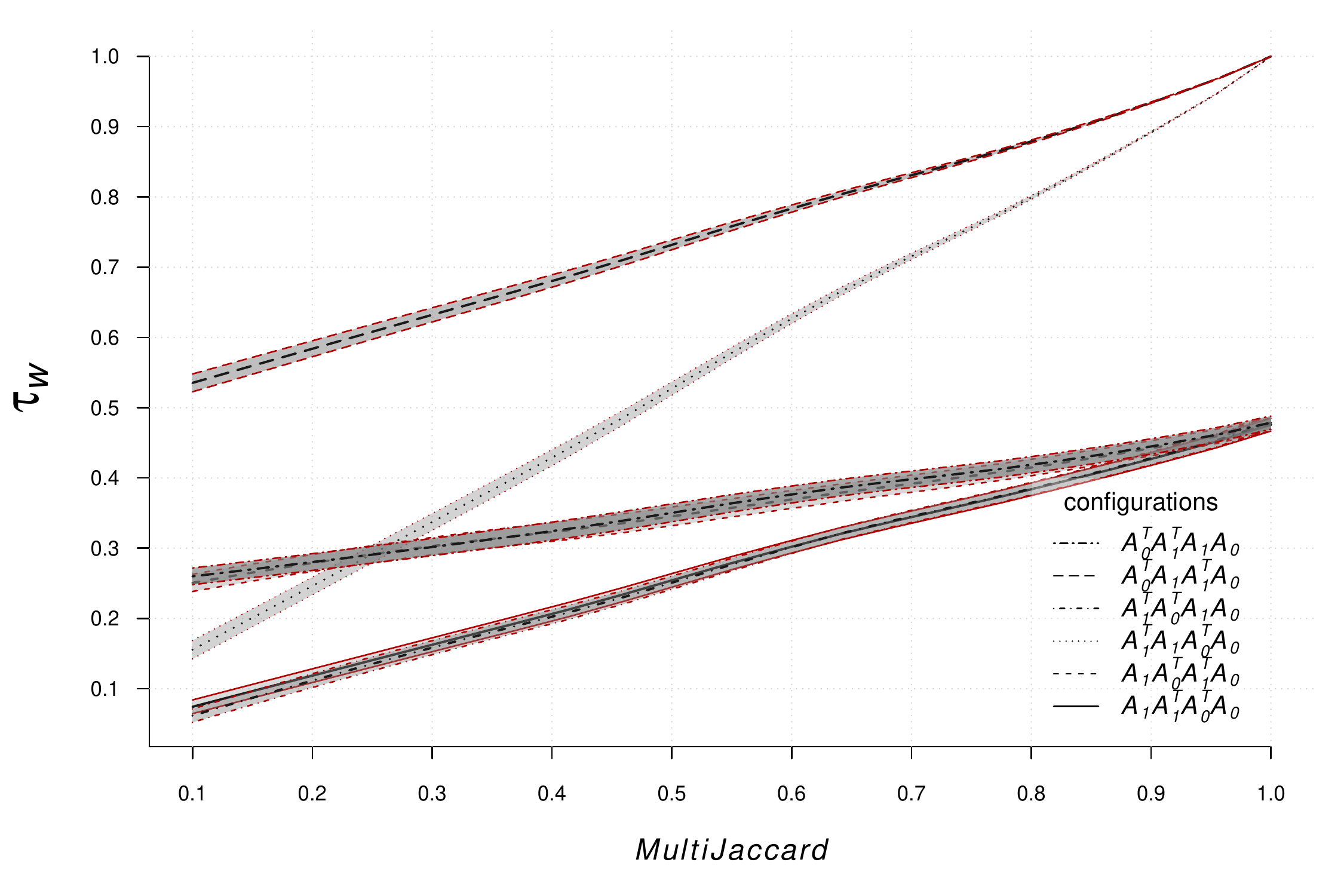}
  \caption{\label{fig:e2_s_erdos_shift_1}}
  \end{subfigure}
  \begin{subfigure}[b]{0.48\linewidth}
  \includegraphics[width =\linewidth]{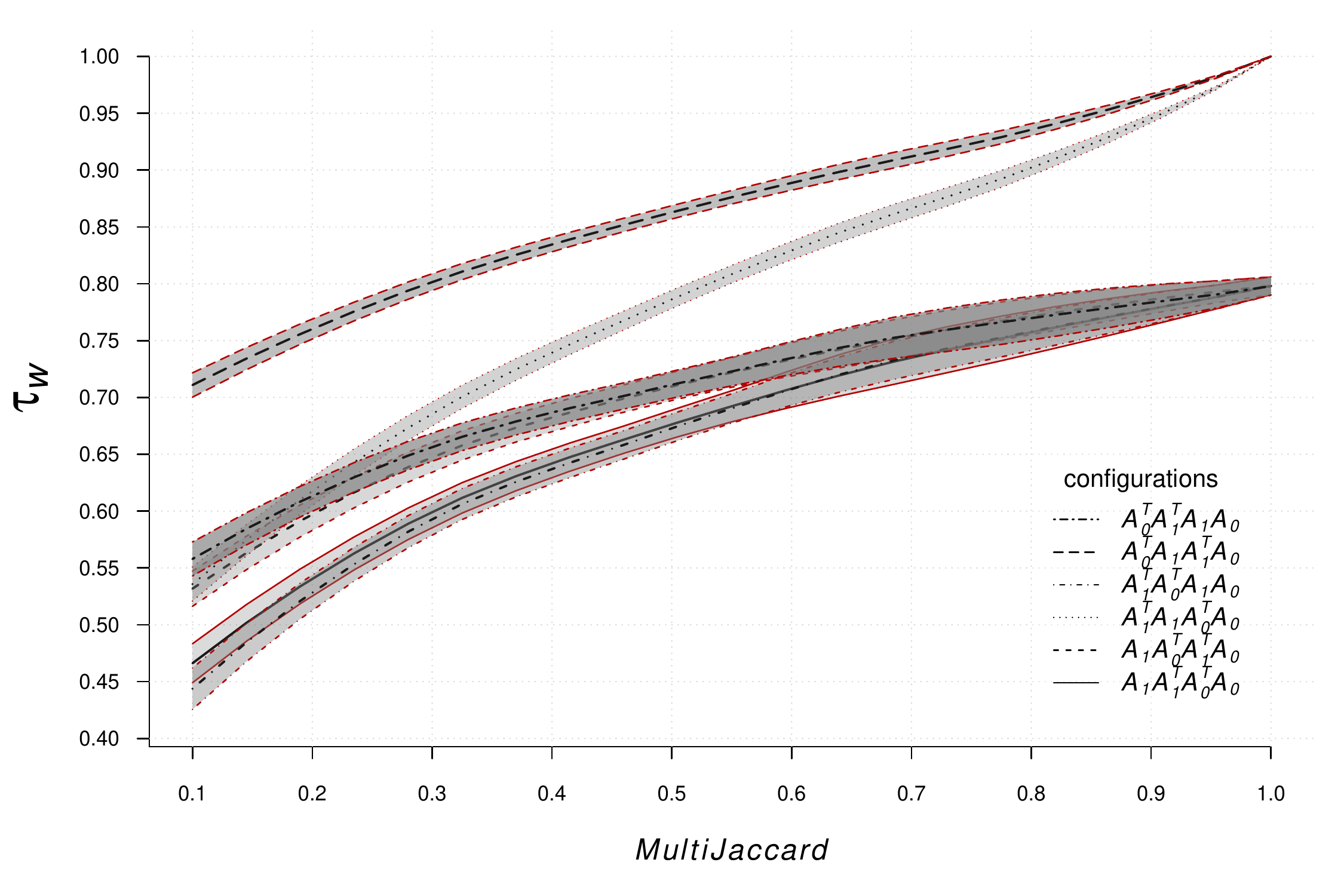}
  \caption{\label{fig:e2_s_sbm_shift_1}}
  \end{subfigure}
  \\
  \begin{subfigure}[b]{0.48\linewidth}
  \includegraphics[width =\linewidth]{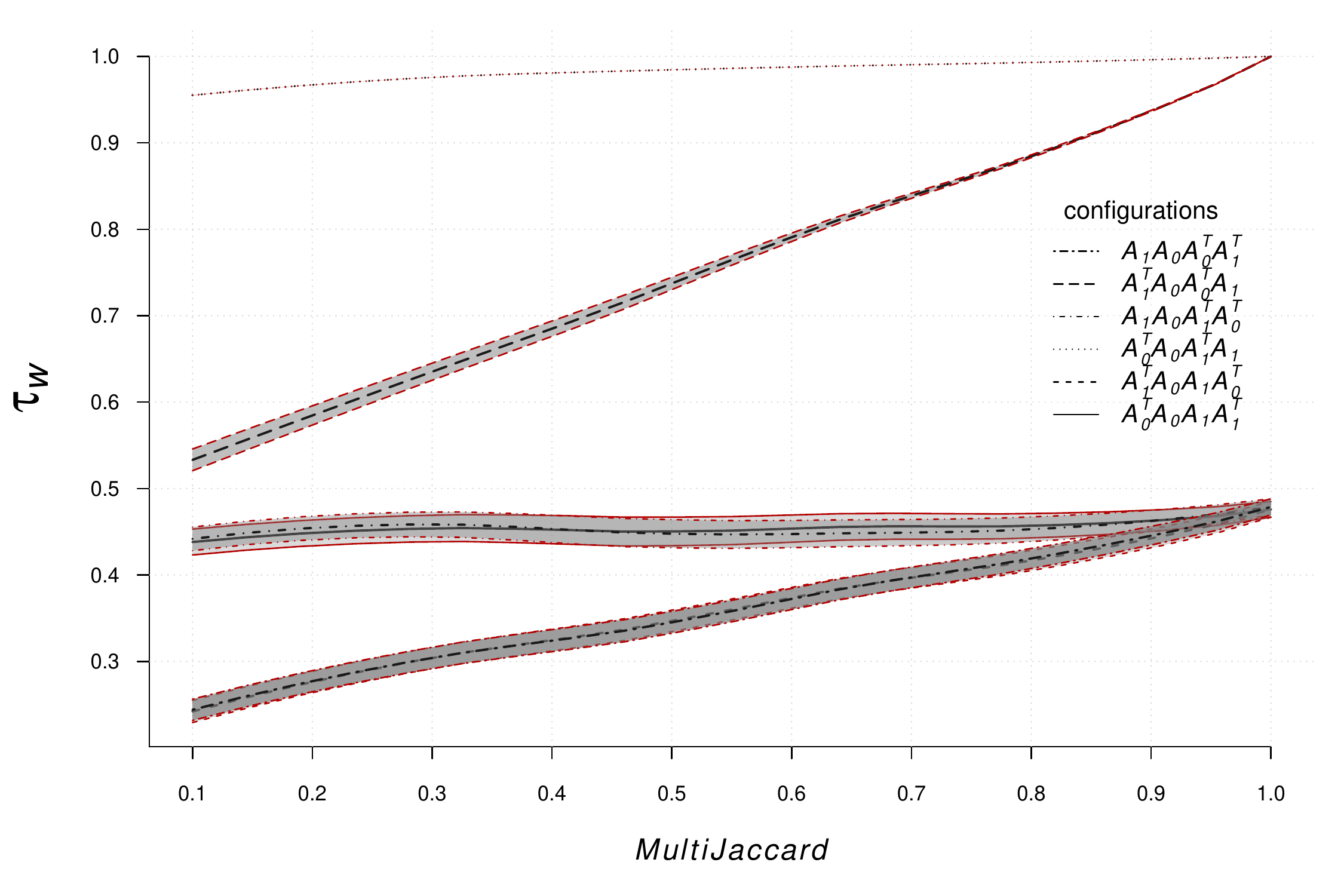}
  \caption{\label{fig:e2_s_erdos_shift_3}}
  \end{subfigure}
  \begin{subfigure}[b]{0.48\linewidth}
  \includegraphics[width =\linewidth]{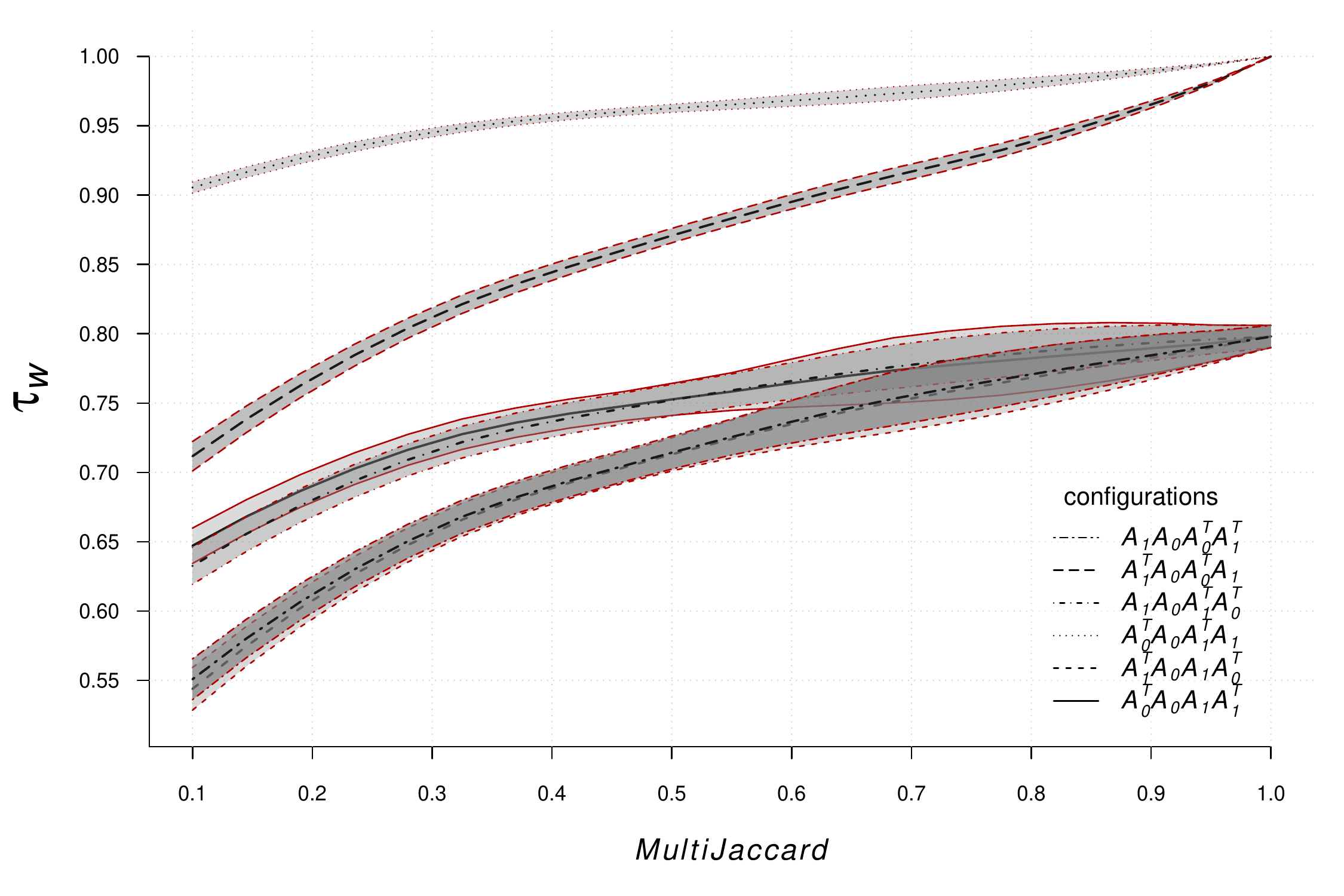}
  \caption{\label{fig:e2_s_sbm_shift_3}}
  \end{subfigure}

  \caption{Similarity between rankings produced by \textit{HITS} and \MRF for a multiplex network of two layers using all possible configurations. (a) and (b): $shift_1$; (c) and (d):  $shift_3$. (a) and (c) use  an Erdos-Renyi graph as generator; (b) and (d) use  a stochastic block model graph as generator.
  \label{fig:e2_shift_short}}
\end{figure}
\end{center}

\subsection{Convergence of the Method}\label{exp:batch3}

Let us now denote by $v(\tau)$ the ranking vector corresponding to the matrix $M(\tau)$ as in Section \ref{subsubsec:proof}, and set $v(0) = \lim_{\tau \to 0} v(\tau)$.

The function $\tau \mapsto v(\tau)$ is analytic for all (but finitely many) values of $\tau$, hence for all values of $\tau$ lying in a suitable neighborhood of $0$, with the possible exception of $0$. By the arguments in Section \ref{subsec:iterative-method}, it is also analytic in $0$ in real cases. Therefore, one has a first order expansion $v(\tau) \approx v(0) + {\mathbf{q}} \tau,$ which implies $||v(\tau) - v(0)||_1 \sim q \tau$ and also $||v(2\tau) - v(\tau)||_1 \sim q\tau,$ where $q = ||{\mathbf{q}}||_1$ is an opportune constant.

The method we outline in Section \ref{sec:exp} computes each $v(\tau)$ by iterated (renormalized) applications of $M = M(\tau)$. The speed of convergence of each iterated vector $\rr^t$ to $v(\tau)$ depends on (the maximal value of the complex norm of) the ratio between the principal eigenvalue of $M(\tau)$ and other eigenvalues. This ratio, as a function of $\tau$, will also tend to $1$ linearly when $\tau$ approaches $0$.\footnote{If this ratio stays away from $1$ when $\tau \to 0$, then convergence is much faster and the principal eigenvalue stays simple.} As the initial approximate value of the ranking vector we use is the one obtained in the previous iteration, i.e., $v(2\tau)$, the number $\delta$ of iterations that are needed in order to ensure convergence to $v(\tau)$ at each $\tau$-step is therefore (asymptotically, when $\tau \to 0$) constant.

We would like to stress the fact that in the examples that we have worked out $\delta$ is surprisingly small ($4 \lesssim \delta  \lesssim 8$ will typically suffice to yield machine precision convergence at each given value of $\tau$), showing effectiveness of our implementation. Also, the error between $v(\tau)$ and $v(0)$ is proportional to $\tau$; as at each step $\tau$ gets halved, we achieve exponential convergence to $v(0)$.

Figures \ref{fig:e3_s_conv} and \ref{fig:e3_s_const} show results from experiments run with two layers (hence four matrices) in a \textit{HITS-like} framework on graphs with sparse connectivity. Averages (black dotted lines) have been obtained by running each experiments 32 times and are decorated by the grey areas showing up confidence interval described in \ref{exp:env:measures}. Figures \ref{fig:e3_conv} and \ref{fig:e3_const} at the end of the paper are magnified versions of Figures \ref{fig:e3_s_conv} and \ref{fig:e3_s_const}.

\begin{center}
\begin{figure}[ht]
    \begin{subfigure}[b]{0.48\linewidth}
  \includegraphics[width =\linewidth]{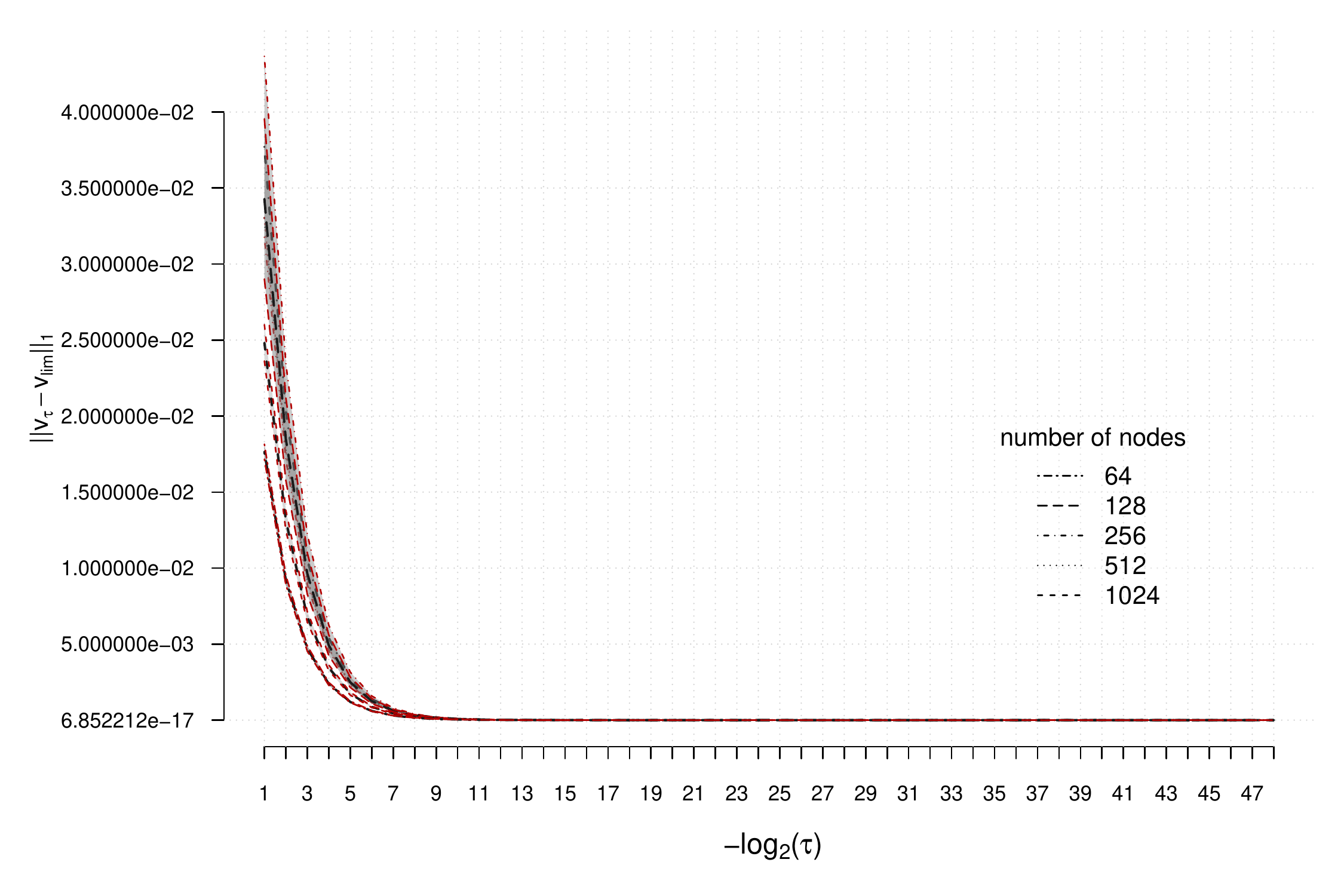}
  \caption{\label{fig:e3_s_erdos_conv}}
  \end{subfigure}
  \begin{subfigure}[b]{0.48\linewidth}
  \includegraphics[width =\linewidth]{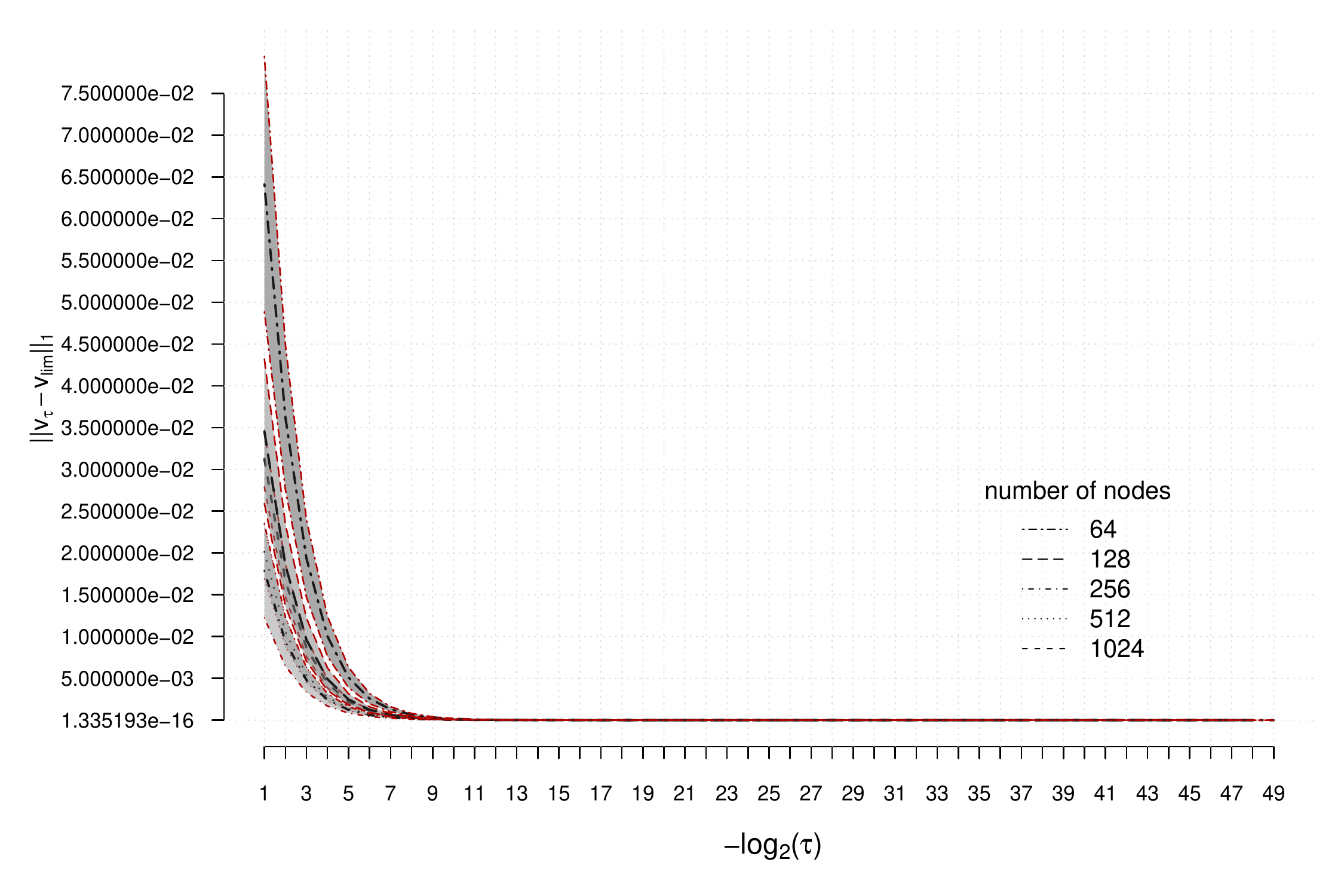}
  \caption{\label{fig:e3_s_sbm_conv}}
  \end{subfigure}  
  \caption{\label{fig:e3_s_conv} Convergence to the limit ranking vector using an Erdos-Renyi graph and a stochastic block model graph. The plots show $||v(\tau)-v(0)||_1$ against the number of times $\tau$ has been halved.}  
\end{figure}
\end{center}

\begin{center}
\begin{figure}[ht]
    \begin{subfigure}[b]{0.48\linewidth}
  \includegraphics[width =\linewidth]{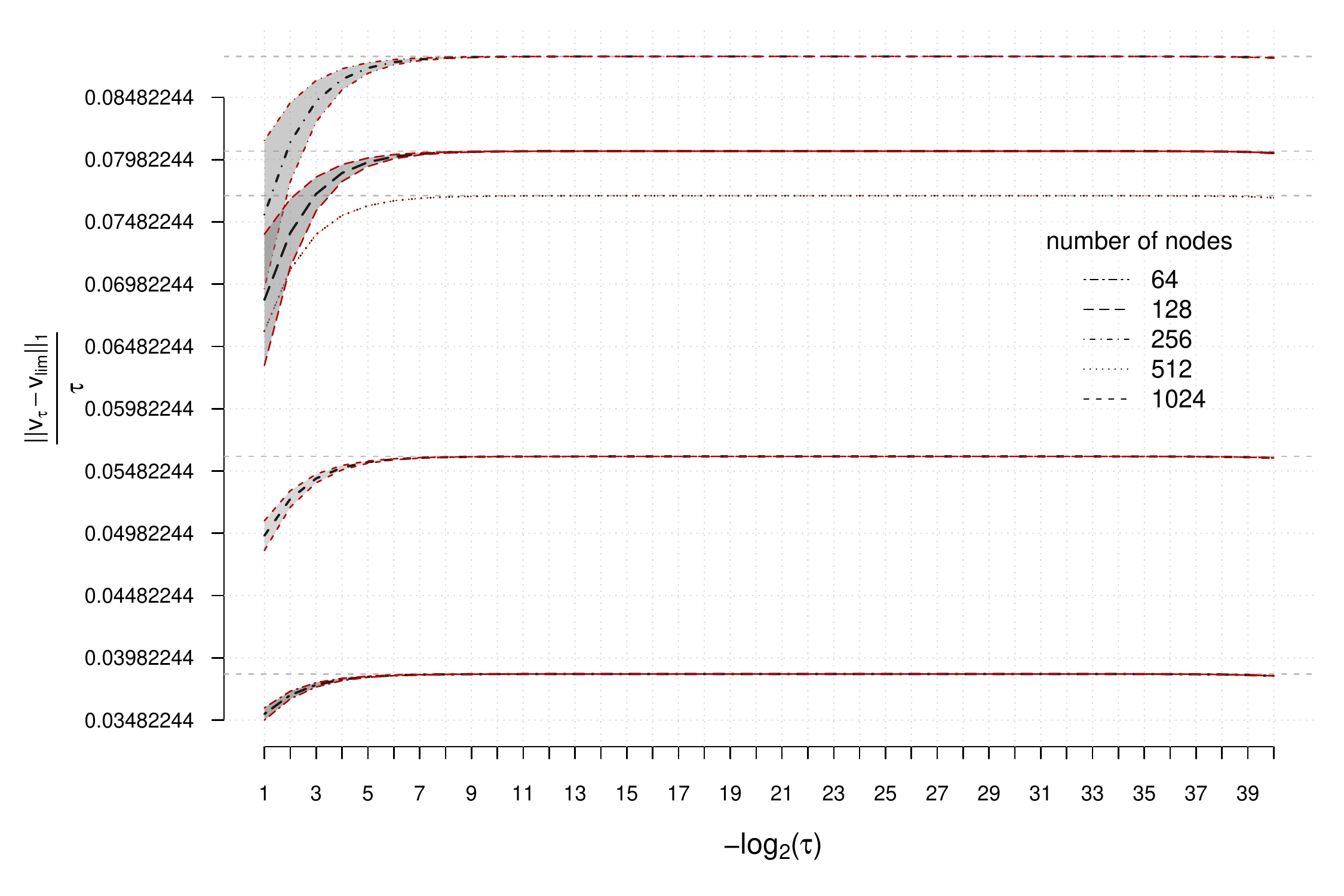}
  \caption{\label{fig:e3_s_erdos_const}}
  \end{subfigure}
  \begin{subfigure}[b]{0.48\linewidth}
  \includegraphics[width =\linewidth]{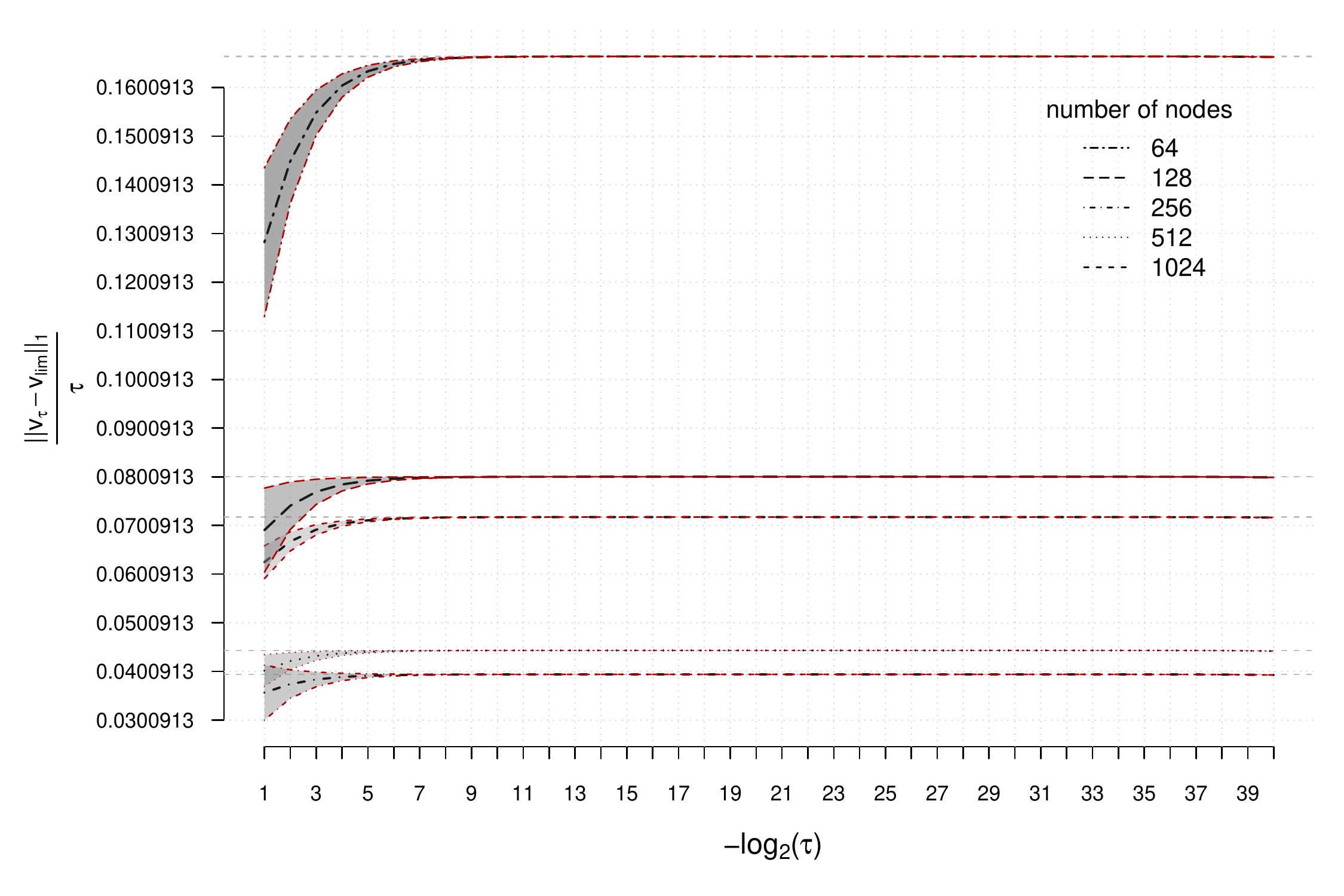}
  \caption{\label{fig:e3_s_sbm_const}}
  \end{subfigure}  
  \caption{\label{fig:e3_s_const} The quantity $||v(\tau) - v(0)||_1/\tau$ is plotted against $-\log_2\tau$, thus highlighting the (asymptotical) linear behavior of $||v(\tau)- v(0)||_1$ as a function of $\tau$.}  
\end{figure}
\end{center}

Figures \ref{fig:e3_s_erdos_conv} and \ref{fig:e3_s_sbm_conv} provide evidence of exponential convergence in experiments run in an Erdos-Renyi and SBM setting.

Linear dependence of $||v(\tau) - v(0)||_1$ on $\tau$ is highlighted in Figures \ref{fig:e3_s_erdos_const} and \ref{fig:e3_s_sbm_const} in an Erdos-Renyi and SBM setting. The behavior becomes fully linear after very few (about $10$) halvings of $\tau$. In both settings, the linearity constant appears to depend only on the number of nodes, yet in an erratic way (the dependence is neither increasing nor decreasing). The linearity constants stay, however, within the same order of magnitude in all experiments. 

The choice of multiplying $\tau$ by $1/2$ at each step is arbitrary, as any other positive constant smaller than $1$ would yield the same goal of achieving exponential decay of $\tau$. Constants that are too close to $1$, however, make the convergence $\tau \to 0$ more time consuming, and constants that are too close to $0$ may dramatically increase the number $\delta$ of iterations necessary at each step in order to ensure machine precision convergence.

Figure \ref{fig:e3_s_lay}, a larger version of which is Figure \ref{fig:e3_lay}, shows the total number of iterations of our implementation in experiments run in an Erdos-Renyi or stochastic block model setting for different values of the number of layers. The number of iterations needed for convergence, as a function of the number of layers, shows an overall decreasing trend along with smaller oscillations that appear to occur together with a widening of confidence intervals. It is likely that these irregularities are only transitory, as it becomes more evident in the SBM setting, and are due to some critical interaction between the number of nodes and the number of layers.

\begin{center}
\begin{figure}[ht]
    \begin{subfigure}[b]{0.48\linewidth}
  \includegraphics[width =\linewidth]{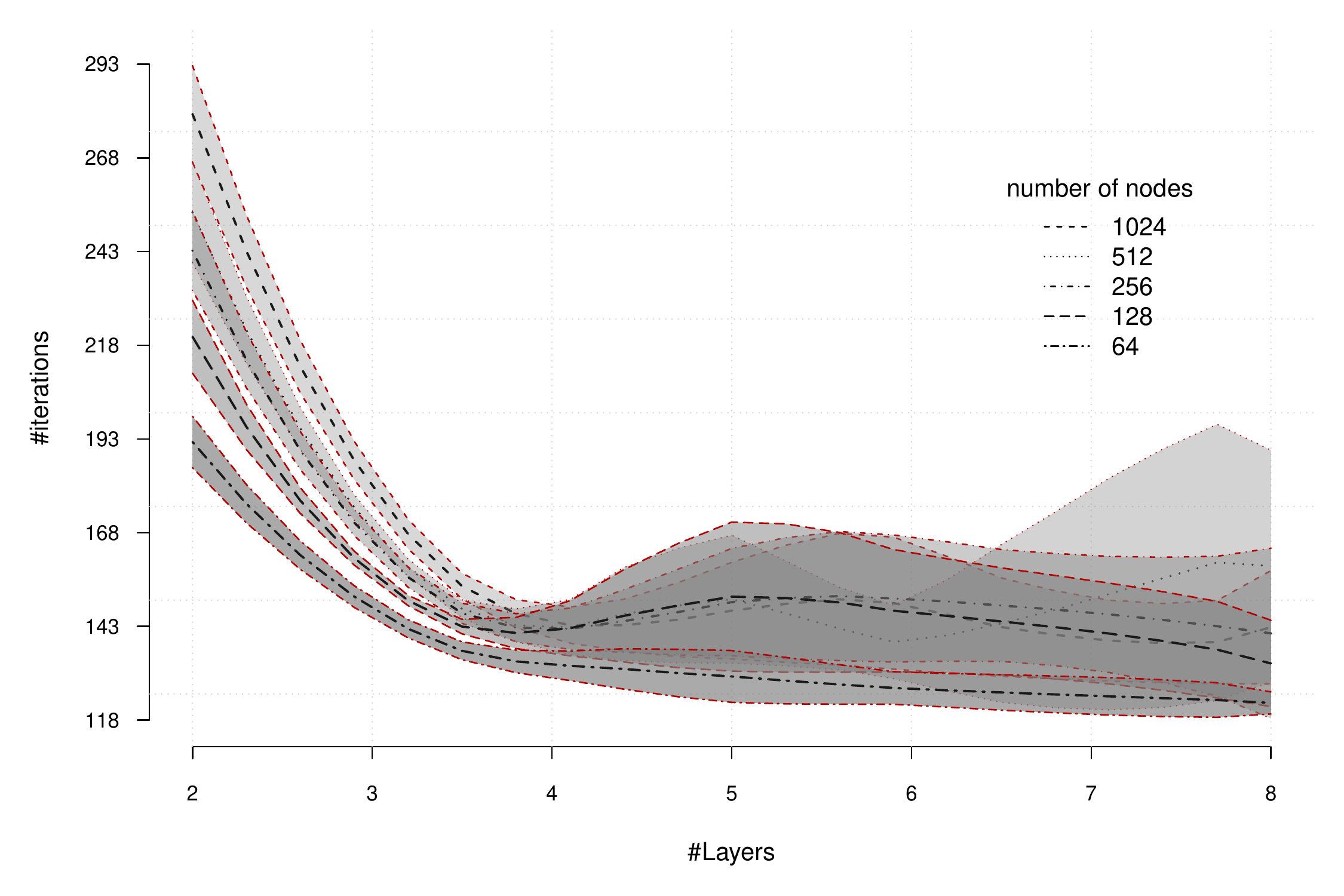}
  \caption{\label{fig:e3_s_erdos_lay}}
  \end{subfigure}
  \begin{subfigure}[b]{0.48\linewidth}
  \includegraphics[width =\linewidth]{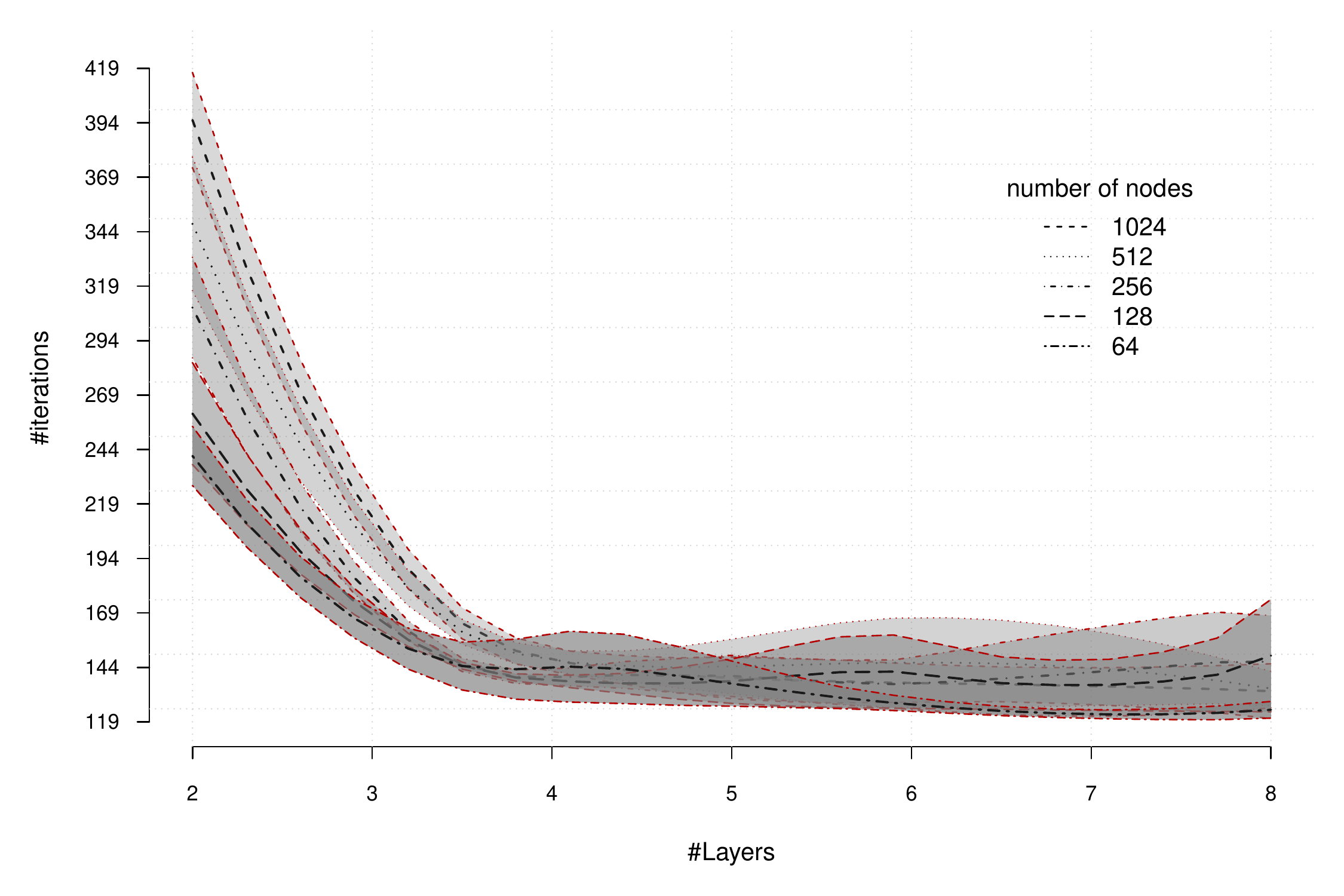}
  \caption{\label{fig:e3_s_sbm_lay}}
  \end{subfigure}  
  \caption{\label{fig:e3_s_lay} Total number of iterations performed in each Erdos-Renyi and SBM experiment for different values of the number of layers.} 
\end{figure}
\end{center}

\begin{table*}[h!t]
  \centering
    \begin{adjustbox}{width=0.98\linewidth}
	\begin{tabular}{r r r r r r r}
    $\mathbf{|V|}$&\multicolumn{1}{c}{\textbf{PageRank}}&\multicolumn{1}{c}{\textbf{PageRank-like}}&\multicolumn{1}{c}{\textbf{HITS}}&\multicolumn{1}{c}{\textbf{HITS-like}}&\multicolumn{1}{c}{\textbf{Versatile}}&\multicolumn{1}{c}{\textbf{Versatile-like}}\\
        \hline 
\textbf{64}&8,192&266,240&532,480&790,528&524,288&266,240\\
\textbf{128}&32,768&2,113,536&4,227,072&6,307,840&4,194,304&2,113,536\\
\textbf{256}&131,072&16,842,752&33,685,504&50,397,184&33,554,432&16,842,752\\
\textbf{512}&524,288&134,479,872&268,959,744&402,915,328&268,435,456&134,479,872\\
\textbf{1024}&2,097,152&1,074,790,400&2,149,580,800&3,222,274,048&2,147,483,648&1,074,790,400\\
\textbf{2048}&8,388,608&8,594,128,896&17,188,257,792&25,773,998,080&17,179,869,184&8,594,128,896\\
\textbf{4096}&33,554,432&68,736,253,952&137,472,507,904&206,175,207,424&137,438,953,472&68,736,253,952\\
     \hline
  \end{tabular}
  \end{adjustbox}  
  \caption{Simulation of theoretical running times for a multiplex network of two layers.\label{tab:theo-complexity}}  
  
  \subsubsection{}
\end{table*}

\subsection{Computational Complexity} \label{exp:batch4}
In this section we study the theoretical complexity for all the  methods  analyzed in Section \ref{exp:batch1}. 
Lets remember that $n$ is the number of  nodes of each network, and $\mathcal{L}$ is the number of layers that compose the multiplex $\mathcal{M}$. For simplicity, we analyze only the cost to compute one iteration of each method.

\textit{PageRank} \cite{ilprints422}:
The core of the method involves the multiplication of the PageRank matrix with the score vector of the previous iteration, which costs $n \cdot n$. The method is applied over all the layers of the multiplex, thus the total cost for  \textit{PageRank}  in a multiplex is of the order:
$$[\mathcal{L} \cdot \mathbf{(\, n \cdot n\, )}]$$

\textit{HITS} \cite{Kleinberg:1999}:
The HITS method involved the multiplication of the two matrices $A^T$ and $A$, which costs $n \cdot n \cdot n$. The obtained matrix must be multiplied by the score vector of the previous iteration $(\, n \cdot n\, )$. The method is applied over all the layers of the multiplex, and its total cost is of the order:
$$\mathcal{L} \cdot [ \mathbf{(\, n \cdot n \cdot n \,)} + (\, n \cdot n\, )]$$

\textit{Versatile} \cite{dedomenico2015ranking}:
Considering the formula $M_{i\alpha}^{j\beta} \Theta_{i\alpha} = \lambda_1 \Theta_{j\beta}$, presented in \cite{dedomenico2015ranking}, we can summarize the cost of the method by computing the multiplication between the tensor and the matrix at each iteration. This operation can be seen as the multiplication of a opportunely built matrix of size $n\times p$, where $p=(n\cdot\mathcal{L})$  encodes the tensor, with the score matrix of size  $n \times n$. The cost of this multiplication is $n \cdot ( \,p\, ) \cdot n)$. Then the total cost  is of the order:
$$[\,\mathbf{n} \cdot (\,\mathcal{L} \cdot \mathbf{n} \,)\cdot \mathbf{n}\,]$$

\textit{\MRF}:
\MRF multiplies the configuration matrices. The matrix multiplication is performed $(k-1)$ times, where $k$ is the length of the configuration. It follows that the computational cost of the core part of the algorithm is $((k-1)\, n \cdot n \cdot n \,)$. 
The obtained matrix must be multiplied by the score vector of the previous iteration, which costs $(\, n \cdot n\, )$. Thus, the total cost amounts to:
$$[(k-1) \cdot \mathbf{(\, n \cdot n \cdot n \,)} + (\, n \cdot n\, )]$$

When the configuration length $k$ grows,  the cost of  the computation will grow linearly in $k$, but a larger number of tailored rankings will be obtained. 

For  \textit{PageRank-like}, the length of the configuration $k$ depends on the number of layers, and the cost is $(\mathcal{L}-1) \cdot (\, n \cdot n \cdot n \,) + (\, n \cdot n\, )$. 
For \textit{HITS-like}, the length of the configuration $k$ is two times the number of layers $\mathcal{L}$, and the cost is $[(2 \cdot \mathcal{L})-1] \cdot (\, n \cdot n \cdot n \,) + (\, n \cdot n\, )$.  
The configuration of \textit{Versatile-like} is the same of \textit{PageRank-like}, and the cost is $(\mathcal{L}-1) \cdot (\, n \cdot n \cdot n \,) + (\, n \cdot n\, )$.

In summary, all the methods are asymptotically dominated by $\mathbf{(n \cdot n \cdot n)}$, with the exception of \textit{PageRank}, which is dominated by $\mathbf{(n \cdot n)}$. Table \ref{tab:theo-complexity} shows these trends, and we can see how simpler methods, such as \textit{PageRank}, are also computationally lighter in practice. We can also observe that \textit{Versatile-like}  is faster than  \textit{Versatile}, and \textit{HITS-like} has the same order of complexity as \textit{HITS}.


\section{Conclusion}
\label{conclusion}
In this paper, we presented a general methodology to iteratively compute the rankings $\rr_s$ associated to any possible configuration defined on a collection $\mathcal{A}$ of matrices of a multiplex network $\mathcal{M}$.
Our framework is flexible and can accommodate any kind of centrality measures, defined by a cyclic dependency. It can be tailored by the user to capture the semantics of the specific scenario at hand. We developed a theoretically sound iterative method, based on Perron-Frobenius theory, to compute the rankings. Our solution has guaranteed convergence. Our empirical evaluation confirms that our approach encompasses and generalizes a variety of standard measures. Its complexity is modulated by the number of tailored rankings one is interested in.
One aspect, that it is currently under investigation, is the possibility of modulating the importance that each component of the multiplex network has on the tailored rankings.

\bibliographystyle{ACM-Reference-Format}
\bibliography{dcp} 

\appendix

\begin{center}
\begin{figure*}[ht]
  \begin{subfigure}[b]{0.48\linewidth}
  \includegraphics[width =\linewidth]{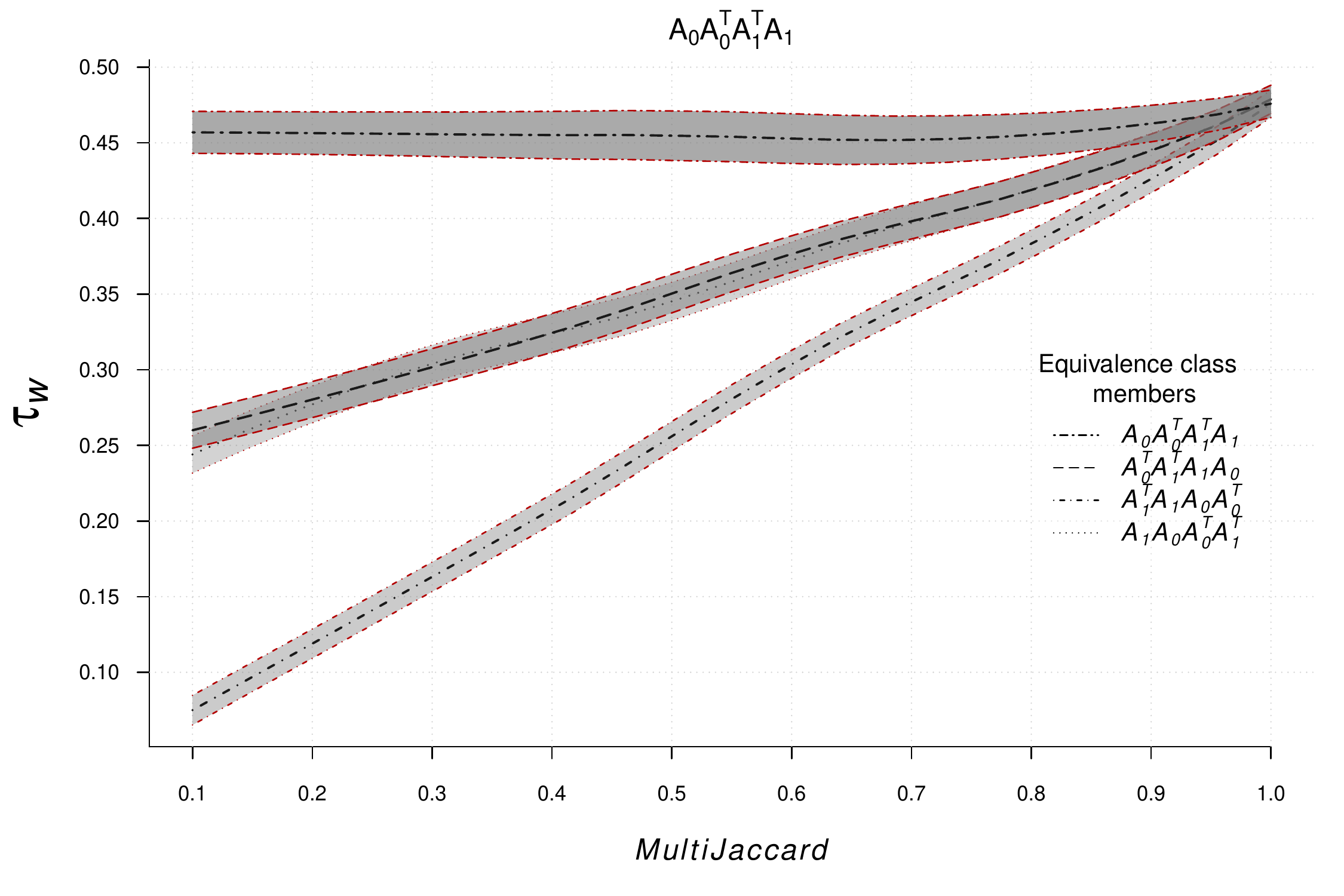}
  \caption{\label{fig:e2_erdos_c_A0A0TA1TA1}}
  \end{subfigure}
  \begin{subfigure}[b]{0.48\linewidth}
  \includegraphics[width =\linewidth]{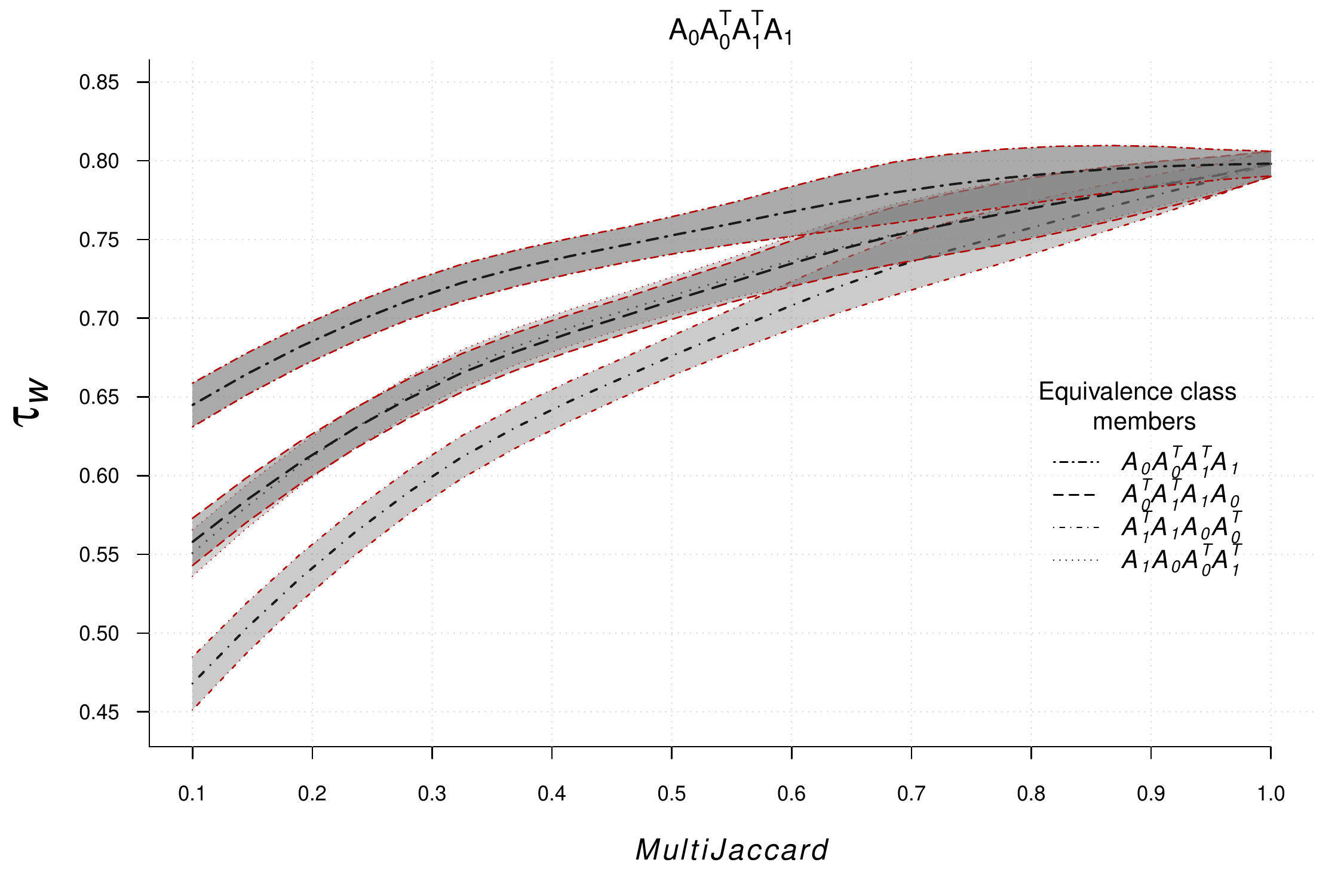}
  \caption{\label{fig:e2_sbm_c_A0A0TA1TA1}}
  \end{subfigure}
  \\
  \begin{subfigure}[b]{0.48\linewidth}
  \includegraphics[width =\linewidth]{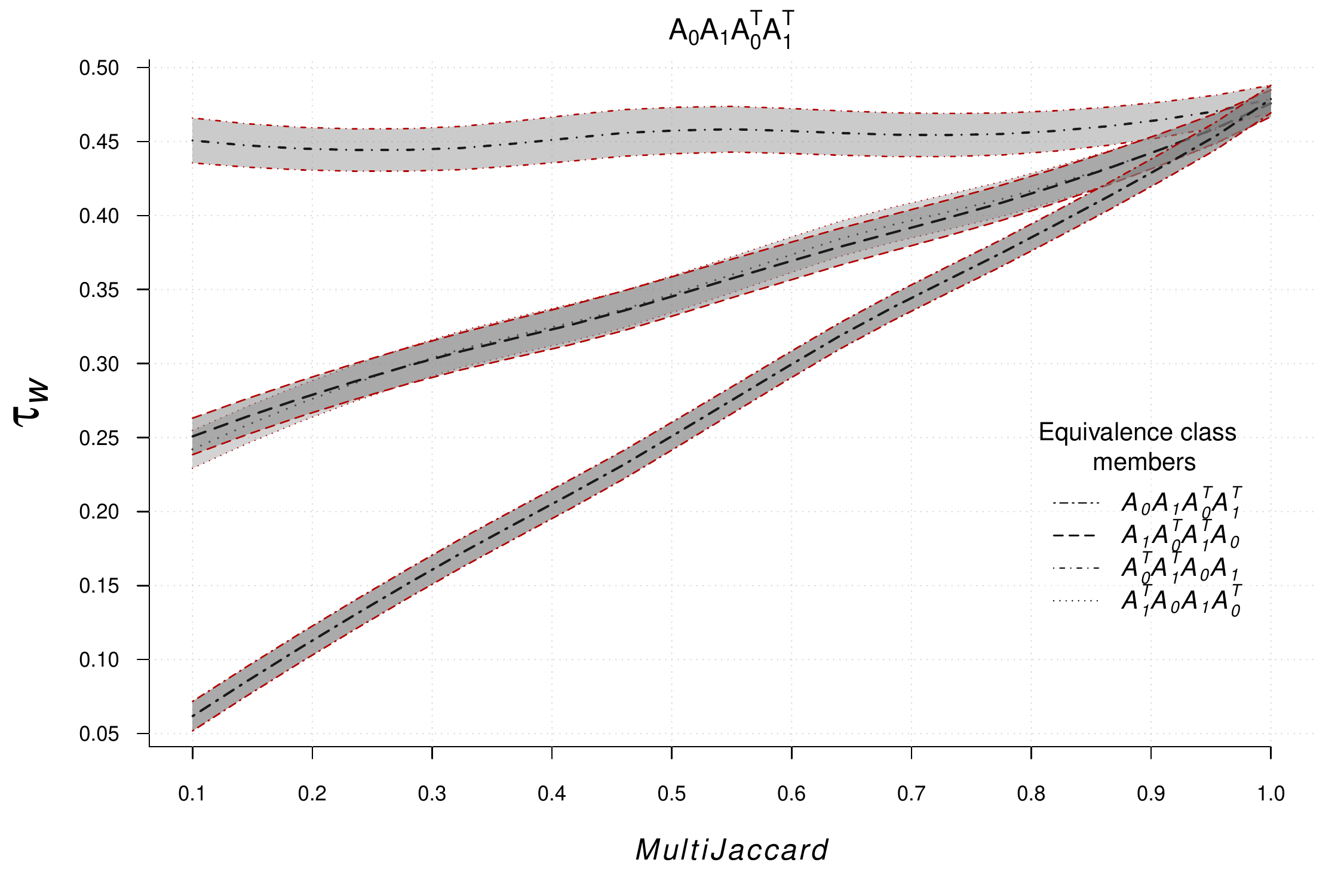}
  \caption{\label{fig:e2_erdos_c_A0A1A0TA1T}}
  \end{subfigure}
  \begin{subfigure}[b]{0.48\linewidth}
  \includegraphics[width =\linewidth]{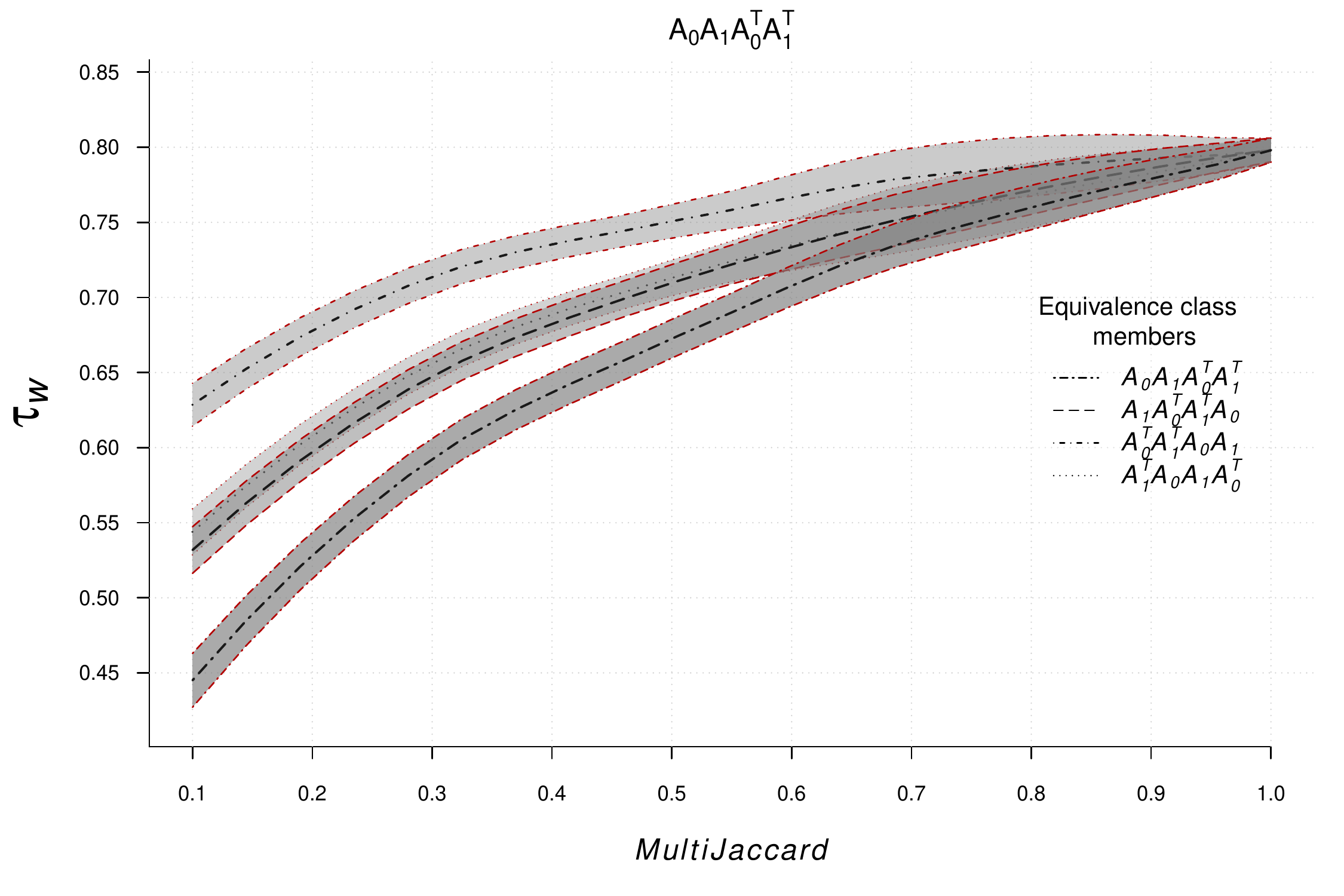}
  \caption{\label{fig:e2_sbm_c_A0A1A0TA1T}}
  \end{subfigure}
  \\
  \begin{subfigure}[b]{0.48\linewidth}
  \includegraphics[width =\linewidth]{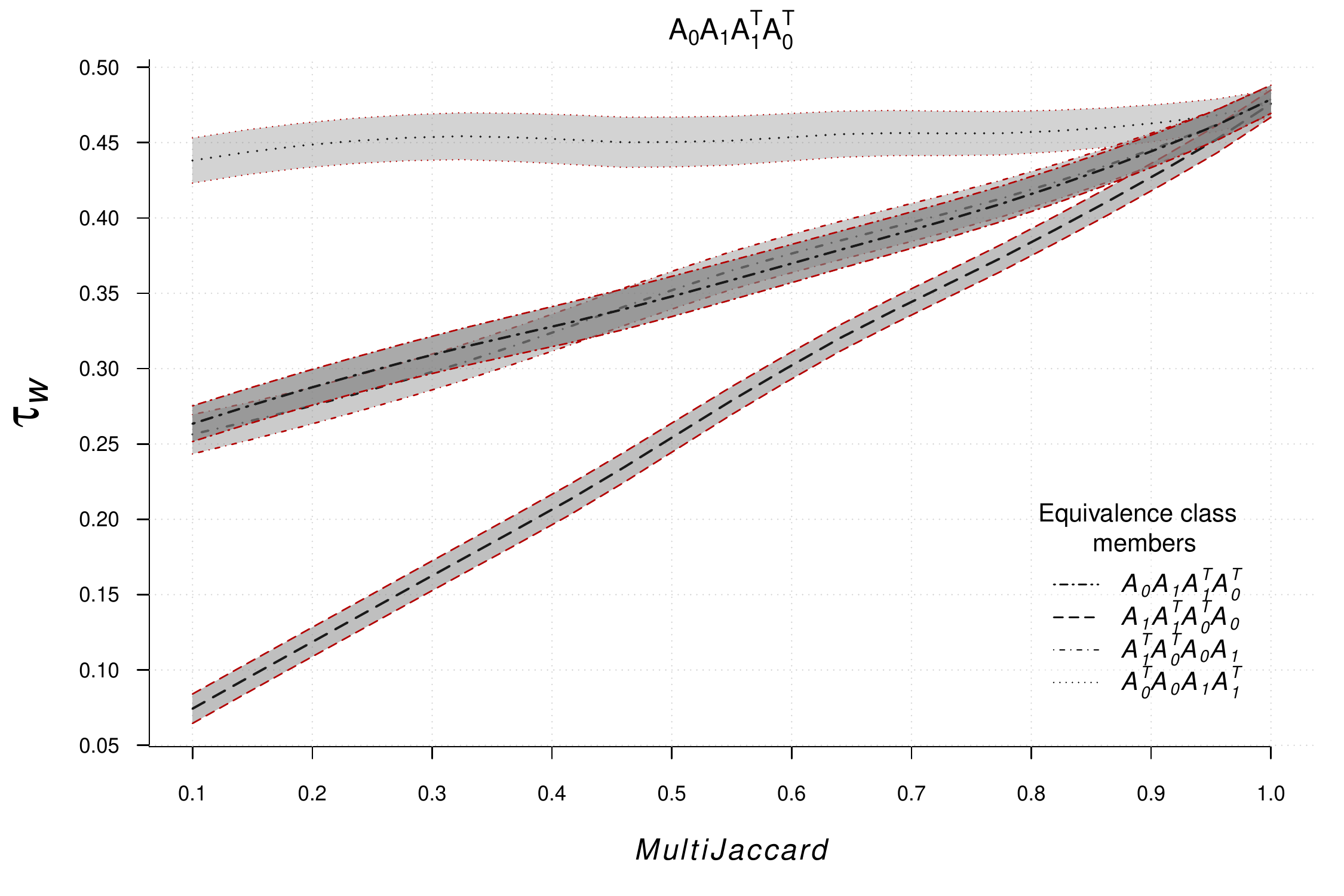}
  \caption{\label{fig:e2_erdos_c_A0A1A1TA0T}}
  \end{subfigure}
  \begin{subfigure}[b]{0.48\linewidth}
  \includegraphics[width =\linewidth]{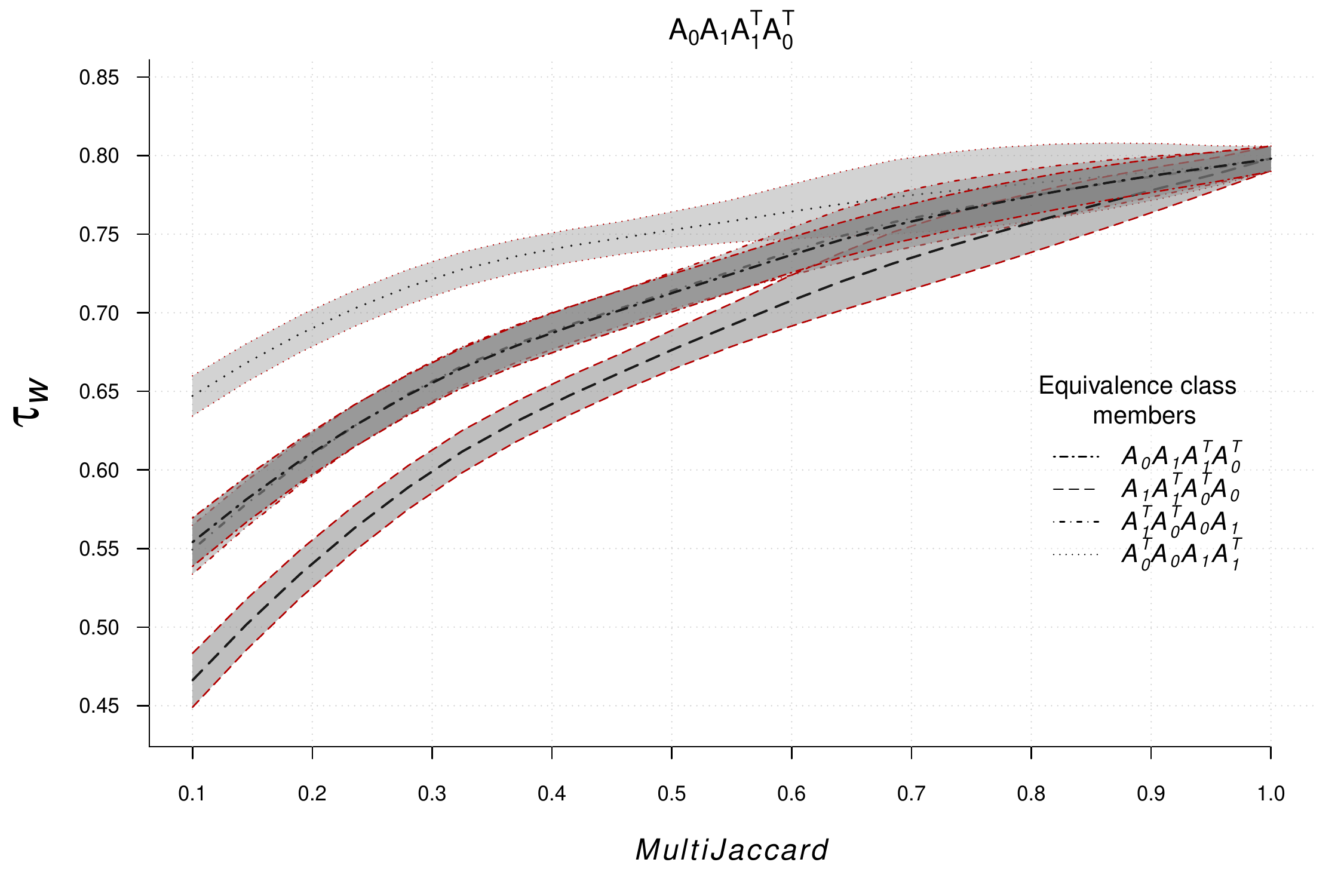}
  \caption{\label{fig:e2_sbm_c_A0A1A1TA0T}}
  \end{subfigure}
  
  \caption{Similarity between rankings produced by \textit{HITS} and \MRF   for a multiplex network of two layers considering all the members of the equivalence class. (a) and (b): $A_0A_0^TA_1^TA_1$; (c) and (d): $A_0A_1A_0^TA_1^T$; (e) and (f): $A_0A_1A_1^TA_0^T$.
 (a), (c), and (e) use an Erdos-Renyi graph as generator; (b), (d), and (f) use stochastic block model graph as generator.}
  \label{fig:e2_conf_A}
\end{figure*}
\end{center}

\begin{center}
\begin{figure*}[ht]
  \begin{subfigure}[b]{0.48\linewidth}
  \includegraphics[width =\linewidth]{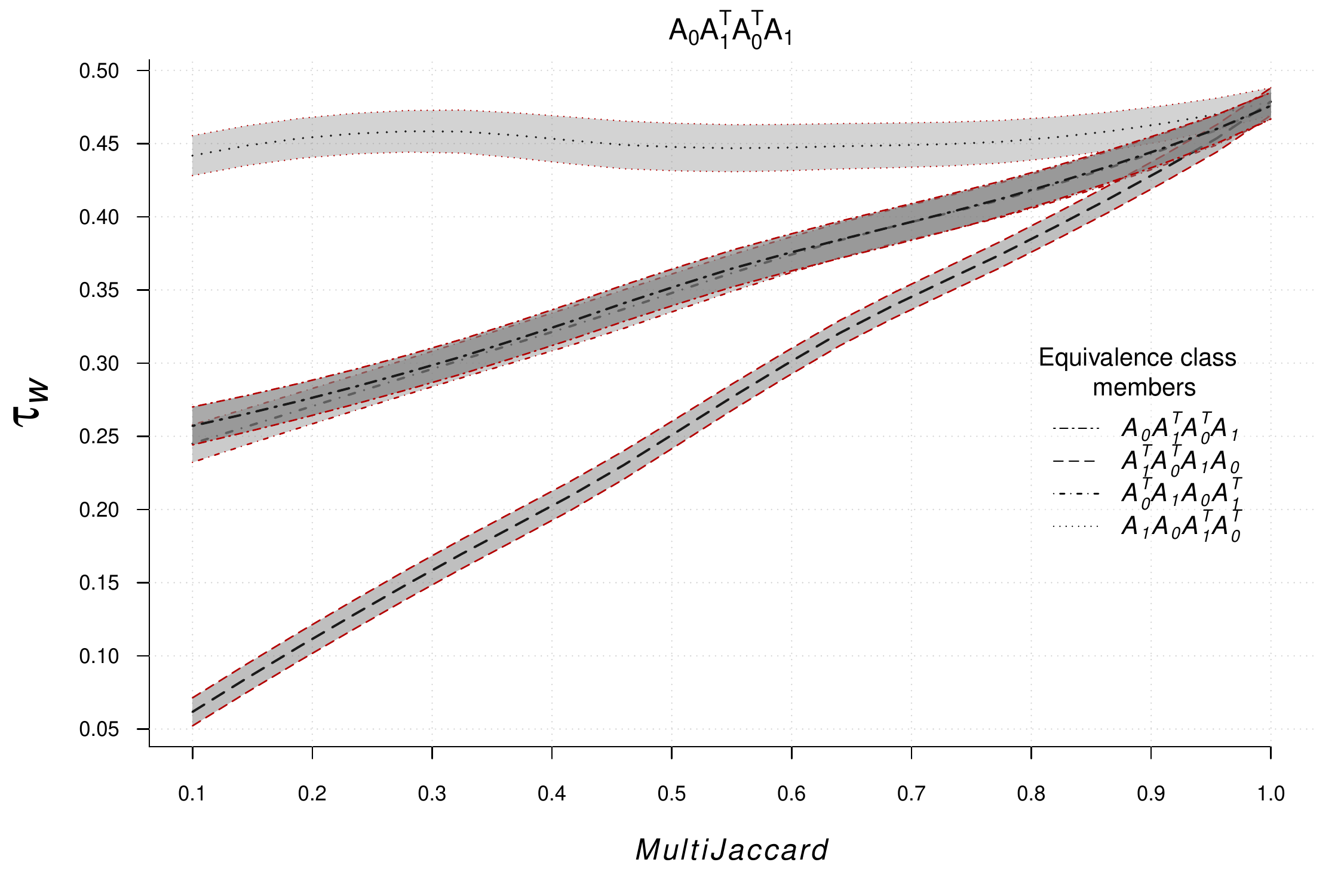}
  \caption{\label{fig:e2_erdos_c_A0A1TA0TA1}}
  \end{subfigure}
  \begin{subfigure}[b]{0.48\linewidth}
  \includegraphics[width =\linewidth]{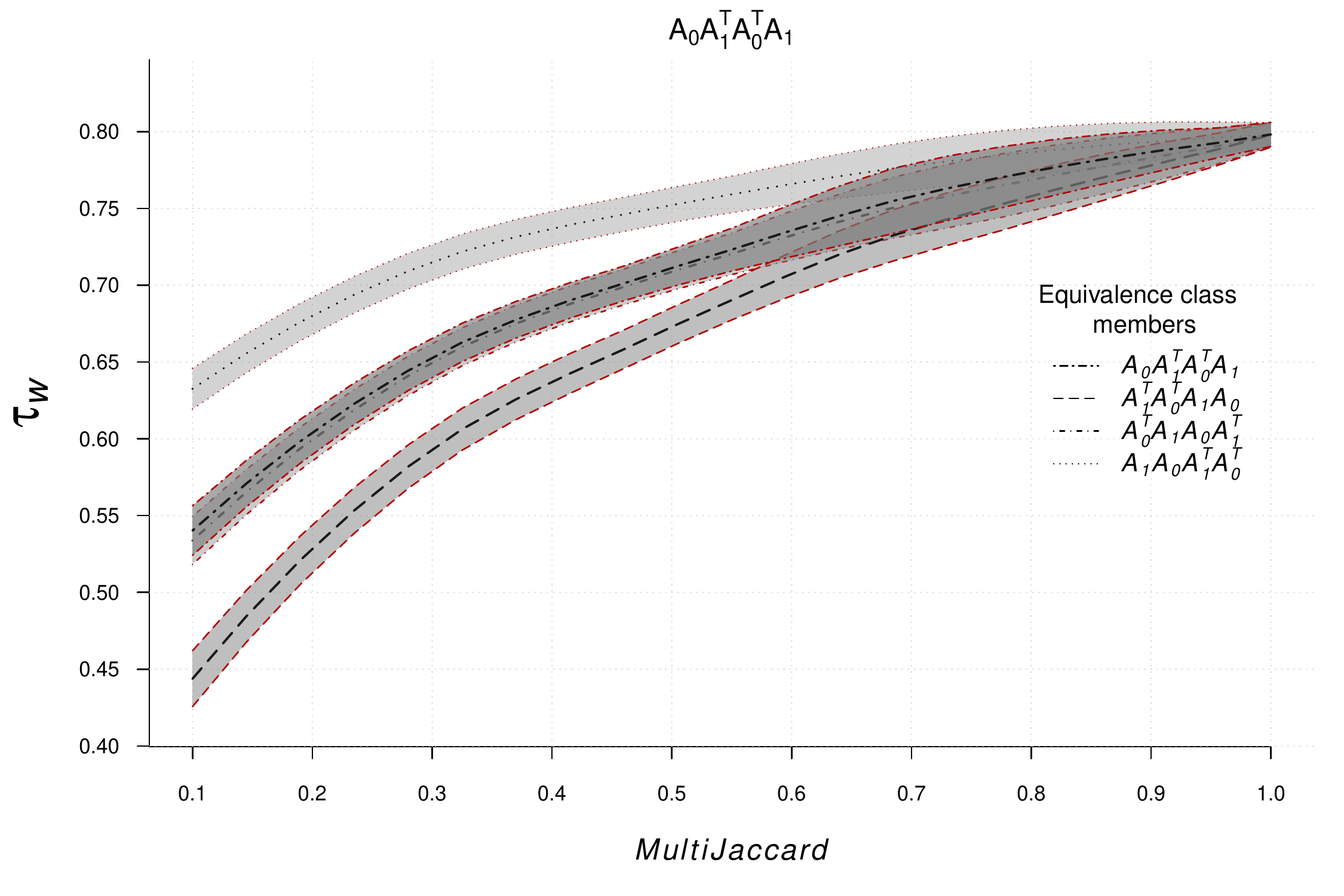}
  \caption{\label{fig:e2_sbm_c_A0A1TA0TA1}}
  \end{subfigure}
  \\
  \begin{subfigure}[b]{0.48\linewidth}
  \includegraphics[width =\linewidth]{exp2/e2_erdos_config_A_0A_1_TA_1A_0_T.pdf}
  \caption{\label{fig:e2_erdos_c_A0A1TA1A0T}}
  \end{subfigure}
  \begin{subfigure}[b]{0.48\linewidth}
  \includegraphics[width =\linewidth]{exp2/e2_sbm_config_A_0A_1_TA_1A_0_T.pdf}
  \caption{\label{fig:e2_sbm_c_A0A1TA1A0T}}
  \end{subfigure}
  \\
  \begin{subfigure}[b]{0.48\linewidth}
  \includegraphics[width =\linewidth]{exp2/e2_erdos_config_A_0A_0_TA_1A_1_T.pdf}
  \caption{\label{fig:e2_erdos_c_A0A0TA1A1T}}
  \end{subfigure}
  \begin{subfigure}[b]{0.48\linewidth}
  \includegraphics[width =\linewidth]{exp2/e2_sbm_config_A_0A_0_TA_1A_1_T.pdf}
  \caption{\label{fig:e2_sbm_c_A0A0TA1A1T}}
  \end{subfigure}
 
 \caption{Similarity between rankings produced by \textit{HITS} and \MRF  for a multiplex network of two layers considering all the members of the equivalence class. (a) and (b): $A_0A_1^TA_0^TA_1$; (c) and (d): $A_0A_1^TA_1A_0^T$; (e) and (f): $A_0A_0^TA_1A_1^T$. (a), (c), and (e) use an Erdos-Renyi graph as generator; (b), (d), and (f) use stochastic block model graph as generator.   \label{fig:e2_conf_B}}
\end{figure*}
\end{center}

\begin{center}
\begin{figure*}[ht]
  \begin{subfigure}[b]{0.38\linewidth}
  \includegraphics[width =\linewidth]{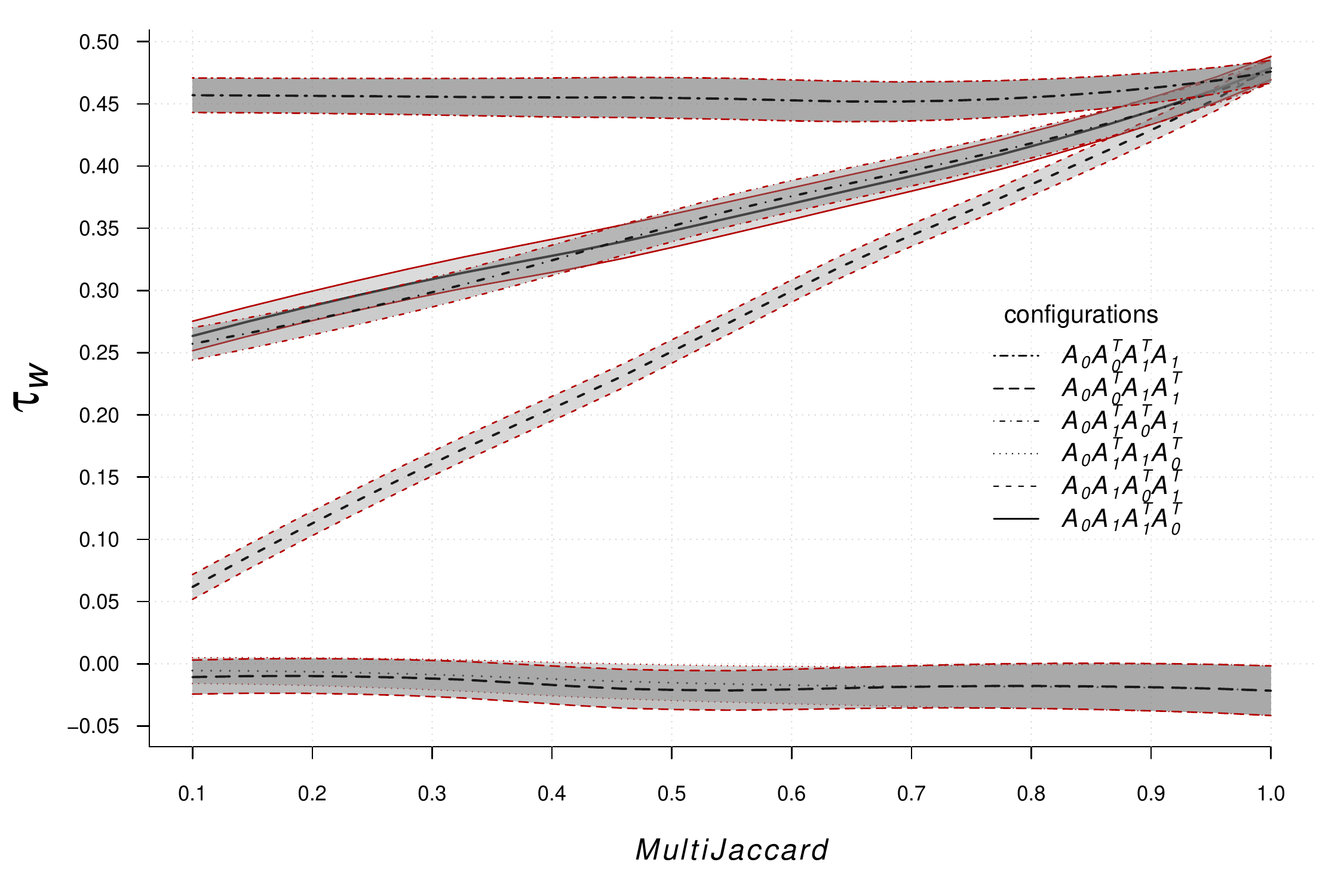}
  \caption{\label{fig:e2_erdos_shift_0}}
  \end{subfigure}
  \begin{subfigure}[b]{0.38\linewidth}
  \includegraphics[width =\linewidth]{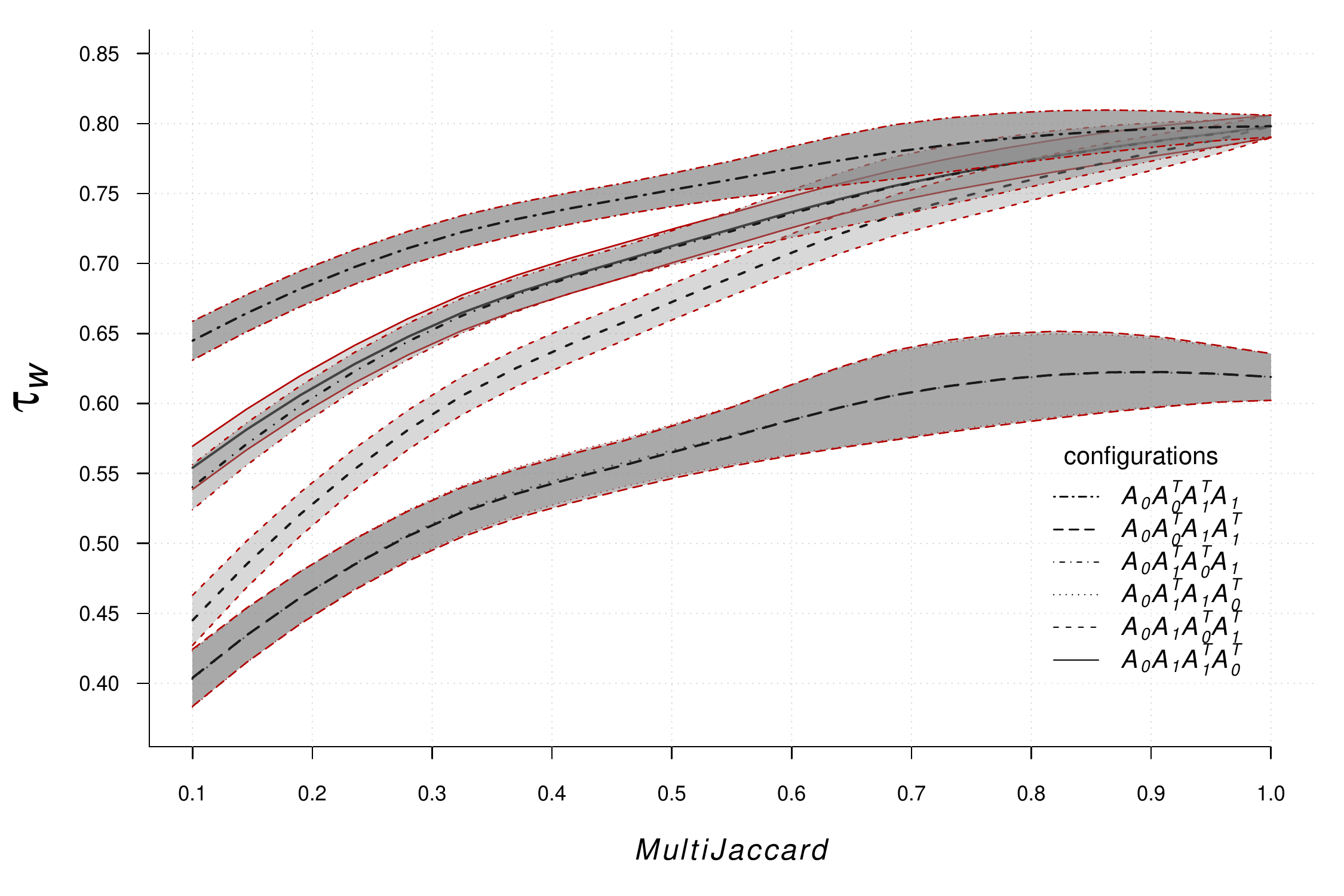}
  \caption{\label{fig:e2_sbm_shift_0}}
  \end{subfigure}
  \\
  \begin{subfigure}[b]{0.38\linewidth}
  \includegraphics[width =\linewidth]{exp2/e2_erdos_shift_1.pdf}
  \caption{\label{fig:e2_erdos_shift_1}}
  \end{subfigure}
  \begin{subfigure}[b]{0.38\linewidth}
  \includegraphics[width =\linewidth]{exp2/e2_sbm_shift_1.pdf}
  \caption{\label{fig:e2_sbm_shift_1}}
  \end{subfigure}
  \\
  \begin{subfigure}[b]{0.38\linewidth}
  \includegraphics[width =\linewidth]{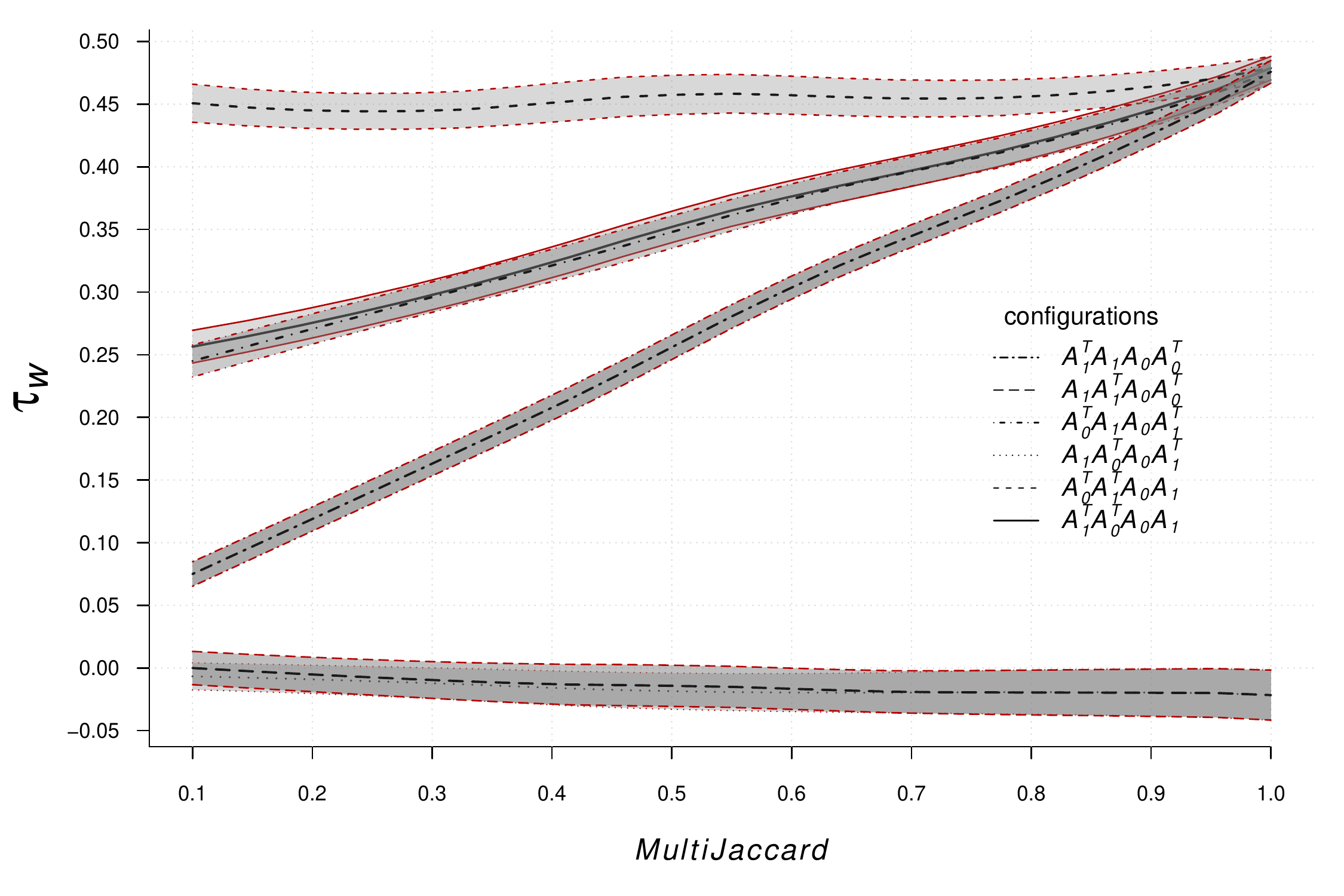}
  \caption{\label{fig:e2_erdos_shift_2}}
  \end{subfigure}
  \begin{subfigure}[b]{0.38\linewidth}
  \includegraphics[width =\linewidth]{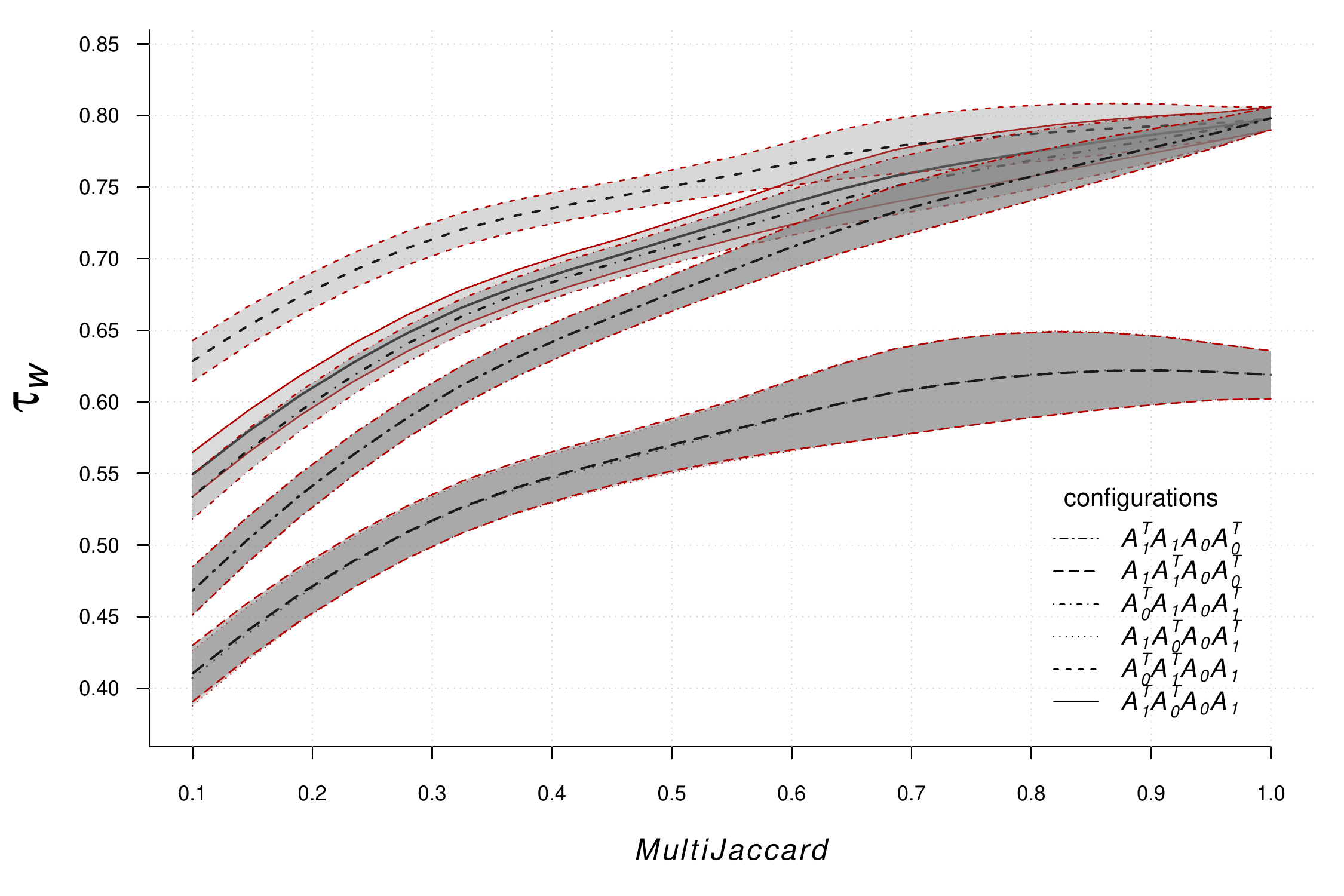}
  \caption{\label{fig:e2_sbm_shift_2}}
  \end{subfigure}
  \\
  \begin{subfigure}[b]{0.38\linewidth}
  \includegraphics[width =\linewidth]{exp2/e2_erdos_shift_3.pdf}
  \caption{\label{fig:e2_erdos_shift_3}}
  \end{subfigure}
  \begin{subfigure}[b]{0.38\linewidth}
  \includegraphics[width =\linewidth]{exp2/e2_sbm_shift_3.pdf}
  \caption{\label{fig:e2_sbm_shift_3}}
  \end{subfigure}

 \caption{Similarity between rankings produced by \textit{HITS} and \MRF for a multiplex network of two layers using all possible configurations. (a) and (b): $shift_0$; (c) and (d):  $shift_1$; (e) and (f):  $shift_2$; (g) and (h):  $shift_3$. (a) and (c) and (e) and (g) use  an Erdos-Renyi graph as generator; (b) and (d) and (f) and (h) use  a stochastic block model graph as generator.}
  \label{fig:e2_shift}
\end{figure*}
\end{center}

\begin{center}
\begin{figure*}[ht]
    \begin{subfigure}[b]{0.45\linewidth}
  \includegraphics[width =\linewidth]{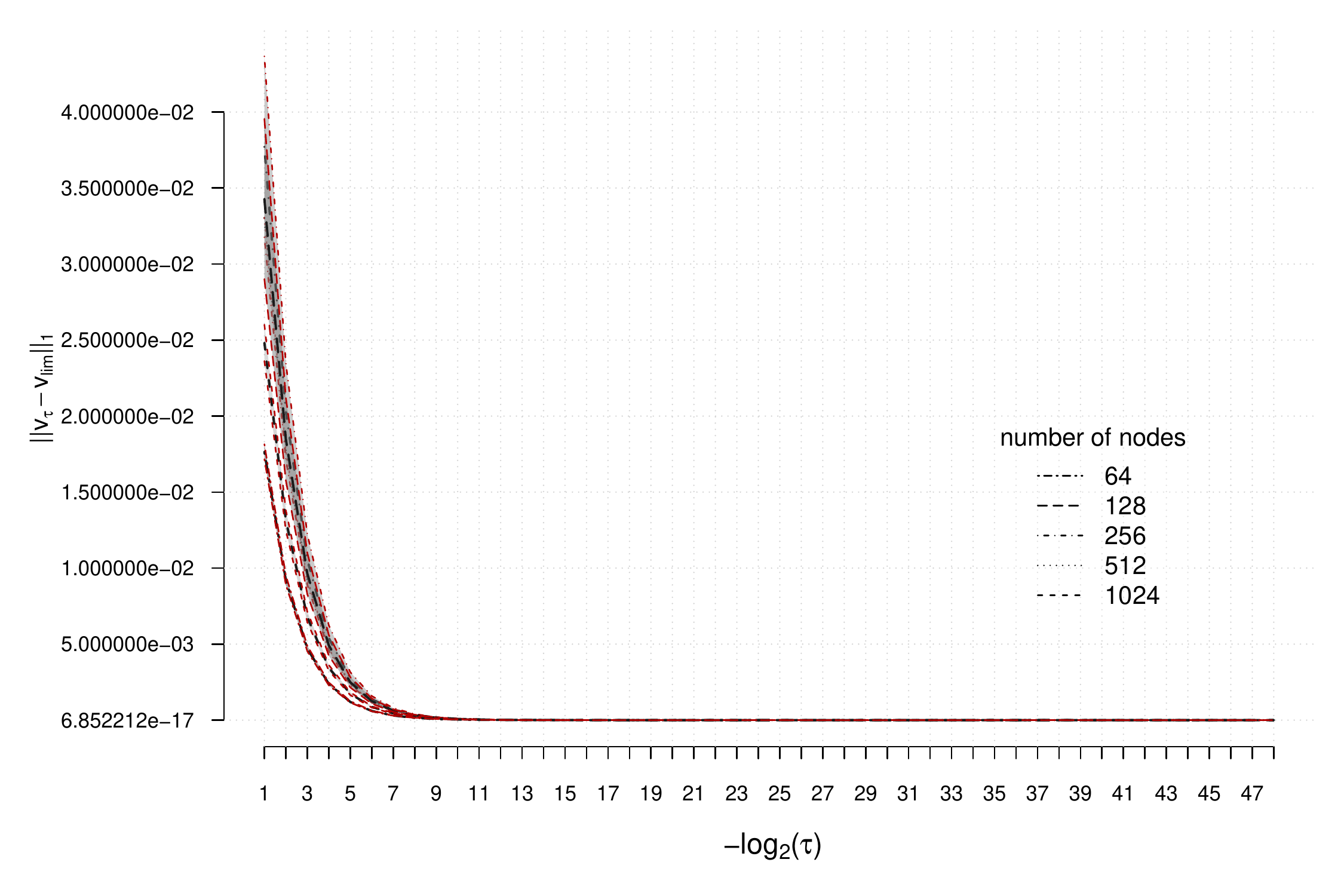}
  \caption{\label{fig:e3_erdos_conv}}
  \end{subfigure}
  \begin{subfigure}[b]{0.45\linewidth}
  \includegraphics[width =\linewidth]{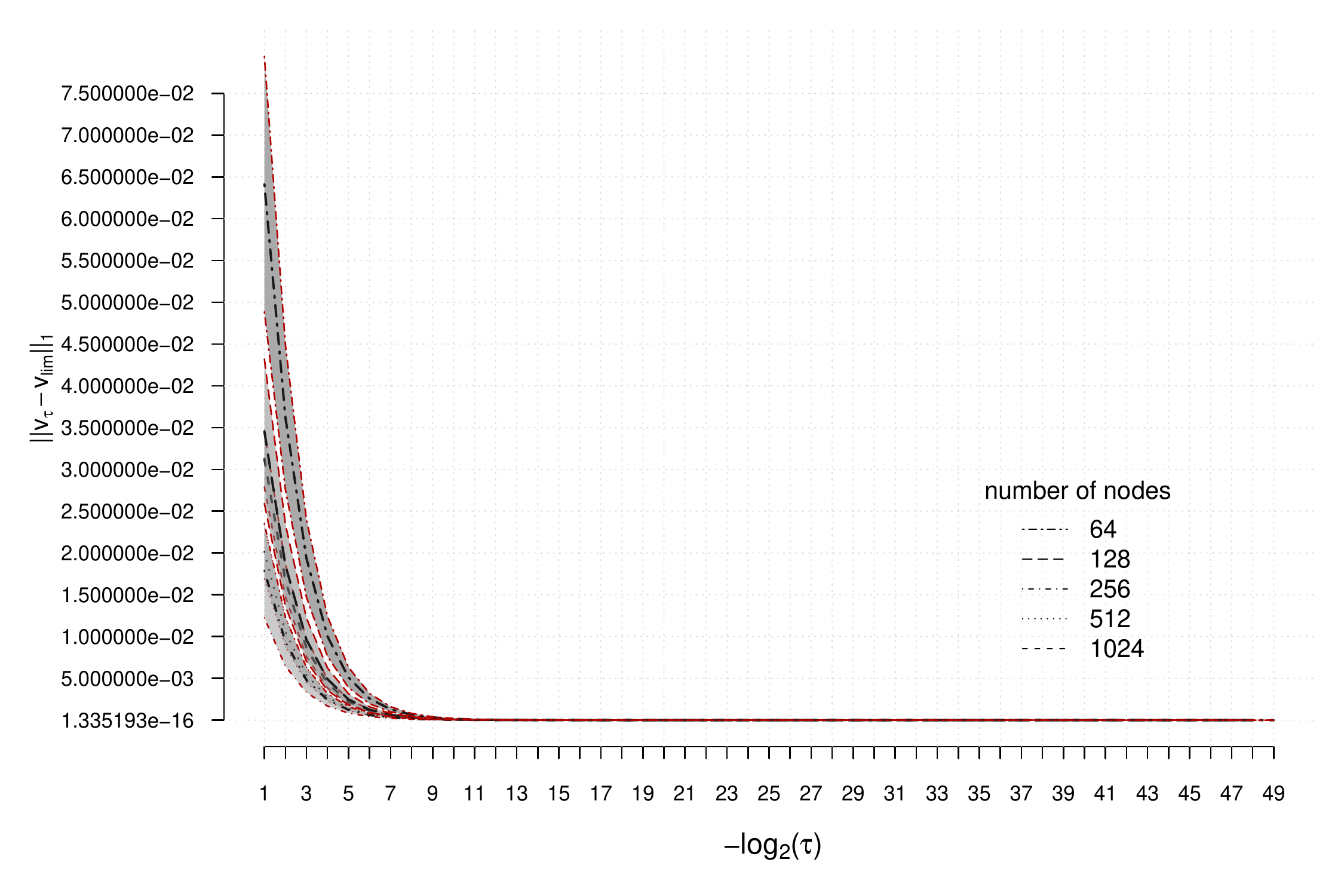}
  \caption{\label{fig:e3_sbm_conv}}
  \end{subfigure}  
  \caption{\label{fig:e3_conv} Convergence to the limit ranking vector using an Erdos-Renyi graph and a stochastic block model graph. The plots show $||v(\tau)-v(0)||_1$ against the number of times $\tau$ has been halved.}  
\end{figure*}
\end{center}

\begin{center}
\begin{figure*}[ht]
    \begin{subfigure}[b]{0.45\linewidth}
  \includegraphics[width =\linewidth]{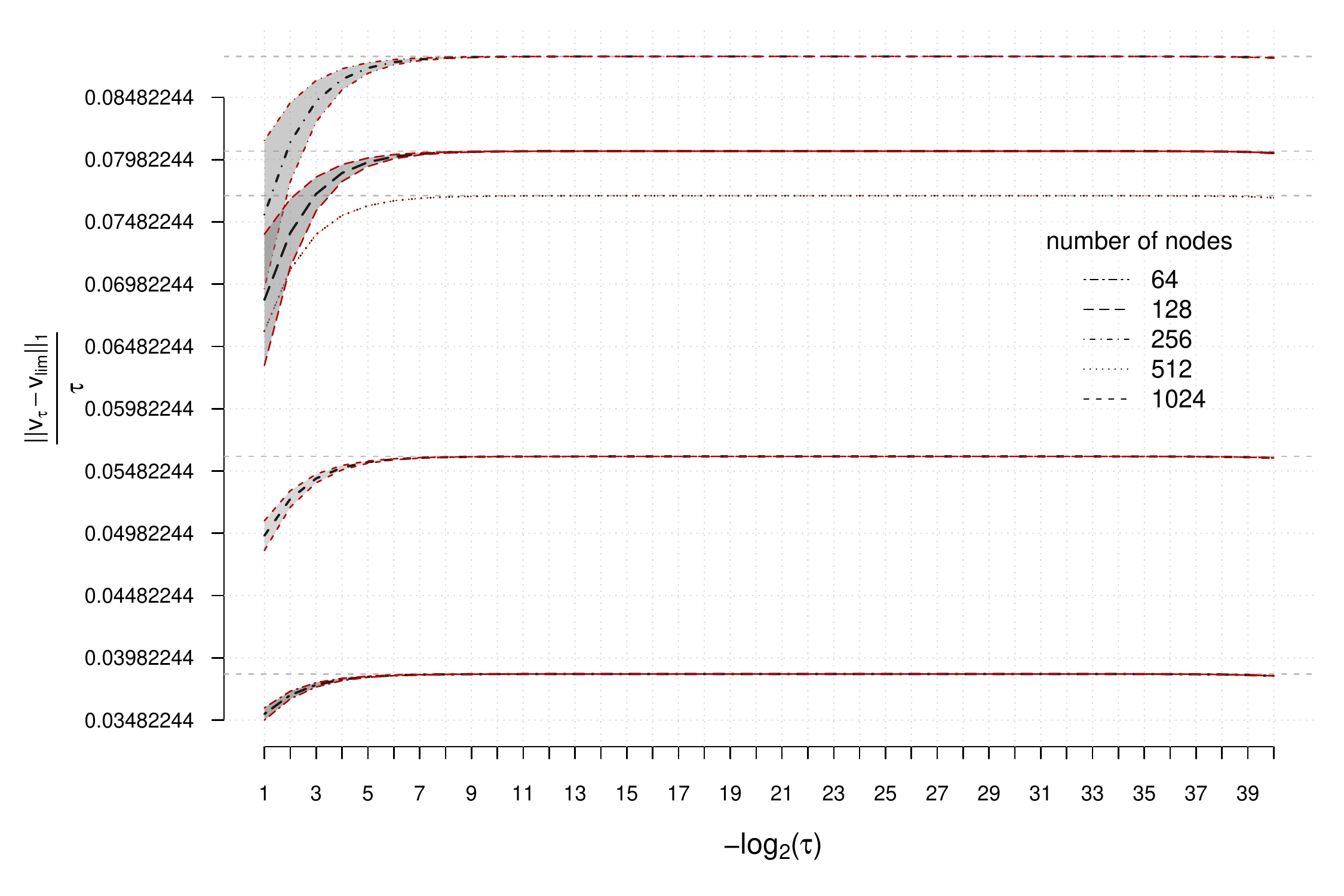}
  \caption{\label{fig:e3_erdos_const}}
  \end{subfigure}
  \begin{subfigure}[b]{0.45\linewidth}
  \includegraphics[width =\linewidth]{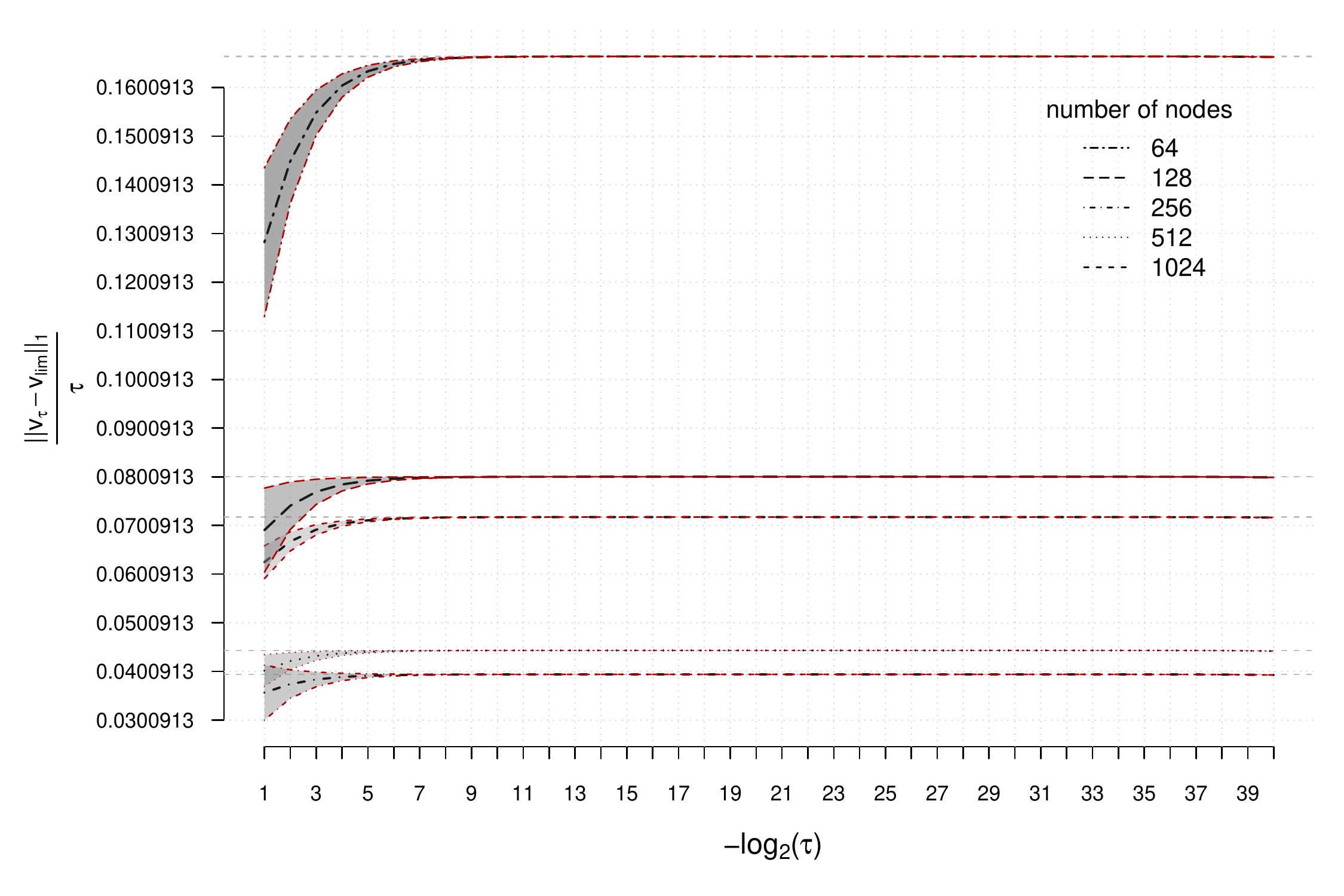}
  \caption{\label{fig:e3_sbm_const}}
  \end{subfigure}  
  \caption{\label{fig:e3_const} The quantity $||v(\tau) - v(0)||_1/\tau$ is plotted against $-\log_2\tau$, thus highlighting the (asymptotical) linear behavior of $||v(\tau)- v(0)||_1$ as a function of $\tau$.}  
\end{figure*}
\end{center}

\begin{center}
\begin{figure*}[ht]
    \begin{subfigure}[b]{0.45\linewidth}
  \includegraphics[width =\linewidth]{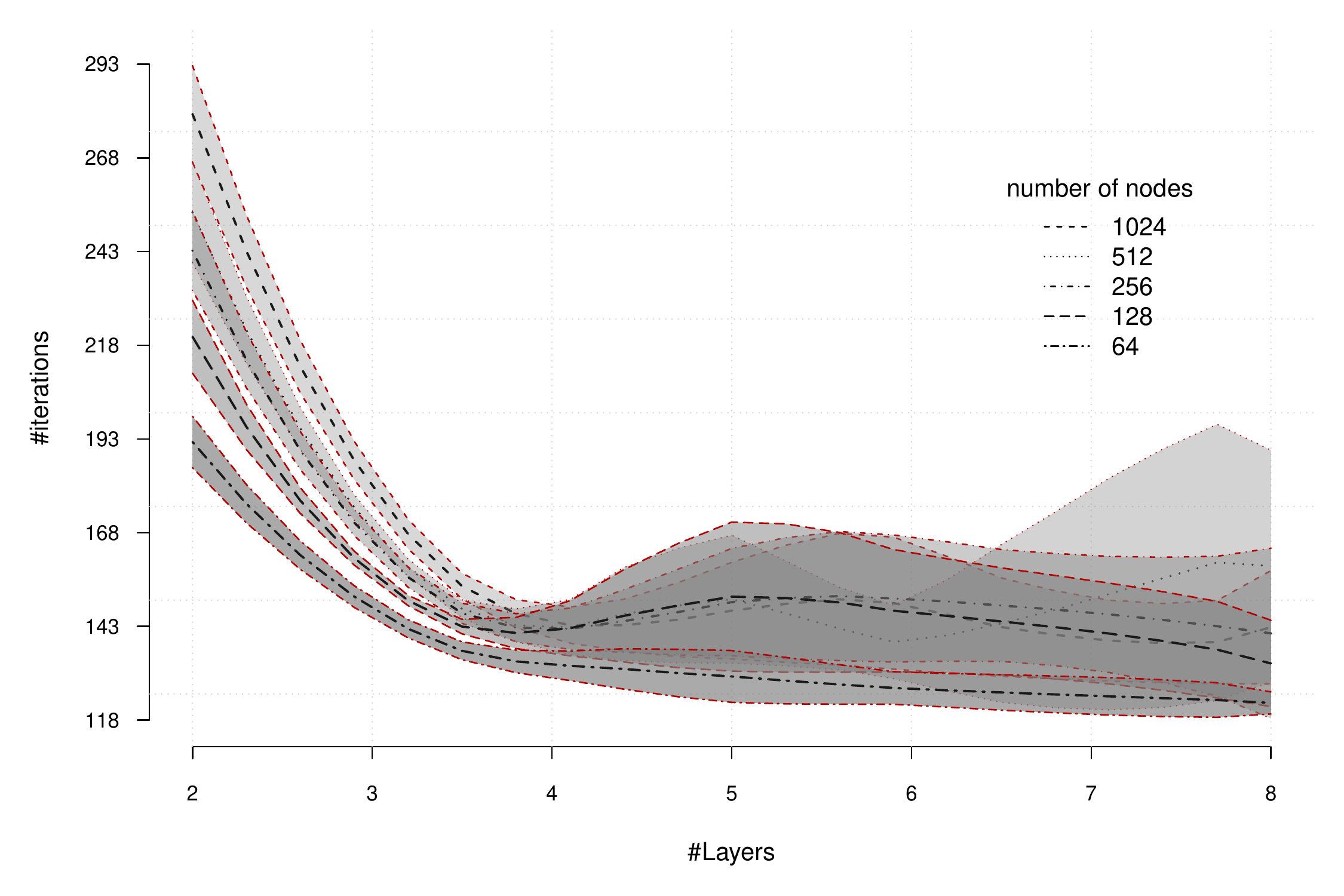}
  \caption{\label{fig:e3_erdos_lay}}
  \end{subfigure}
  \begin{subfigure}[b]{0.45\linewidth}
  \includegraphics[width =\linewidth]{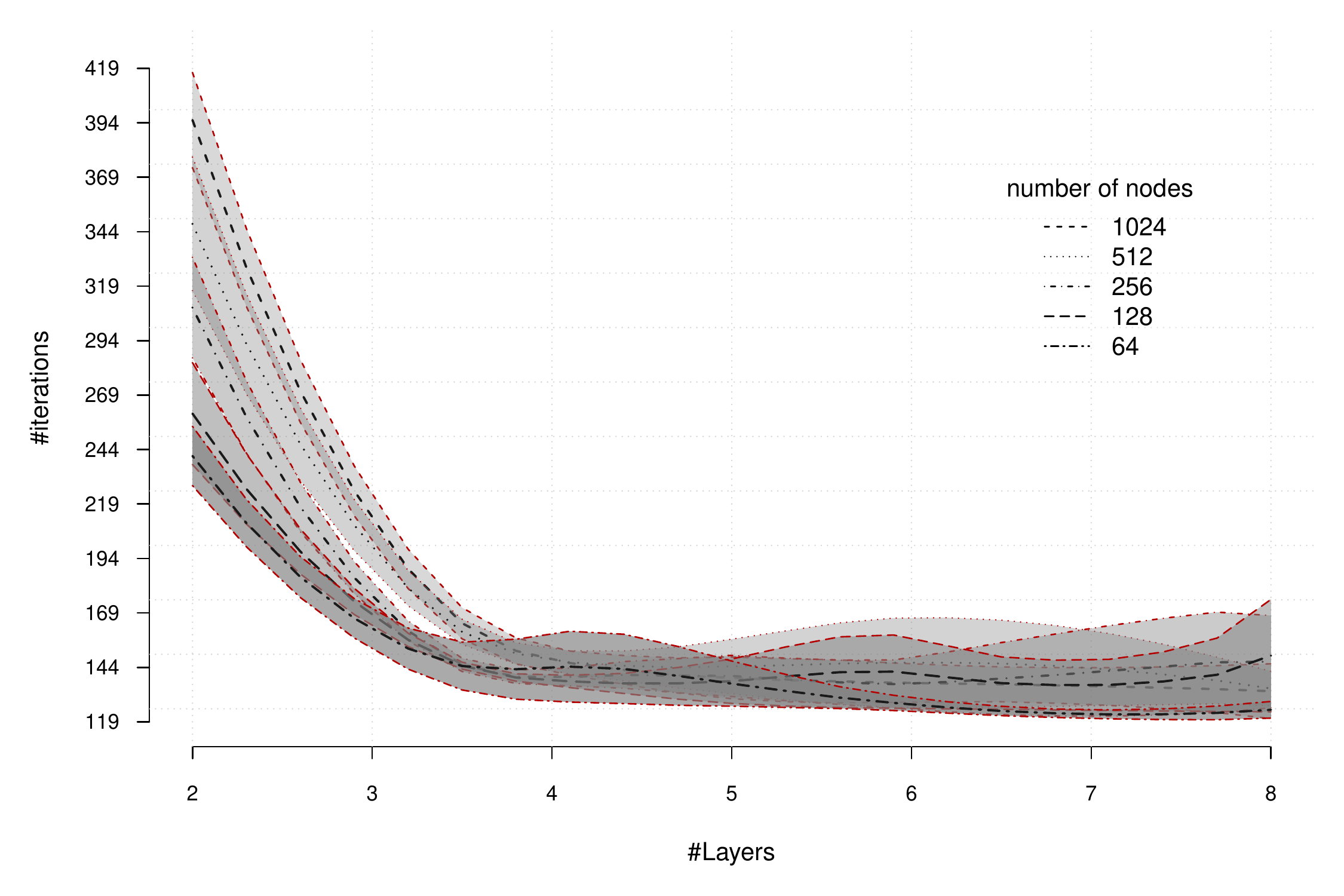}
  \caption{\label{fig:e3_sbm_lay}}
  \end{subfigure}  
  \caption{\label{fig:e3_lay} Total number of iterations performed in each Erdos-Renyi and SBM experiment for different values of the number of layers.}  
\end{figure*}
\end{center}

\end{document}